\documentclass[11pt,letterpaper]{article}
\usepackage{times}
\usepackage{fullpage}
\usepackage{amsmath}
\usepackage{amssymb, amsthm, latexsym, mathtools, hyperref, mathrsfs,float}
\usepackage{wrapfig}
\usepackage{soul,xcolor}
\usepackage[vlined,ruled,linesnumbered]{algorithm2e}
\SetEndCharOfAlgoLine{}

\usepackage{mdframed} 
\usepackage{xfrac}
\usepackage[inline]{enumitem}   
\DeclareMathOperator*{\argmax}{arg\,max}
\DeclareMathOperator*{\argmin}{arg\,min}

\theoremstyle{definition}
\newtheorem{theorem}{Theorem}

\newtheorem{lemma}{Lemma}[section]

\newtheorem{proposition}[lemma]{Proposition}
\newtheorem{definition}[lemma]{Definition}
\newtheorem{claim}{Claim}[section]

\definecolor{light-gray}{gray}{0.925}
\mdfdefinestyle{theoremstyle}{%
linecolor=light-gray,%
linewidth=2pt,%
innertopmargin=0,
backgroundcolor=light-gray
}
\newmdtheoremenv[
linewidth=2pt,
backgroundcolor=light-gray,
linecolor=light-gray,
innertopmargin=0
]
{result}{Theorem}

\renewcommand{\paragraph}[1]{
    \medskip\noindent{\bf #1}
    }
    
\newcommand{\N}{\mathbb{N}}
\newcommand{\poly}{\mathsf{poly}}
\newcommand\floor[1]{\lfloor#1\rfloor}
\newcommand{\junk}[1]{}

\newcommand{\Placed}{\mathsf{Placed}}
\newcommand{\iter}[1]{^{(#1)}}
\newcommand{\overlap}{\eta}
\newcommand{\step}{\varsigma}
\newcommand{\superstep}[1]{^{(\step#1)}}

\newcommand{\chain}{\mathscr{C}}
\newcommand{\link}[1]{c_{#1}}

\newcommand{\LP}[1]{LP~(#1)}
\newcommand{\start}[1]{S_{#1}}
\newcommand{\comp}[1]{C_{#1}}
\newcommand{\LPmakespan}{C}

\newcommand{\phasevariable}[2]{y_{#1,#2}}
\newcommand{\machinephase}[3]{z_{#1,#2,#3}}
\newcommand{\jobmachine}[2]{x_{#1,#2}}

\newcommand{\groupspeed}[1]{\overline{s}_{#1}}
\newcommand{\groupsize}[1]{m_{#1}}
\newcommand{\group}[1]{\Gamma_{#1}}
\newcommand{\assignfunc}{\kappa} 
\newcommand{\med}[1]{\mathsf{median}(#1)}
\newcommand{\jobtogroup}[1]{\kappa(#1)}
\newcommand{\grouptojobs}[1]{\kappa^{\scriptscriptstyle-1}(#1)}
\newcommand{\fastergroupstojobs}[1]{\kappa^{\scriptscriptstyle-1}_{\ge}(#1)}

\newcommand{\delay}{\rho}
\newcommand{\optmakespan}{C^*}
\newcommand{\speed}[1]{s_{#1}}
\newcommand{\proc}[1]{p_{#1}}
\newcommand{\setproc}[1]{\proc{}(#1)}
\newcommand{\predex}[1]{A_{#1}}

\newcommand{\band}[1]{B_{#1}}

\newcommand{\rhophase}{\rho\text{-phase}}

\newif\ifnotes\notesfalse

\ifnotes

\newcommand{\david}[1]{\textcolor{cyan}{\textrm{[David says: #1]}}}
\newcommand{\zoya}[1]{\textcolor{blue}{\textrm{[Zoya says: #1]}}}
\newcommand{\biswa}[1]{\textcolor{olive}{\textrm{[biswa says: #1]}}}
\newcommand{\rajmohan}[1]{\textcolor{red}{\textrm{[Rajmohan says: #1]}}}
\newcommand{\aravind}[1]{\textcolor{orange}{\textrm{[Aravindan says: #1]}}}

\else

\newcommand{\david}[1]{}
\newcommand{\zoya}[1]{}
\newcommand{\biswa}[1]{}
\newcommand{\rajmohan}[1]{}
\newcommand{\aravind}[1]{}

\fi

\begin{document}

\title{Scheduling Precedence-Constrained Jobs on Related Machines with Communication Delay}
\author{
    Biswaroop Maiti\thanks{Northeastern University, Boston, MA, USA. Email: \texttt{m.biswaroop@gmail.com}}
\and    
    Rajmohan Rajaraman\thanks{Northeastern University, Boston, MA, USA. Email: \texttt{rraj@ccs.neu.edu}} 
\and
    David Stalfa\thanks{Northeastern University, Boston, MA, USA. Email: \texttt{stalfa@ccis.neu.edu}}
\and    
    Zoya Svitkina\thanks{Google Research, Mountain View, CA, USA. 
    Email: \texttt{zoya@google.com}}
\and
    Aravindan Vijayaraghavan\thanks{Northwestern University, Evanston, IL, USA.  Email: \texttt{aravindv@northwestern.edu}}
}
\date{}

\maketitle

\begin{abstract}
    We consider the problem of scheduling $n$ precedence-constrained jobs on $m$ uniformly-related machines in the presence of an arbitrary, fixed communication delay $\delay$. Communication delay is the amount of time that must pass between the completion of a job on one machine and the start of any successor of that job on a different machine.  We consider a model that allows job duplication, i.e.\ processing of the same job on multiple machines, which, as we show, can reduce the length of a schedule (i.e., its makespan) by a logarithmic factor. Our main result is an $O(\log m \log \delay / \log \log \delay)$-approximation algorithm for minimizing makespan, assuming the minimum makespan is at least $\delay$.  Our algorithm is based on rounding a linear programming relaxation for the problem, which includes carefully designed constraints capturing the interaction among communication delay, precedence requirements, varying speeds, and job duplication.  To derive a schedule from a solution to the linear program, we balance the benefits of duplication in satisfying precedence constraints early against its drawbacks in increasing overall system load. Our result builds on two previous lines of work, one with communication delay but identical machines \cite{LR02}, and the other with uniformly-related machines but no communication delay \cite{chudak1999approximation,Li17}.

We next show that the integrality gap of our mathematical program is $\Omega(\sqrt{\log \delay})$.  Our gap construction employs expander graphs and exploits a property of robust expansion and its generalization to paths of longer length, which may be of independent interest. Finally, we quantify the advantage of duplication in scheduling with communication delay.  We show that the best schedule without duplication can have makespan $\Omega(\rho/\log \rho)$ or $\Omega(\log m/\log\log m)$ or $\Omega(\log n/\log \log n)$ times that of an optimal schedule allowing duplication.  Nevertheless, we present a polynomial time algorithm to transform any schedule to a schedule without duplication at the cost of a $O(\log^2 n \log m)$ factor increase in makespan.  Together with our makespan approximation algorithm for schedules allowing duplication, this also yields a polylogarithmic-approximation algorithm for  the setting where duplication is not allowed.  
\end{abstract}

\thispagestyle{empty}
\newpage


\setcounter{page}{1}

\newpage
\section{Introduction}

As computational workloads get larger and more complex, it becomes necessary to distribute tasks across multiple heterogenous processors. For example, the process of training and evaluating neural network models is often distributed over diverse devices such as CPUs, GPUs, or other specialized hardware; this process, commonly referred to as {\em device placement}\/ has gained significant interest \cite{mirhoseini2017device,gao+cl:place,mirhoseini:hierPlace}.
This gives rise to a multiprocessor scheduling problem of optimizing both the assignment of tasks to processors and the order of their execution.  
We address this problem, taking into account several complications that such a distributed setting presents, including job dependencies, heterogeneous machine speeds, and a communication delay between them.

The jobs comprising a workload can have data dependencies between them,  where the output of one job serves as the input to another. As is common in scheduling literature, we model these  dependencies using a directed acyclic graph (DAG), where a directed edge $uv$ represents that job $u$ must be scheduled before $v$. However, if these two jobs are executed on different machines, additional time is needed to transfer the data from one machine to the other. We model this time as a communication delay: this delay is zero if the two jobs run on the same machine, and is equal to some value $\delay$ if they run on different machines. 
Considering that the communication delay can be substantial, another aspect of the problem comes into play. Instead of a machine waiting for the result of some computation to be communicated from another machine, it may be advantageous for it to perform this computation itself, thus {\em duplicating}\/ work in order to obtain the result sooner (as highlighted in early work~\cite{papadimitriou1990towards}). Indeed, the technique of duplication to hide latency has been incorporated in schedulers proposed for grid computing and cloud environments~\cite{bozdag+oc:schedule,casas+trwz:schedule,hu+k:schedule,song+yjc:schedule}.
In addition, jobs may have different processing sizes and the  devices may run at different speeds, representing either different types (e.g.\ CPU, GPU, or TPU), or  differences in machine model.  

Optimization problems associated with scheduling
under  communication  delays  have  been  studied  over  the  last  three  decades, but  provably  good
approximation  bounds  are  few  and  several  challenging  open  problems remain~\cite{ahmad+k:schedule,bampis+gk:schedule,darbha+a:schedule,hoogeveen+lv:schedule,LR02,munier1999approximation,munier+h:schedule,palis1996task,papadimitriou1990towards,picouleau1991two,rayward1987uet}.  It is known
that scheduling a DAG of uniform size jobs on identical machines with a communication delay is
NP-hard, even when the number of machines is infinite~\cite{rayward1987uet,picouleau1991two}.  Several inapproximability
results have also been derived~\cite{bampis+gk:schedule,hoogeveen+lv:schedule}.  
However, these results are very limited and the approximability status of scheduling under communication delay is listed as one of the top ten open problems in scheduling surveys~\cite{bansal:survey,schuurman+w:survey}.  For the special case of uniform speeds and unit jobs, a logarithmic-approximation algorithm is presented in~\cite{LR02}.
In recent work~\cite{kulkarni}, a quasi-polynomial time approximation scheme is developed for the problem when the number of machines is $O(1)$, communication delays are $O(1)$, and the machines are identical.  Our focus in this paper is on deriving approximation algorithms for scheduling a DAG with {\em non-uniform size}\/ jobs on an {\em arbitrary}\/ number of {\em related machines}\/ (arbitrary speeds) and an {\em arbitrary communication delay}. 

    \subsection{Our results and techniques}
    We study the problem of scheduling a DAG with $n$ jobs of arbitrary sizes on $m$ {\em related machines}, connected by a network with a fixed communication delay.
In the related machines model, machine $i$ has a speed $s_i$, and the time taken to complete a job $v$ of size $p_v$ on $i$ is given by $p_v/s_i$.  We represent the network communication delay 
as $\delay$ times the processing time of the smallest job on the fastest machine.

\paragraph{Approximation algorithm for makespan.}
We focus on the {\em makespan}\/ objective, which is defined as the time taken by a given schedule to complete the given DAG on the machines.  We consider scheduling policies that allow duplication of jobs, which, as we discuss below, can reduce  makespan when compared to schedules that do not allow duplication. 

\begin{result}[Makespan approximation]
    There is a polynomial time algorithm that, given an instance of DAG scheduling with fixed communication delay, computes a schedule whose makespan is $O(\log m \log \delay/ \log \log \delay)(OPT + \delay)$, where $OPT$ is the optimal makespan for the given instance.
\end{result}
\noindent
We thus obtain an $O(\log m \log \delay/\log \log \delay)$-approximation algorithm as long as $OPT \ge \delay$, which is a natural requirement since it takes $\delay$ time to distribute the jobs to the machines at the start of the schedule as well as to synchronize termination at the end of the schedule. We note that the $\log m$ factor in our approximation corresponds to an upper bound on the number of geometrically separated speed groups. This entails that, for the special case of uniform speeds, our algorithm constructs a schedule with makespan upper bounded by $O(\log\delay/\log\log\delay)(OPT+\delay)$, thus extending the result of~\cite{LR02} to non-uniform job sizes.

A central component of our algorithm is a linear programming relaxation.
A significant challenge in this regard is to capture the precedence requirement in the presence of communication delays: 
we would like to determine where and when to schedule individual jobs while, at the same time, adjusting the start time of each job to account for communication delays in relation to \textit{all} its predecessors. 
We consider several related LPs and their natural extensions and show that these approaches are inadequate for our algorithm (see Appendix~\ref{sec:alternate_LPs}).
To overcome these challenges, we introduce a set of variables that indicate whether a job and its predecessor are scheduled within $\delay$ time of each other, and incorporate these variables into two new sets of constraints. The first enforces the delay requirement on jobs that do not start within $\delay$ time of each other, and the second upper bounds the total size of all predecessors that can be executed within $\delay$ time of their successor. The addition of these constraints exponentially reduces the integrality gap of our program.

Our rounding algorithm has two components.  First, we process a fractional solution to our linear relaxation to determine a tentative assignment of jobs to groups of machines, along the lines of \cite{chudak1999approximation}.
Next, we convert the group assignment to an actual schedule.  Unlike in the case of related machines with no communication delay, we cannot invoke a list scheduling type of policy.  Furthermore, our algorithm needs to duplicate jobs judiciously so as to hide the communication latency and achieve the desired approximation ratio. 
The main challenge in this regard is that, in order to make sufficient progress on the LP solution, we must duplicate some jobs on machines much slower than their assigned machine. We overcome this obstacle by upper bounding the total size of any duplicated jobs and structuring the machines such that those with slower speed have, as a whole, higher capacity. 

\paragraph{Integrality gap.}
We next study the integrality gap of the linear program underlying our approximation algorithm, and its dependence on the communication delay $\delay$.  Previous work of~\cite{chudak1999approximation} on scheduling on related machines implies an integrality gap of $\Omega(\log m/\log \log m)$ for non-uniform speeds and non-uniform job sizes, but it does not consider communication delays and hence does not yield any gap in terms of $\delay$.  

\junk{
\begin{itemize}
    \item We show that the integrality gap of the linear program is $\Omega(\sqrt{\log \delay})$, even when all the machines are identical and all the jobs are of the same size.
\end{itemize}
(See Theorem~\ref{thm:IG} for a formal statement).
}
\begin{result}[Integrality gap]
    There is a family of instances with uniform speeds and uniform job sizes such that for any $\rho$ that is at least some sufficiently large constant, our linear programming relaxation has a gap of at least $\Omega(\sqrt{\log \rho})$.
\end{result}
\noindent
This integrality gap gives the first evidence that constant factor approximations may not be tractable or may be out of reach of existing techniques when the communication delay $\delay$ is super-constant, even with uniform job sizes and identical machines. 
The integrality gap also extends to variants of time-indexed linear programs and, we suspect, to a wider class of mathematical programming relaxations. 
Given that without communication delay, the unit speed and unit job size case has an integrality gap of at most 2 by Graham's list scheduling~\cite{graham:schedule}, our result suggests a separation in the approximability between the variants of precedence-constrained scheduling with and without communication delays. 

Our gap construction consists of a layered DAG with $L=\omega(1)$ layers, where the dependency graph between successive layers corresponds to a random graph. 
The main technical challenge is to argue that $\Omega(L)$ phases (a phase here corresponds to roughly $\delay$ time units) are needed in order to schedule all the jobs for the optimal integral solution. The expansion of the random graph implies that at most $o(1)$ fraction of the jobs can be scheduled in the first phase. However, in the next phase, the jobs that were completed previously are now available on all the machines; moreover, the remaining graph (on the unscheduled jobs) in subsequent phases is {\em not random} any longer!
To overcome this technical hurdle, 
we identify and exploit a property of ``robust expansion'' and its generalization to paths of longer length, which may be of independent interest.  Section~\ref{sec:overview.integrality} provides an overview of our integrality gap result, and Section~\ref{sec:integrality} contains the full proof.

\paragraph{Bounding the duplication advantage.}
Given the potential of duplication to effectively hide communication latency, a natural question arises: how much smaller can the makespan of a schedule with duplications be, when compared to a \emph{no-duplication} schedule, i.e., a schedule in which each job is processed exactly once?  Our final set of results formally quantifies the {\em duplication advantage}.
\junk{
\begin{itemize}
    \item We show that the makespan of the best no-duplication schedule can be $\Omega()$ times the makespan of the best general schedule.
    We next show that any schedule can be transformed into a no-duplication schedule with an increase in makespan bounded by a factor of $O(\log^2 n \log m)$ in polynomial time.  
\end{itemize}
}
\begin{result}[Bounding the duplication advantage]
\textbf{Upper bound:} Given any instance with $n$ jobs, $m$ machines, communication delay $\rho$, and a schedule with makespan $C^* \ge \rho$, there exists a polynomial-time computable no-duplication schedule with makespan $O(C^* \cdot \log^2 n \log m)$.
\textbf{Lower bound:} There exists an instance with $n/2 = m = 2^\rho$ for which any no-duplication schedule has makespan at least $\rho/\log \rho$ times the optimal makespan.  
\end{result}
\noindent
Together with our makespan algorithm for general schedules, the algorithm of the above theorem yields a {\em polylogarithmic approximation makespan algorithm for no-duplication schedules}. 
Note that the preceding approximation ratio holds even when the makespan of an optimal no-duplication schedule is less than $\rho$ since it is straightforward to determine whether there is a no-duplication schedule that completes all jobs in less than $\rho$ time without any communication. Section~\ref{sec:overview.duplication} gives an overview of our algorithm that transforms a general schedule to a no-duplication schedule, and Section~\ref{sec:duplication} contains the full proofs for bounding the duplication advantage.
    
    \subsection{Related work}

Scheduling theory has a rich history and there is extensive work on scheduling jobs with precedence constraints dating back to over three decades.  In the following, we review scheduling work most closely related to this paper: scheduling DAGs on related machines, and scheduling DAGs under communication delays.

\paragraph{Scheduling DAGs on related machines.}
The problem of scheduling DAGs on related machines (with no communication delays) to minimize weighted completion time was first studied by Jaffe, who gave an $O(\sqrt{m})$ approximation algorithm~\cite{Jaffe_1980}.  This was significantly improved by Chudak and Shmoys who first derived an $O(\log m)$ asymptotic approximation ratio for minimizing makespan~\cite{chudak1999approximation} and then invoked a general framework due to Hall et al~\cite{Hall_1997} and Queyranne and Sviridenko~\cite{Queyranne_2002} to convert an approximation algorithm for makespan to an approximation algorithm for weighted completion time.  The Chudak-Shmoys algorithm for makespan minimization first solves an LP relaxation for the problem, and then assigns each job to a group of machines whose speeds are within a factor of two of one another.  Using Graham's list scheduling~\cite{graham:schedule}, they then schedule the jobs within each group of machines.  The $O(\log m)$ factor arises due to the number of machine groups.  In subsequent work, Chekuri and Bender derived the same $O(\log m)$ approximation via a combinatorial algorithm~\cite{chekuri1999precedence}.  In recent work, Shi Li improved the approximation ratio to $O(\log m/\log\log m)$ by a more careful tradeoff between the factor lost for organizing the machines into groups and the factor lost while assigning jobs to machine groups~\cite{Li17}.  

With regard to hardness, it is known that the problem is hard to approximate to within a constant factor even for the special case of identical machines, where the particular constant depends on underlying complexity theory assumptions~\cite{LenstraK78,Bansal_2009,Svensson_2010}.  Recent work has also shown that the problem is hard to approximate to within any constant assuming the hardness of a particular optimization problem on $k$-partite graphs~\cite{Bazzi_2015}. 

\paragraph{ Scheduling under communication delays.}
As discussed above, optimization problems associated with scheduling under communication delays have been studied for three decades since the early work of~\cite{rayward1987uet,papadimitriou1990towards,veltman+ll:schedule}, but provably good approximation bounds are few.  All  previous work assumes uniform machines and either uniform job sizes or special cases such as $O(1)$ machines and $O(1)$ communication delay.  For instance, in the special case of unit-size jobs, identical machines, and unit communication delay, a 7/3-approximation is presented in~\cite{munier+h:schedule}, while~\cite{hoogeveen+lv:schedule} show that it is NP-hard to approximate better than a factor of 5/4.  Hardness results are also shown in~\cite{bampis+gk:schedule,picouleau1991two,rayward1987uet}.  To the best of our knowledge, our work is the first to develop algorithms for scheduling non-uniform jobs with precedence constraints on related machines connected by an arbitrary communication network with fixed delay. 

\junk{
It is known
that scheduling a DAG of uniform size jobs on identical machines with a communication delay is
NP-hard, even when the number of machines is infinite~\cite{rayward1987uet,picouleau1991two}.  Several inapproximability
results have also been derived~\cite{bampis+gk:schedule,hoogeveen+lv:schedule}.  }

The natural idea of duplication to hide communication latency was first studied by Papadimitriou and Yannakakis, who proposed a 2-approximation algorithm for scheduling DAGs on an unbounded number of identical machines with a fixed communication delay~\cite{papadimitriou1990towards}.  Improved bounds for infinite machines have been given in~\cite{ahmad+k:schedule,darbha+a:schedule,munier+k:schedule,palis1996task}.  For the case of a bounded number of machines,~\cite{munier1999approximation,munier+h:schedule} give approximation algorithms under some special cases of either very small or very large communication delay or with the DAG restricted to be a tree-precedence graph.  The only provable guarantee for a bounded number of machines with an arbitrary communication delay parameter is the work of Lepere and Rapine, who present an approximation algorithm for scheduling a DAG of unit-size jobs on identical machines with communication delay of $\delay$ units, which achieves a makespan $O((OPT + \delay)\log\delay/\log\log\delay)$~\cite{LR02}.  The recent work of~\cite{kulkarni} presents a novel quasi-polynomial time approximation scheme, based on the Sherali-Adams hierarchy framework, for the problem with $O(1)$ identical machines, non-uniform job sizes, and $O(1)$ communication delays. 
\junk{The approximation guarantee is asymptotic owing to the $\delay\log\delay/\log\log\delay$ additive term in the bound.}

\section{Problem formulation and notation}
\label{sec:prelim}
An instance of precedence constrained scheduling with fixed communication delay is a triple $(G,M,\delay)$  where $G$ is a directed acyclic graph, $M$ is a set of machines, and $\delay$ is the communication delay. In the graph $G=(V,E)$, the $n$ nodes of $V$ represent jobs and the edges of $E$ represent precedence constraints. 
Each job $v$ has a \textit{size} $\proc{v} > 0$ and for any subset $U\subseteq V$, we define $\setproc{U} = \sum_{u \in U} \proc{u}$. 
In the set of machines $M = \{1, \ldots, m \}$, each machine $i \in M$ has a speed $s_i$. We order the machines such that $s_1 \le s_2 \le \ldots \le s_m$. Processing a job $v$ on a machine $i$ takes $\proc{v}/s_i$ units of time. We normalize these values so that the shortest job has size 1 and the fastest machine has speed 1, in which case one time unit is defined as the time needed to process the shortest job on the fastest machine. 
Each job may be \emph{duplicated}, i.e.\ copies of it processed on different machines. 
Preemption is not allowed, and at most one job can run on a machine at any given time.

\begin{wrapfigure}{r}{0.55\textwidth}
    \vspace{-20pt}
    \begin{center}
        \begin{tabular}{|cl|cl|}
            \hline
            $m$ & number of machines & $n$ & number of jobs \\
            $i,j$ & machines & $v,u$ & jobs \\
            $\speed{i}$ &  speed of machine $i$ & $\proc{v}$ &  size of job $v$ \\
            $\delay$ & communication delay & $\predex{v}$ &  predecessors of $v$ \\
            \hline
        \end{tabular}
    \end{center}
    \vspace{-20pt}
\end{wrapfigure}
We say that $u$ is a {\em predecessor} of $v$, denoted $u \prec v$, if there is some (non-zero length) directed path from $u$ to $v$ in $G$. We denote the set of all predecessors of $v$ by $\predex{v}$ (note that $v \not\in \predex{v}$). The parameter $\delay$ specifies the time needed to communicate the result of a job computed on one machine to a different machine.  So if $u \prec v$ and $v$ starts on machine $i$ at time $t$, then there must be a copy of $u$ that completes either on machine $i$ by time $t$ or on a different machine by time $t-\delay$. 

We represent a schedule as a function $\sigma: V \times M \to \mathbb{R} \cup \{\infty\}$ mapping pair $(v,i)$ to the start time of $v$ on $i$, or to $\infty$ if $v$ is not scheduled on $i$. \junk{Note that this implicitly assumes each job is executed at most once on each machine, which can be done without any loss in makespan.} 
We say that $\sigma$ is a schedule of $(G,M,\delay)$ if all jobs in $G$ have a finite start time on some machine in $M$ subject to the constraints listed above.
The objective is to find a $\sigma$ with minimum makespan, which is the maximum (finite) completion time in $\sigma$ of any copy of any job. 
Since this objective is trivial if there is only one job or one machine, we assume $n, m \ge 2$. In the three field notation, this problem is denoted $Q | \mbox{duplication, prec}, c | C_{\max}$ where $c$ indicates uniform communication delay. 

\section{Overview of the results and techniques}
\label{sec:overview}

    \subsection{Approximation algorithm for makespan}
    \label{sec:overview.makespan}
    At a high level, our algorithm finds a fractional solution to the scheduling problem and then, through a series of refinements, constructs a final schedule for the given instance.  The various components of the algorithm are highlighted in the figure below.  The first step is a standard preprocessing of the instance, in which we eliminate machines that are slower than the fastest machine by a factor of $m$ or more, while incurring at most a constant factor increase in makespan.  We refer the reader to Section~\ref{sec:makespan} for details.

\begin{figure}[H]
    \centering
    \includegraphics[scale=0.9]{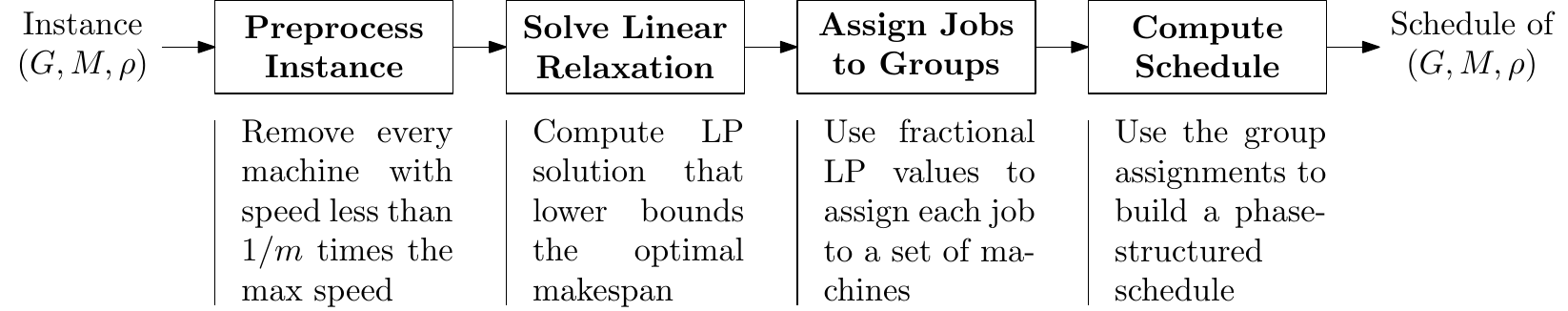}
    \label{fig:high_level_algo}
\end{figure}

\vspace{-5mm}
\paragraph{Approaches based on previous related work.}
We briefly review natural approaches to the problem of scheduling on related machines with communication delay, based on previous related work and indicate the ways in which these approaches are inadequate for our setting. 

One approach is that taken in \cite{LR02}, which uses a combinatorial algorithm for the case with unit-speed machines, unit-size jobs, communication delay $\delay$, and duplication allowed. The crux of the algorithm is to repeatedly find jobs that can be completed in $\delay$ steps and schedule them (duplicating their uncompleted predecessors, if necessary) until all such jobs have at least half their remaining predecessors already scheduled. At this point the algorithm introduces a delay on all machines, removes all previously scheduled jobs, and repeats. While this approach may work for arbitrary job sizes, accounting for variable speeds is more difficult. In Appendix~\ref{sec:combinatorial_approach} we show that a natural extension of this combinatorial algorithm fails. 

A more effective approach is to develop a suitable LP relaxation. We consider natural variants of two relaxations developed in related work. The first captures precedence constraints and communication delays by relating job-machine assignment variables to job start and completion time variables. In this way, precedence constraints can be addressed effectively, as has been shown by~\cite{chudak1999approximation, Li17}, but communication delays are much more challenging to capture. One natural approach is to add same-machine indicator variables $\delta_{u,v,i}$. Intuitively $\delta_{u,v,i} =1$ if $u$ and $v$ are both scheduled on machine $i$, otherwise $\delta_{u,v,i} = 0$. We could then add the following constraint, where $S_v$ and $C_v$ represent the start and completion times of job $v$.
\vspace{-0.1in}
\begin{align*}
    S_v \ge C_u + \delay \Big(1 - \sum_i \delta_{u,v,i}\Big) &&\forall u,v,i : u \prec v
\end{align*}
\vspace{-0.2in}

\noindent We can think of the constraint as stating that any job $v$ must begin at least $\delay$ steps after any of its predecessors $u$, if $v$ and $u$ are not executed on the same machine.  
Unfortunately, a simple instance with a fractional solution that spreads each job among all the machines and sets the $\delta$ values to $1/m$ leads to an integrality gap as large as a polynomial in  $\delay$, $m$, and $n$. See Appendix~\ref{sec:same_machine_variable} for details.  

A different strategy for constructing a linear relaxation is to use time-indexed job-machine assignment variables $x_{v,i,t}$ to indicate the completion time $t$ of job $v$ on machine $i$. 
Indeed, such a program capturing both precedence constraints and communication delays is used in~\cite{kulkarni} to obtain a quasi-polynomial time approximation scheme when the number of machines $m$ is $O(1)$, the communication delays are $O(1)$, and all machines are identical.  Unlike \cite{kulkarni}, however, we are working with an arbitrary number of machines of arbitrary speeds, and arbitrarily large communication delay.  In this case, the time-indexed relaxation has an integrality gap as large as a polynomial in $\delay$, $m$, and $n$. See Appendix~\ref{sec:time_indexed}. 

\junk{
Our approach is along the lines of \cite{chudak1999approximation,Li17}. We develop a more sophisticated set of linear delay constraints, in addition to the standard scheduling constraints. We then use elements of the rounding in \cite{chudak1999approximation} and the algorithm in \cite{LR02} to obtain a schedule from a fractional solution to our relaxation. The combination of these prior rounding techniques with the new program yields a schedule with polylogarithmic approximation guarantees.}

\paragraph{Developing our relaxation.}
To overcome the challenges mentioned above, we introduce two new sets of constraints - \textit{delay constraints} and \textit{phase constraints} - in addition to the usual related machines scheduling constraints of \cite{chudak1999approximation,Li17}, where a \textit{phase} is any interval of $\delay$ time in a schedule. To build intuition, we introduce these constraints in the setting with unit speeds and unit job sizes. We then provide a natural (but weak) generalization of these constraints to the setting with arbitrary speeds and job sizes which, unfortunately, has a large integrality gap. 
Finally, we refine the constraints yielding our linear relaxation. 

For unit speeds and unit jobs size, the phase constraints require that if a job $v$ is scheduled to start at time $t$ on machine $i$, then the total number of $v$'s predecessors that are scheduled to start in the interval $[t-\delay,t)$ is at most $\delay$ because they must all be scheduled on the same machine. To capture this property, we introduce {\em same-phase} variables $\phasevariable{u}{v}$ for each pair of jobs $u,v$ such that $u \prec v$. We can view $\phasevariable{u}{v}$ as indicating whether $u$ is scheduled within $\delay$ steps of the start of $v$. We can then give the following constraints.  
\vspace{-0.08in}
\begin{align*}
    \start{v} \ge \start{u} + \delay(1-\phasevariable{u}{v}) \;\;\; \forall u,v: u \prec v 
    \hspace{.8in} \text{and} \hspace{1in}
    \delay \ge \sum_{u \prec v} \phasevariable{u}{v}    \;\;\;\forall v
\end{align*}
\vspace{-0.15in}

\noindent The first is the delay constraint and states that the difference in start times for $v$ and $u$ is at least $\delay$ if $u$ is not scheduled within $\delay$ of the start time of $v$. The second is the phase constraint and states that the total number of $v$'s predecessors that are scheduled to start within $\delay$  time of $v$ is at most $\delay$. While this relaxation has a small integrality gap in the unit case, adapting it to the non-unit case is not straightforward.

In the case with arbitrary speeds and job sizes, we would like to capture the property analogous to the one used in the unit case: if a job $v$ is scheduled to start at time $t$ on machine $i$ then the set of all $v$'s predecessors that are scheduled to start in the interval $[t-\delay,t)$ should have total size at most $\delay\speed{i}$.
The following relaxation, which retains the same-phase variables of the unit relaxation as well as the unit delay constraint, shows a natural way to extend the phase constraint to capture this property.
\vspace{-0.08in}
\begin{align*}
    \delay \sum_{i}  \speed{i} \jobmachine{v}{i}  &\ge \sum_{u \prec v} \proc{u} \phasevariable{u}{v}   &\forall v
\end{align*}
\vspace{-0.15in}

However, these constraint have a flaw.  If, say a small fraction of $v$ is placed on the fastest machine and the rest on the slowest, then the left-hand term will allow too many predecessors to be scheduled in the same phase. As shown in Appendix~\ref{sec:simple_same_phase_extension}, this leads to an integrality gap as large as $\delay$ or polynomial in $m$ and $n$. 

A key idea in our linear relaxation is the introduction of {\em machine-dependent same-phase}\/ variables, which tie the notion of a phase to the speed of a particular machine. Using these variables, we introduce new phase and delay constraints which rely crucially on our ordering of machines by increasing speed. 
Our linear relaxation LP minimizes $C$ subject to the following constraints. 

\vspace{-3mm}
\noindent
\begin{minipage}{.56\linewidth}
  \begin{flalign}
        \LPmakespan &\ge \start{v} + \proc{v} \sum_i \jobmachine{v}{i} / \speed{i} &\forall v
        \tag{\ref{phaseLP_completionlb}}
        \\
        \start{v} &\ge \start{u} + \proc{u} \sum_i  \jobmachine{u}{i} / \speed{i} &\forall u,v: u \prec v
        \tag{\ref{phaseLP_execution}}
        \\
        \start{v} &\ge \start{u} + \delay \Big( \sum_{j \le i} \jobmachine{v}{j} - \machinephase{u}{v}{i} \Big) &\forall u,v,i: u \prec v
        \tag{\ref{phaseLP_delay}}
        \\
        \sum_{j \le i} \jobmachine{v}{j} &\ge \sum_{u \prec v} \proc{u} \machinephase{u}{v}{i} / \delay\speed{i}   &\forall v, i
        \tag{\ref{phaseLP_phasepredecessors}}
  \end{flalign}
\end{minipage}\hfill
\begin{minipage}{.405\linewidth}
  \begin{flalign}
        \LPmakespan\speed{i} &\ge \sum_v \proc{v} \jobmachine{v}{i} &\forall i
        \tag{\ref{phaseLP_loadlb}}
        \\
        \sum_i  \jobmachine{v}{i} &= 1 &\forall v
        \tag{\ref{phaseLP_finishjob}}
        \\
        \start{v} &\ge 0 &\forall v
        \tag{\ref{phaseLP_nonnegstart}}
        \\
        \jobmachine{v}{i} &\in (0,1) &\forall v,i
        \tag{\ref{phaseLP_jobmachine}}
        \\
        \machinephase{u}{v}{i} &\in (0,1) &\forall u,v,i: u \prec v
        \tag{\ref{phaseLP_samephase}}
  \end{flalign}
\end{minipage}
\vspace{2mm}

We provide some intuition behind the variables and constraints.  We interpret the variables $\jobmachine{v}{i}$ as giving the ``primary" placement of $v$ and $\start{v}$ as the corresponding start time of $v$.  Then, for any jobs $u$ and $v$ such that $u \prec v$ and for any machine $i$, we can understand the variable $\machinephase{u}{v}{i}$ as indicating, first, whether $v$ is executed on a machine indexed $i$ or lower, and second, whether the start time of $u$ is within $\delay$ of the start time of $v$. The significance of this indication is that, if these conditions are met, then some \textit{copy} of $u$ must execute on the same machine as $v$ within $\delay$ time of $v$ and, therefore, only predecessors of total size at most $\delay\speed{i}$  can meet these conditions.  The remaining variable $C$ captures the makespan of the resulting schedule. 

\junk{
In the duplication setting, we can interpret the values of $\jobmachine{v}{i}$ as giving the \textit{primary} placement of $v$. That is, in order to finish $v$ as early as possible, we should place it on a machine where it has high mass according to $\jobmachine{v}{i}$. Now, suppose we decide to schedule $v$ on machine $i$ and consider the $\delay$ length time interval before the start of $v$. Since there is no communication during this time, we can select some subset of $v$'s predecessors with total size up to $\delay \speed{i}$ to place on $i$ in this interval.  We can then think of $\machinephase{u}{v}{i}$ as giving this \textit{secondary} placement of $u$. That is, if we schedule $v$ on machine $i$ and $\machinephase{u}{v}{i}$ is large then, in order to finish $v$ as early as possible, we should place a \textit{copy} of $u$ in the same phase as $v$ on $i$. 
\david{see Zoya's comment on this explanation in Section 4}
}

The delay constraint~(\ref{phaseLP_delay}) states that if $v$ is scheduled on a machine slower than $i$, then $v$ should start at least $\delay$ time after any predecessor $u$ unless $u$ is scheduled in the same phase as $v$.  The phase constraint~(\ref{phaseLP_phasepredecessors}) states that if $v$ is scheduled on a machine slower than $i$, then the total size of $v$'s predecessors scheduled in the same phase is at most $\delay\speed{i}$.
The remaining constraints ensure that no job completion time exceeds the makespan 
(\ref{phaseLP_completionlb}), that jobs are executed completely and in order  (\ref{phaseLP_execution}, \ref{phaseLP_finishjob}), and that the total load on any machine does not exceed the makespan (\ref{phaseLP_loadlb}).

\paragraph{Group assignment.} 
The fractional solution we obtain for the relaxation of LP gives us a fractional assignment of jobs to machines, as well as lower bounds on start times of jobs. The objective function is the maximum over all job completion times as well as over all machine loads, and so lower bounds the optimal makespan. 
The next step is to convert this solution into an assignment $\assignfunc$ of each job to some set of machines. This assignment will guide our final construction of the schedule. We partition the set of machines into $K \le \log m$ groups $\group{1}, \ldots, \group{K}$ of increasing speed and define a job's ``median'' machine group as the lowest (slowest) one such that the job's total fractional assignment to this and slower groups is at least $\sfrac{1}{2}$.  Our group assignment follows an approach similar to~\cite{chudak1999approximation,Li17}: we assign each job to the highest capacity group that is at least as fast as its median group. 
Note that, if there are jobs assigned to groups $\group{k}$ and $\group{k'}$, with $k < k'$, then the minimum speed in group $\group{k}$ is less than that in group $\group{k'}$, but the capacity of $\group{k}$ is at least that of $\group{k'}$, since the jobs assigned to $\group{k}$ could have been assigned to group $\group{k'}$ but were not.  

\paragraph{Computing the schedule.}
Our scheduling algorithm (Algorithm \ref{alg:scheduler}) takes the group assignment $\assignfunc$ and produces a schedule, with possible duplications, for all jobs. 
The main challenge in constructing the schedule is balancing two conflicting incentives. On the one hand, the more we allow a set of jobs to be duplicated, the faster we can finish any jobs preceded by jobs in the set. On the other hand, if we duplicate too often, then we risk overloading machines with too many jobs to execute. 
Specifically, we want to avoid scheduling too much load assigned to higher capacity groups on lower capacity (faster speed) groups, even when doing so would allow us to complete some jobs earlier. We strike this balance by allowing a job to be duplicated only in groups with capacity higher than its assigned group. Furthermore, similar to \cite{LR02}, when the scheduler places a set of jobs on a machine, we require that at least a $1/\overlap$ fraction of the total size of that set be from jobs that have not yet been placed on any machine, where $\overlap$ will be set later.

The scheduling algorithm proceeds in a series of rounds. In each round, the algorithm iterates through each machine group $\group{k}$ and considers each job $v$ with $\assignfunc(v)=k$ that has not yet been scheduled. On a machine $i \in \group{k}$ the algorithm schedules $v$ and its predecessors that have not been completed in earlier phases if the following three conditions are satisfied: \ref{condition:predecessor_size} $v$'s incomplete predecessors can be completed on $i$ in time $O(\delay)$; \ref{condition:overlap} the total size of $v$ and its predecessors not already scheduled (on any machine) is at least a $1/\overlap$ fraction of the total size of its uncompleted predecessors; and \ref{condition:predecessor_groups} all of $v$'s remaining predecessors have been assigned to higher indexed groups.  
Condition~\ref{condition:predecessor_groups} ensures that we duplicate jobs only from lower capacity groups to higher capacity groups.
Condition~\ref{condition:predecessor_size} ensures that any jobs we duplicate from lower capacity, higher speed groups won't take too long on the lower speed group. 
Condition~\ref{condition:overlap} ensures two things: first, it guarantees that the total increase in load from duplication is no more than $\overlap$ and, second, it guarantees any large gaps in the schedule result from the fact that all those jobs with a small number of predecessors have the total size of their remaining predecessor reduced by a factor of $\overlap$. 
The usefulness of these conditions is made more explicit in the analysis section. 

\junk{
Some of the intricacies of the algorithm are due to the fact that there could be long jobs that do not finish in one phase. So consider first a special case when there are no long jobs. 
By conditions (C1) and (C2), if $v$ itself can finish in at most $\delay$ amount of time on machine $i$, then even together with its predecessors it will finish by the end of that phase. There is some idle time in the beginning of the next phase to allow for the communication delay, which ensures that the resulting schedule is feasible. In case that there are long jobs in the problem instance, the algorithm only schedules their successors in a phase after the long job has been completed. This is enforced by (C1).
So the communication requirement is satisfied in that case as well. 
}

\paragraph{Overview of the analysis.}
For the purposes of analysis, we divide our schedule into phases of length $\delay$ and partition these phases into three types. We then bound the makespan of our schedule by bounding the total number of phases of each type. 
Our analysis combines elements of the analysis in \cite{chudak1999approximation} and \cite{LR02}.

The three types of phases are \textit{chain} phases, \textit{load} phases, and \textit{height} phases. 
We define a chain $\chain$ such that each element in $\chain$ precedes the next, and each element has an instance which takes a sufficiently long time in the schedule. Chain phases are those phases in which some machine spends most of its time working on some chain element. 
All non-chain phases are divided into load and height phases.
Load phases are those non-chain phases in which every machine of some group is working on jobs for most of the phase. 
The remaining phases are height phases. 
We can think of the three categories more intuitively as follows. Chain phases primarily reduce the remaining execution time of the chain. Load phases primarily reduce the remaining execution time of the set of all jobs.
Height phases primarily reduce the amount of time before the next chain phase (or the end of the schedule if the chain has been completed).
Figure~\ref{fig:chain_consruction_no_label} depicts the relationship between the chain and the sets of jobs on which height phases make progress.
\begin{figure}
    \centering
    \includegraphics[width=\textwidth]{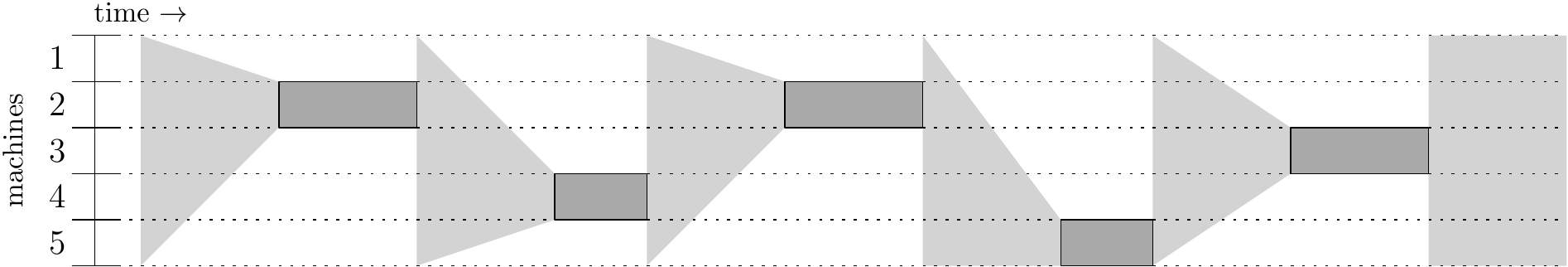}
    \caption{Machines are shown vertically on the left and time increases from left to right. The chain is shown as dark gray boxes. Each Light gray, borderless area represents the set of jobs that precede the chain job to its right (if it exists) and complete after the chain job to its left (if it exists).  
    }
    \label{fig:chain_consruction_no_label}
\end{figure}

We now briefly overview how we bound the number of phases of each type. We first discuss chain phases. Since chain jobs take a long time in the schedule, condition~\ref{condition:predecessor_size} ensures that every chain job is scheduled only on machines in its assigned group. Since we derived the group assignments from LP, the time spent executing jobs in the chain is at most $O(OPT)$, so the total number of chain phases is at most $O(OPT/\delay)$. We next consider load phases. Condition~\ref{condition:predecessor_groups} guarantees that the set of jobs scheduled on groups $\group{k},\ldots,\group{K}$ is a subset of the jobs assigned to these groups by $\assignfunc$. So, by condition~\ref{condition:overlap}, we have that for any $k$, the total load on groups $\group{k},\ldots,\group{K}$ is at most an $\overlap$ factor above the total load assigned to those groups by $\assignfunc$. Using a lemma from \cite{chudak1999approximation}, this entails that the total number of phases is no more than $O(OPT \cdot K\overlap/\delay)$.

Bounding the number of height phases is more involved as it requires a closer analysis of the linear program as well as a more detailed understanding of the step-by-step operation of the scheduling algorithm. We first partition the jobs in \textit{bands} $\band{1},\band{2},\ldots$ according to their start times as given by LP. We show that, for each job $v$ in a band, the total size of $v$'s predecessors in the same band is small enough to be completed in $O(\delay)$ time on $v$'s assigned group. Then, for each height phase $\tau$, we consider the lowest band $\band{r}$ with some job scheduled after phase $\tau$ and the slowest group $\group{k}$ with a job in that band. Let $v$ be some unscheduled job in $\band{r}$ assigned to group $\group{k}$. We consider a series of height phases separated by at most $O(1)$ height phases. We show, for each height phase in this series, that there is some iteration of our scheduling algorithm in which the algorithm \textit{considers} placing $v$ with its remaining predecessors on some machine in $\group{k}$ and in which all of $v$'s predecessors that started in the previous height phase in the series have completed with enough time to communicate the results to all machines. Due to our choice of $v$, we can then infer that, if the algorithm does not place $v$ in this iteration, it is because $v$'s uncompleted predecessor set violates condition~\ref{condition:overlap}. This entails that by the next height phase in the series, the size of $v$'s remaining predecessor set is reduced by a factor of $\overlap$. Since $v$'s predecessors within the band can be completed in $O(\delay)$ time on any machine in group $\group{k}$, we have that after $O(\log_{\overlap} \delay)$ height phases $v$'s predecessor set is empty. This entails that $v$ is scheduled before (or during) the next height phase in the series. Letting $r^*$ be the number of bands, this argument upper bounds the number of height phases by $O(Kr^* \log_{\overlap} \delay)$. We then show that the number of bands $r^*$ is $O((OPT+\delay)/\delay)$, which gives the desired bound on the number of height phases.

Finally, we set $\overlap$ to $\log\delay/\log\log\delay$. Summing over the number of phases of each type, we have that the length of our schedule is upper-bounded by $O(K \cdot \log\delay/\log\log\delay)(OPT + \delay)$.

    \subsection{Integrality gap}
    \label{sec:overview.integrality}
    We construct a new integrality gap instance that achieves a $\omega(1)$ integrality gap in the presence of communication delays. The gap construction consists of a layered DAG with $L=\omega(1)$ layers and $n$ vertices in each layer, where each job in layer $\ell$ has dependencies on $d$ randomly chosen jobs in $V_{\ell+1}$ as shown in Figure~\ref{fig:lb:intro}. In particular, $\rho=d^L=n^{c}$ for a small constant $c>0$. The parameters of the construction are set up in such a way that fractionally all the jobs can be assigned in one phase (hence the LP solution value is at most $\rho$).

\begin{wrapfigure}{r}{0.6\textwidth}
  \vspace{-20pt}
  \begin{center}
    \includegraphics[width=0.59\textwidth]{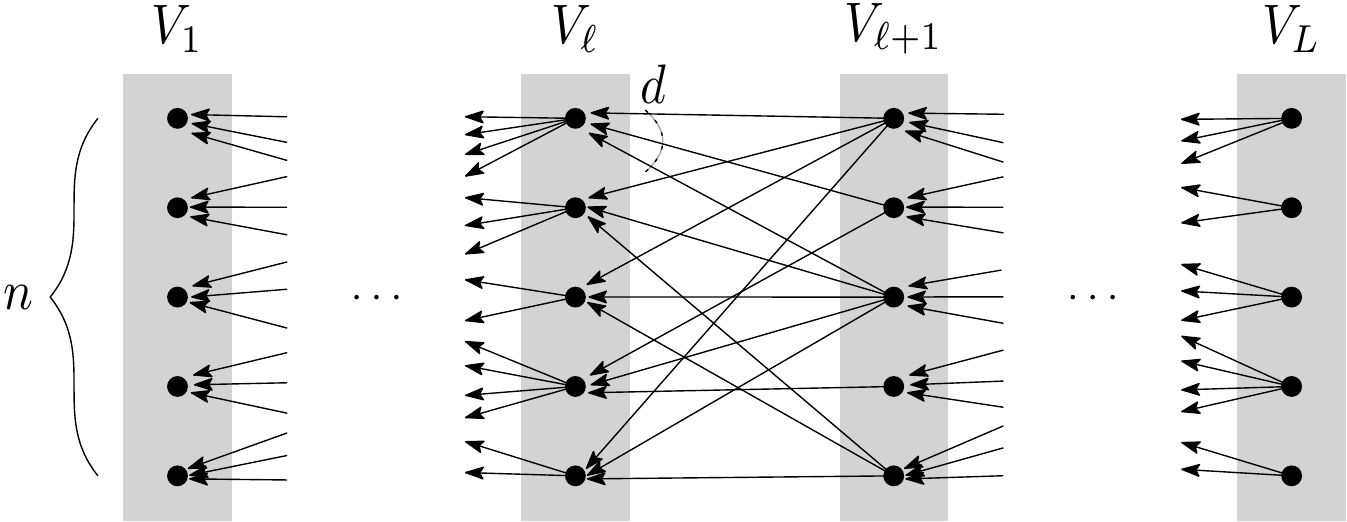}
  \end{center}
  \vspace{-10pt}
  \caption{\label{fig:lb:intro} The figure shows the DAG with $L$ layers $V_1, \dots, V_L$ representing the $nL$ jobs. Each of the $n$ jobs in $V_{\ell}$ has  dependencies on $d$ randomly chosen jobs in $V_{\ell+1}$. We set $\rho=d^{L}$, $m=\rho$, and the parameters $L=c_1 \sqrt{\log n}, d= 2^{c_2 \sqrt{\log n}}$ for some appropriate constants $c_1, c_2>0$.} 
\vspace{-10pt}
\end{wrapfigure}

The main technical challenge is to argue that $\Omega(L)$ phases are needed to schedule all the jobs in order to get a lower bound of $\Omega(L\rho)$ for the integer solution value. This gives a gap of $\Omega(L)=\Omega(\sqrt{\log \rho})$. From the expansion of the random graph in each layer, it is easy to argue that at most a $o(1)$ fraction of the jobs in layers $\{1,\dots, L-2\}$ can be scheduled in the first phase (since at most $\rho \ll n$ of the jobs can be on one machine). However, in the next phase the  results of all jobs that were scheduled previously are now available to all the machines; moreover the choice of these jobs could depend on the randomness in the DAG. Hence the remaining graph in each layer (after removing vertices that have already been scheduled) in the subsequent phases is {\em not random} any longer!

To overcome this technical hurdle, 
we identify and exploit a property of {\em robust expansion}, which may be of independent interest. The standard vertex expansion property of a random graph says that w.h.p.\ any subset $S \subset V_\ell$ of size $|S| \le n/d$ has a neighborhood of size $|\Gamma(S)| =\Omega(d |S|)$. However, random graphs have the stronger property that no subset $T$ of size $o(d |S|)$ can have $\Omega(d |S|)$ of the edges from $S$ incident on it. For our analysis, we need to prove a generalization for paths of length $\ell < L$ (Lemma~\ref{lem:generalized}): {\em w.h.p. for any $S \subset V_i$ (of sufficiently small size), there is no subset of size $o(d^\ell |S|)$ that can have $\Omega(d^\ell |S|)$ of the length-$\ell$ paths going into $S$}. 

Each job $u$ in layer $i$ (i.e. a vertex in $V_i$) has $d^{\ell-i}$ incoming paths from layer $V_\ell$, and all of the vertices in these paths need to be scheduled before scheduling $u$ -- either in a previous phase, or on the same machine in the current phase. The above robust expansion property is used to upper bound the number of jobs completed in each phase in two different ways: 1) to upper bound the number of jobs in $V_i$ whose dependencies in $V_\ell$ ``mostly'' consists of jobs scheduled in previous phases, and 2) to upper bound the number of jobs in $V_i$ such that most of their dependencies in $V_{\ell}$ need to be resolved in the current phase. This allows us to prove that we need at least $L/2$ phases before most of the jobs in $V_1$ can be scheduled. 

We believe that our integrality gap argument applies to a wider class of relaxations for the problem.  Any program that captures communication delay and precedence requirements through individual constraints for each job and has independent load constraints for each machine is likely to incur a similar gap.  
    
    \subsection{Bounding the duplication advantage}
    \label{sec:overview.duplication}
    The final contribution of this paper is to quantitatively characterize the duplication advantage.  While it is easy to construct instances where the makespan of a schedule allowing duplication (which we refer to as a general schedule) is better than that of a no-duplication schedule (one in which all jobs are processed exactly once), our goal is to place upper and lower bounds on the duplication advantage.

\paragraph{Lower bound.} We first present a simple family of instances with $m$ identical machines, $n = 2m$ unit jobs, and $\rho = \log m$, for which any no-duplication schedule has makespan $\Omega(\rho^2/\log \rho)$, while the optimal makespan is at most $\rho$.  The DAG for such an instance consists of a rooted binary tree with $m$ leaves and edges directed away from the root, such that an optimal schedule executes each root-leaf path on a separate machine (with necessary duplication), while any no-duplication schedule is essentially forced to decompose the tree into $\rho/(\log \rho)$ phases, interspersed with communication delays. Note that we thus have  $\Omega(\log m/\log \log m)$ and $\Omega(\log n/\log\log n)$ bounds on the duplication advantage. 

\paragraph{Upper bound.} Our main result in this section is that the duplication advantage is, in fact, also upper-bounded by a polylogarithmic factor $O(\log^2 n \log m)$.  Our proof is through a polynomial-time algorithm that transforms any schedule to a no-duplication schedule with the polylogarithmic factor loss in makespan.  The algorithm processes a given (general) schedule in ``phases" of length $\rho$.  The core of the algorithm is to transform each phase into a no-duplication schedule of length 
$O(\rho \log^2 n \log m)$.  There are technical complications since (i) processing of jobs may span multiple phases of the schedule, and (ii) the number of phases may be super-polynomial in the size of the instance.  Both these can be handled relatively easily by considering machines that are processing ``long'' jobs separately, and ignoring phases where no jobs are started or completed.  These are detailed in Section~\ref{sec:duplication}.

We now give an overview of the algorithm's core.  Consider the sub-DAG $D$ of the original DAG formed by the jobs that are processed within a particular phase of the general schedule.  We face several technical challenges while designing a no-duplication schedule for this sub-DAG.  First, we need to determine the relative order between the jobs.  For example, if a node serves as an predecessor of many other jobs, it could be given higher priority, but those successors may have been processed in the general schedule at many distinct machines, along with a copy of the predecessor, something we cannot do in the no-duplication schedule.  Second, if we choose to process two jobs on two different machines in a phase, we have to ensure that they do not share a common predecessor.  

To address these challenges, we organize and process the jobs of $D$ as follows.  First, we divide them into $O(\log n \log m)$ groups based on their level of duplication in the general schedule; each group consists of jobs whose duplication level is within a factor of $(1 + 1/(2\log n))$ of one another.  We then process the groups from the highest level of duplication down to the lowest, since the duplication level of a job in $D$ is at least that of any of its successors.  Within a given group, we focus on the sink jobs (which have no predecessors in the group) and construct an undirected graph $H$ over them in which an edge exists between two sinks if they share a common predecessor.  Our key insight about $H$ is that any subset of jobs in $H$ that is composed of regions of diameter $O(\log n)$ that do not share any common neighbors among them can be processed in a single phase in a no-duplication schedule.  We show that using a classic low-diameter decomposition technique from approximation algorithms and distributed computing (e.g., see~\cite{awerbuch+p:partition,linial+s:decomposeJ,peleg:distributeBook}), we can find a subset of $\Omega(H)$ jobs that has the desired structure in $H$.  A recursive use of this subroutine, together with the other techniques indicated above, yields the desired no-duplication schedule. 

\setcounter{theorem}{0}

\section{Approximation Algorithm for Makespan Minimization}
\label{sec:makespan}
Our algorithm for minimizing makespan is based on a linear programming relaxation of DAG Scheduling with Communication Delay. This linear program is rounded to get an assignment of jobs to groups of machines, and the assignment is then used to schedule each job. In Section~\ref{sec:preprocess} we present the first subroutine of our algorithm, which removes all machines that are too slow to be useful.  Section~\ref{sec:lp} contains the second subroutine, which consists of solving the linear programming relaxation. In Section~\ref{sec:group_assign} we present the third subroutine, which consists of finding an assignment of each job to a group of machines, given a fractional assignment of jobs to machines given by the LP solution. In Section~\ref{sec:scheduler} we present our scheduling subroutine which takes an assignment of jobs to groups of machines and computes a schedule. Finally, in Section~\ref{sec:analysis} we show that the schedule produced by our algorithm has makespan within a polylogarithmic factor of the optimal makespan.

\begin{figure}[ht]
    \centering
    \includegraphics[width=\textwidth]{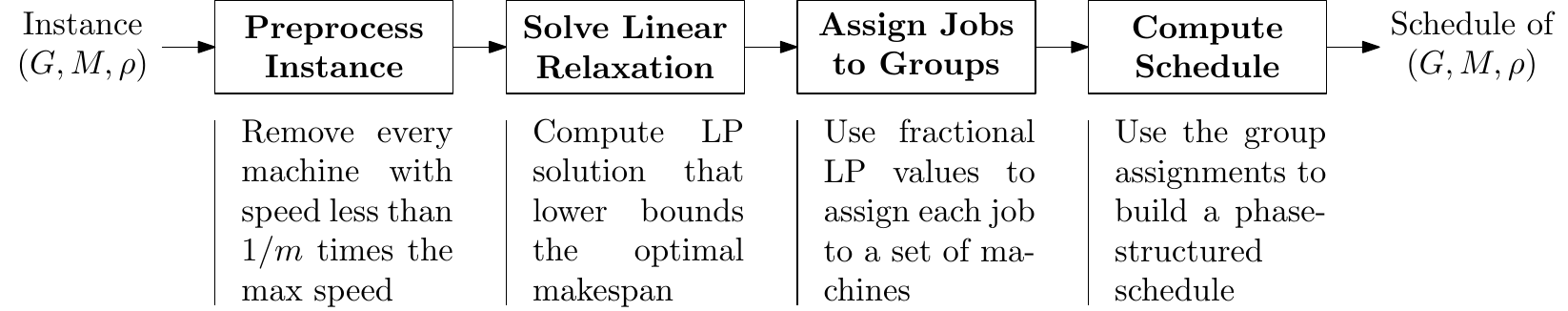}
\end{figure}

    \subsection{Preprocessing the Instance}
    \label{sec:preprocess}
    In this section, we present the first subroutine of our algorithm. Given an instance $(G,M,\delay)$, this subroutine outputs a new instance $(G, M', \delay)$ such that $M' \subseteq M$ where the speed of each machine in $M'$ is within a factor $m$ of the speed of the fastest machine in $M$. More formally, 
\[ M' = \{ i \in M : \speed{i} \ge \speed{m}/m\}  \]
where machine $m$ is the fastest machine in $M$. The following lemma is proved in Appendix \ref{sec:elim_slow_machines}.

\begin{lemma}
    Let $\optmakespan$ be the optimal makespan of any schedule of $(G,M,\delay)$ and let $C'$ be the optimal makespan of any schedule of $(G,M',\delay)$. Then $C' = O(\optmakespan)$.
    \label{lem:elim_slow}
\end{lemma}

The proof in Appendix~\ref{sec:elim_slow_machines} shows that any schedule $\sigma_1$ on $(G,M,\delay)$ with makespan $C_1$ can be converted into a schedule $\sigma_2$ of $(G,M',\delay)$ with makespan $C_2 \le 6 C_1$.
We note that an alternative strategy for eliminating these machines is given in \cite{chudak1999approximation}. There, the authors argue that any solution to their linear program defined over $(G,M)$ can be converted into a solution to their linear program on $(G,M')$. A corresponding result can be proved for our linear program.  We choose to use the method given in Appendix \ref{sec:elim_slow_machines} because the bound proved there holds for the optimal \textit{makespan}, independent of a particular relaxation or algorithm used to solve the problem, and is therefore more general. 

    \subsection{Linear Programming Relaxation}
    \label{sec:lp}
    In this subroutine, we formulate and solve a linear programming relaxation for the given instance of DAG Scheduling with Communication Delay.  The following linear program, called LP, minimizes $\LPmakespan$ subject to the following constraints. Indices $u$ and $v$ refer to jobs and $i$ and $j$ refer to machines.


\noindent
\begin{minipage}{.56\linewidth}
  \begin{flalign}
        \LPmakespan &\ge \start{v} + \proc{v} \sum_i \jobmachine{v}{i} / \speed{i} &\forall v
        \label{phaseLP_completionlb}
        \\
        \start{v} &\ge \start{u} + \proc{u} \sum_i  \jobmachine{u}{i} / \speed{i} &\forall u,v: u \prec v
        \label{phaseLP_execution}
        \\
        \start{v} &\ge \start{u} + \delay \Big( \sum_{j \le i} \jobmachine{v}{j} - \machinephase{u}{v}{i} \Big) &\forall u,v,i: u \prec v
        \label{phaseLP_delay}
        \\
        \sum_{j \le i} \jobmachine{v}{j} &\ge \sum_{u \prec v} \proc{u} \machinephase{u}{v}{i} / \delay\speed{i}   &\forall v, i
        \label{phaseLP_phasepredecessors}
  \end{flalign}
\end{minipage}\hfill
\begin{minipage}{.405\linewidth}
  \begin{flalign}
        \LPmakespan\speed{i} &\ge \sum_v \proc{v} \jobmachine{v}{i} &\forall i
        \label{phaseLP_loadlb}
        \\
        \sum_i  \jobmachine{v}{i} &= 1 &\forall v
        \label{phaseLP_finishjob}
        \\
        \start{v} &\ge 0 &\forall v
        \label{phaseLP_nonnegstart}
        \\
        \jobmachine{v}{i} &\in (0,1) &\forall v,i
        \label{phaseLP_jobmachine}
        \\
        \machinephase{u}{v}{i} &\in (0,1) &\forall u,v,i: u \prec v
        \label{phaseLP_samephase}
  \end{flalign}
\end{minipage}
\vspace{4mm}

We give an intuitive interpretation of the variables and constraints.  We interpret the variables $\jobmachine{v}{i}$ as giving the ``primary" placement of $v$ and $\start{v}$ as the corresponding start time of $v$.  Then, for any jobs $u$ and $v$ such that $u \prec v$ and for any machine $i$, we can understand the variable $\machinephase{u}{v}{i}$ as indicating, first, whether $v$ is executed on a machine indexed $i$ or lower, and second, whether the start time of $u$ is within $\delay$ of the start time of $v$. The significance of this indication is that, if these conditions are met, then some \textit{copy} of $u$ must execute on the same machine as $v$ within $\delay$ time of $v$ and, therefore, only predecessors of total size at most $\delay\speed{i}$  can meet these conditions.  The remaining variable $C$ captures the makespan of the resulting schedule. 

\junk{
In the duplication setting, we can interpret the values of $\jobmachine{v}{i}$ as giving the \textit{primary} placement of $v$. That is, in order to finish $v$ as early as possible, we should place it on a machine where it has high mass according to $\jobmachine{v}{i}$. Now, suppose we decide to schedule $v$ on machine $i$ and consider the $\delay$ length time interval before the start of $v$. Since there is no communication during this time, we can select some subset of $v$'s predecessors with total size up to $\delay \speed{i}$ to place on $i$ in this interval.  We can then think of $\machinephase{u}{v}{i}$ as giving this \textit{secondary} placement of $u$. That is, if we schedule $v$ on machine $i$ and $\machinephase{u}{v}{i}$ is large then, in order to finish $v$ as early as possible, we should place a \textit{copy} of $u$ in the same phase as $v$ on $i$. 
\david{see Zoya's comment on this explanation in Section 4}
}

Constraint~(\ref{phaseLP_completionlb}) states that the makespan should be at least as the amount of time to execute any job after its start time. Constraint~(\ref{phaseLP_execution}) states that a job should start after the completion of its predecessor.
The delay constraint~(\ref{phaseLP_delay}) states that if $v$ is scheduled on a machine slower than $i$, then $v$ should start at least $\delay$ time after any predecessor $u$ unless $u$ is scheduled in the same phase as $v$.  The phase constraint~(\ref{phaseLP_phasepredecessors}) states that if $v$ is scheduled on a machine slower than $i$, then the total size of $v$'s predecessors scheduled in the same phase is at most $\delay\speed{i}$. 
Constraint~(\ref{phaseLP_loadlb}) states that the makespan should be at least as large as the total load on any machine. Constraint~(\ref{phaseLP_finishjob}) states that each job should be completely scheduled. 

\begin{lemma}[{\bf LP is a valid relaxation}]
    For any instance $(G,M,\delay)$ for which the optimal makespan is $\optmakespan$, the value of the optimal solution to LP is at most  $2\optmakespan$. 
\label{lem:Gen:Relaxation}
\end{lemma}

\begin{proof}
    Let $\sigma$ be a schedule of $(G,M,\delay)$ with makespan $C'$. We construct a solution to LP with objective value $2C'$. We first construct the schedule $\sigma'$ as follows. We define \textit{phase} $\tau$ to be the interval of time $[\tau\delay,(\tau+1)\delay)$. For any job-machine pair $(v,i)$ such that $v$ is scheduled on $i$ and $\sigma(v,i) \in [\tau\delay, (\tau+1)\delay)$, we set $\sigma'(v,i) = \sigma(v,i) + \tau\delay$. Note that the makespan of $\sigma'$ is at most $2C'$.  
    We show that $\sigma'$ is a valid schedule of the instance $(G,M,\delay)$. To show that the precedence and communication requirements are satisfied, suppose job $v$ is scheduled on machine $i$ and job $u$ is scheduled on machine $j$. By definition of the phases, we have $\sigma'(v,i) - \sigma'(u,j) \ge (\sigma(v,i) - \sigma(u,j) + \delay \cdot \floor{\sigma(v,i) - \sigma(u,j)/\delay} \ge \sigma(v,i) - \sigma(u,j)$. So, the time between two executions in $\sigma'$ is at least the between the same executions in $\sigma$. This entails that both precedence and communication requirements are satisfied in $\sigma'$. The other requirements are easy to check. Therefore, $\sigma'$ is a valid schedule.
    
    Given $\sigma'$, we set the variables of LP as follows. We assume, without loss of generality, that each job is executed at most once on each machine. For each $v$, let $i^*$ be some machine on which $\sigma'$ first completes $v$, choosing arbitrarily if there is more than one. Set $\jobmachine{v}{i^*} = 1$ and $\jobmachine{v}{i} = 0$ for $i \ne i^*$. Let $t^*$ be the start time of $v$ on $i^*$ in $\sigma'$ and set $\start{v} = t^*$. For all $u \prec v$ and $i \ge i^*$, set $\machinephase{u}{v}{i} = 1$ if $\start{v} - \start{u} \le \delay$, and $\machinephase{u}{v}{i} = 0$ otherwise. Set $C = 2C'$.
    
    We now show that our assignment of values to each variable satisfies all constraints of the linear program.
    It is easy to verify that constraints (\ref{phaseLP_completionlb}), (\ref{phaseLP_execution}), (\ref{phaseLP_loadlb}) - (\ref{phaseLP_samephase}) are satisfied. 
    For Constraint~(\ref{phaseLP_delay}), let us fix $u,v$ and $i$, and let $\sum_{j \le i} \jobmachine{v}{j} = \alpha$ and $\machinephase{u}{v}{i} = \beta$. Then there are four cases to verify: (a) $\alpha =\beta = 1$, (b) $\alpha= 1$ and $\beta = 0$, (c) $\alpha = 0$ and $\beta = 1$, or (d) $\alpha = \beta = 0$. In cases (a), (c), and (d) it is easy to see that the constraint is satisfied by the fact that $\start{v} \ge \start{u}$. In case (b), we see that $\machinephase{u}{v}{i}$ is set to 0 only if $\start{v} - \start{u} > \delay$, which entails that constraint is satisfied, or if $i < i^*$, which entails $\sum_{j \le i} \jobmachine{v}{j} = 0$.
    
    \junk{
    To show that  Constraint~(\ref{phaseLP_phasepredecessors}) is satisfied, fix $v$ and $i$ and let $A$ be the set of jobs $u \in \predex{v}$ for which $\machinephase{u}{v}{i} = 1$. By our setting of $\machinephase{u}{v}{i}$, if $\sum_{j \le i} \jobmachine{v}{i} = 0$, then $A = \varnothing$ and the constraint is trivially satisfied. So assume that $\sum_{j \le i} \jobmachine{v}{i} = 1$. This implies that $v$ is first completed on a machine $j$ no faster than $\speed{i}$. Since all jobs in $A$ start their first completed execution less than $\delay$ time before the start of $v$, they all have to be executed on $j$ in $\sigma$, implying that the total size of $A$ is at most $\delay\speed{i}$.
    \zoya{Counterexample to constraint (4). Assume $\rho=10$. $u$, with $p_u=20$, starts on machine $i$, which has speed 1, at time 0 and finishes at time 20. It also starts on a faster machine j, which has speed 20, at time 15 and finishes at time 16. $v$ starts on $i$ at time 20. Then $S_v=20$, $S_u=15$, $z_{u,v,i}=1$. The constraint says $1 \ge 20 / 10$, which is false.} 
    }
    
    To show that  Constraint~(\ref{phaseLP_phasepredecessors}) is satisfied, fix $v$ and $i$ and let $A$ be the set of jobs $u \in \predex{v}$ for which $\machinephase{u}{v}{i} = 1$. By our setting of $\machinephase{u}{v}{i}$, if $\sum_{j \le i} \jobmachine{v}{j} = 0$, then $A = \varnothing$ and the constraint is trivially satisfied. So assume that $\sum_{j \le i} \jobmachine{v}{j} = 1$. This implies that $v$ is first completed on a machine $j$ no faster than $\speed{i}$. Since all jobs in $A$ start their first completed execution less than $\delay$ time before the start of $v$, they all must have some copy executed on $j$ that also completes less than $\delay$ time before the start of $v$. In $\sigma$, let $u$ be some job that completes on $j$ less than $\delay$ time before $\sigma(v,j)$ and starts on $j$ more than $\delay$ time before $\sigma(v,j)$. In this case, $u$ starts in a lower phase than $v$. Therefore, by our construction of $\sigma'$, there is a gap of at least $\delay$ in $\sigma'$ between the completion of $u$ and the start of $v$ during which no jobs are executed. Therefore, in $\sigma'$, all jobs in the set $A$ both start and complete on $j$ less than $\delay$ time before $\sigma'(v,j)$. Therefore, $\setproc{A} \le \delay\speed{i}$, from which (\ref{phaseLP_phasepredecessors}) follows. 
\end{proof}
    
    
    \subsection{Assigning Jobs to Groups}
    \label{sec:group_assign}
    In this section, we present the third subroutine of our algorithm, which takes as input a given instance of DAG Scheduling with Communication Delay, as well as a fractional assignment $\{\jobmachine{v}{i}\}_{v,i}$ of jobs to machines,  and returns an integral assignment of jobs to groups of machines. 
Per Lemma~\ref{lem:elim_slow}, suppose that $\speed{i} \ge \speed{m}/m$ for each $i \in M$.

Recall that the machines are ordered such that $\speed{i} \le \speed{j}$ if $i < j$. We first partition the set of machines $M$ into groups $\group{1}, \group{2}, \ldots, \group{K}$. We define the groups iteratively, with group $\group{1} = \{i : \speed{i} \in [\speed{1}, 2\speed{1})\}$ and group $\group{k+1} = \{i : \speed{i} \in [\speed{j},2\speed{j}) \text{ where } j = \argmin_{j'}\{j' \not\in \bigcup_{k'=1}^k \group{k'}\}\}$. 
Since the speeds of all machines are within a factor of $m$, we have that the number of groups $K$ is at most $\log m$. 

\begin{definition}
    For each job $v$, we define the group $\jobtogroup{v}$ of $v$ as follows.
    Let $\med{v} = \min \{ k : \sum_{k' \le k} \sum_{i \in \group{k}} \jobmachine{v}{i} \ge 1/2 \}$.
    Then, for any $v$, 
    \[ \jobtogroup{v} = \arg\max_k \{ \groupsize{k} \groupspeed{k} : k \ge \med{v} \} . \]
    We also define $\grouptojobs{k} = \{ v : \jobtogroup{v} = k \}$ and $\fastergroupstojobs{k} = \{ v : \jobtogroup{v} \ge k \}$. 
    \label{def:group}
\end{definition}

Informally, $\jobtogroup{v}$ is the highest capacity group such that at least half of job $v$ is assigned to groups no faster than the slowest machine in $\group{k}$.
In the remainder of the section, we prove several properties of this assignment. Toward this end, we consider a set of start times $\{\start{v}\}_{v}$ and makespan $\{\LPmakespan\}$ such that  $\{\start{v}\}_{v} \cup \{\LPmakespan\} \cup \{\jobmachine{v}{i}\}_{v,i}$ is a feasible solution to LP for the given problem instance. We partition $V$ into sets called \textit{bands} where the elements $v$ of each band are determined by the value $\start{v}$.

\begin{definition}
    Band $\band{r}= \{v: \delay(r-1)/4 \leq \start{v} < \delay r/4\}$. 
    \label{def:band}
\end{definition}

In the following lemma, we show that, for any job $v$, the total size of $v$'s predecessors that occupy the same band is no more than the amount that can be completed in a constant number of communication phases on some machine in $v$'s assigned group.

\begin{lemma}[{\bf Upper bound on total size of predecessors of a job in a band}]
    For any job $v \in \band{r}$, we have $\setproc{\band{r} \cap \predex{v}} \le 8\delay \groupspeed{\jobtogroup{v}}$.
\label{lem:Gen:HeightBound}
\end{lemma}

\begin{proof}
    We fix $r$ and choose any $v \in \band{r}$. Let $A = \band{r} \cap \predex{v}$ and let $i^* = \argmax_{i \in \med{v}}\{\speed{i}\}$. By definition of the groups, if $i^* \in \group{k}$ then $k \le \jobtogroup{v}$. In this case, $\speed{i^*} \le 2 \groupspeed{\jobtogroup{v}}$ and it is sufficient to show that $\setproc{A} \le 4\delay\speed{i^*}$.
    For any $u \in A$,
    \begin{align*}
        \start{v} &\ge \start{u} + \delay \Big( \sum_{j \le i^*} \jobmachine{v}{j} - \machinephase{u}{v}{i^*} \Big) &\text{by (\ref{phaseLP_delay})}
        \\
        &\ge \start{u} + \delay ( \sfrac{1}{2} - \machinephase{u}{v}{i^*} ) &\text{by definition of } i^*.
    \end{align*}
    By definition of $\band{r}$ we have that $\start{v} \le \start{u} + \delay/4$, so $\machinephase{u}{v}{i^*} \ge \sfrac{1}{4}$ for all $u \in A$. 
    Suppose for the sake of contradiction that $\setproc{A} > 4 \delay\speed{i^*}$. Then
    \begin{align*}
        \sum_{u \in \predex{v}} \proc{u} \machinephase{u}{v}{i^*} &\ge \sum_{u \in A} \proc{u} \machinephase{u}{v}{i^*} \ge  \setproc{A}/4 > \delay \speed{i^*} &\text{by assumption}
        \\
        &= \delay \speed{i^*} \sum_j \jobmachine{v}{j} \ge \delay \speed{i^*} \sum_{j \le i^*} \jobmachine{v}{j} &\text{by (\ref{phaseLP_finishjob}).}
    \end{align*}
    Therefore, Constraint~(\ref{phaseLP_phasepredecessors}) is violated.
\end{proof}

Although we can upper-bound the size of a job's predecessor set within a band, there is no limit on the total amount of work that can be put on a single band. For example, if all jobs are independent, then the solution to LP might place all jobs on the first band. The following lemma, due to \cite{chudak1999approximation}, addresses this issue by providing a lower bound on the optimal linear programming solution in terms of the total load across all groups. We repeat the lemma and provide a proof for our linear relaxation.

\begin{lemma}[{\bf \cite{chudak1999approximation} Lower bound on $\pmb{C^*}$ in terms of load}]
    $\displaystyle \quad \sum_{k=1}^K \frac{\setproc{\grouptojobs{k}}}{\groupsize{k}\groupspeed{k}} = O(KC^*).$
    \label{lem:LoadBound}
\end{lemma}

\begin{proof}
     Let $y_{v,k} = 1$ if $\jobtogroup{v} = k$ and 0 otherwise, for all $v \in V$ and $k = 1,\ldots,K$. Then, by definition of the group assignment, the set of all variables $y_{v,k}$ gives an optimal solution to the following linear program. 
    \begin{align}
        \min ~ \sum_{k=1}^K \sum_{v \in V} \frac{ \proc{v} y_{v,k}}{\groupsize{k}\groupspeed{k}}
        \notag
        \\
        \text{subject to} \quad \sum_{k=1}^K y_{v,k} &= 1 &\text{for all}~ v \in V, k = 1,\ldots,K
        \label{simplex:cover}
        \\
        y_{v,k} &= 0 &\text{for all}~ v \in V, k = 1,\ldots, \med{v}-1
        \label{simplex:slow}
        \\
        y_{v,k} &\ge 0 &\text{for all}~ v \in V, k = 1,\ldots,K
        \label{simplex:nonneg}
    \end{align}
    We give a feasible solution to this linear program whose objective values is at most $4K\LPmakespan$, thereby proving the lemma. 
    
    Let $\jobmachine{v}{i}, \start{v}, \comp{v},$ and $\LPmakespan$, for all $v \in V$ and $i \in M$, represent the variables in the the feasible solution to LP. 
    Then, for all jobs $v$ and groups $\group{k}$, we set $y_{v,k}$ as follows.  
     \[ y_{v,k} \leftarrow 
    \begin{cases}
        \frac{\sum_{i\in \group{k}} \jobmachine{v}{i} }{ \sum_{\ell = \med{v}}^K \sum_{i \in \group{\ell}} \jobmachine{v}{i}} &\text{if } k \ge \med{v}
        \\
        0 &\text{otherwise.}
    \end{cases}\]
    It is easy to see that this assignment satisfies all constraints (\ref{simplex:cover}) --  (\ref{simplex:nonneg}). We show that the objective is no more than $4K\LPmakespan$.
    \begin{align*}
        \sum_{k=1}^K ~ \sum_{v} ~ \frac{ \proc{v} y_{v,k}}{\groupsize{k}\groupspeed{k}} 
        &= \sum_{k=1}^K ~ \frac{1}{\groupsize{k}\groupspeed{k}} \sum_{v: \med{v} \le k} \frac{ \proc{v} \sum_{i \in \group{k}} \jobmachine{v}{i}}{\sum_{\ell = \med{v}}^K \sum_{i \in \group{\ell}} \jobmachine{v}{i}} &\text{by assignment of}~ y_{v,k}
        \\
        &\le \sum_{k=1}^K ~ \frac{2}{\sum_{i \in \group{k}} \speed{i}} \sum_{v: \med{v} \le k} \frac{\proc{v} \sum_{i \in \group{k}} \jobmachine{v}{i}}{\sum_{\ell = \med{v}}^K \sum_{i \in \group{\ell}} \jobmachine{v}{i}} &\text{by definition of } \groupsize{k} \groupspeed{k}
        \\
        &\le \sum_{k=1}^K ~ \frac{2C}{\sum_{i \in \group{k}} \sum_v \proc{v} \jobmachine{v}{i}} \sum_{v: \med{v} \le k} \frac{\proc{v} \sum_{i \in \group{k}} \jobmachine{v}{i}}{\sum_{\ell = \med{v}}^K \sum_{i \in \group{\ell}} \jobmachine{v}{i}} &\text{by Constraint~(\ref{phaseLP_loadlb})}
        \\
        &\le \sum_{k=1}^K ~ \frac{2C}{\sum_{i \in \group{k}} \sum_v \proc{v} \jobmachine{v}{i}} \sum_{v: \med{v} \le k} 2 \proc{v} \sum_{i \in \group{k}} \jobmachine{v}{i}  &\text{by $\med{v}$}
        \\
        &\le \sum_{k=1}^K ~ \frac{2C}{\sum_{i \in \group{k}}  \sum_v \proc{v} \jobmachine{v}{i}} \cdot \sum_{v} 2 \proc{v} \sum_{i \in \group{k}} \jobmachine{v}{i} = \sum_{k=1}^K 4C = 4KC.
    \end{align*}
    The lemma follows.
\end{proof}

We note that an alternative group construction is given in \cite{Li17} where the author reduces the approximation ratio of the algorithm in \cite{chudak1999approximation} from $O(\log m)$ to $O(\log m /\log\log m)$. 
However, given our analysis, we have been unable to apply the method of \cite{Li17} to similarly reduce our approximation ratio.
While the speeds of all machines within a group in our algorithm are within a factor of 2 of each other, the construction of \cite{Li17} allows for these speeds to be off by up to a factor of $\log m/\log\log m$. 
Use of this alternative construction incurs an additional factor of $O(\log m /\log\log m)$ for the upper bound proved in Lemma~\ref{lem:Gen:HeightBound} on the processing time of a job's predecessors in a band, resulting in a worse approximation ratio than our current bound. 
    
    \subsection{Computing the Schedule}
    \label{sec:scheduler}
    In this section, we describe the fourth subroutine of our algorithm that outputs a schedule $\sigma$ for a given instance $(G,M,\delay)$ and a given assignment of each job $v$ to a group $\jobtogroup{v}$ of machines in $M$. 
Our algorithm builds on ideas used in the scheduling algorithms given in \cite{LR02} and \cite{chudak1999approximation}. We apply a similar approach as that used in \cite{LR02} to decide when to place a job (and all its remaining predecessors). However, unlike \cite{LR02}, our algorithm accounts for arbitrarily slow machines and arbitrarily large jobs, both of which may cause a job's execution to last well beyond after the time of its placement. As in \cite{chudak1999approximation}, we use the group assignment to guide where jobs are scheduled. One fundamental difference between our algorithm and that in \cite{chudak1999approximation} is that, given certain conditions, we allow a job to be scheduled on machines in groups to which the job is not assigned. The subroutine in given in Algorithm~\ref{alg:scheduler}.



\begin{algorithm}
\KwData{instance $(G,M,\delay)$, assignment $\assignfunc$ of jobs to groups, and overlap parameter $\overlap \ge 1$}
\KwResult{a schedule $\sigma$ of $G$ on $M$}
    {\bf Initialize:} $T \leftarrow 0$; \ $\Placed \leftarrow \varnothing$; \ $\forall j: T_j \leftarrow 0$; \ $\forall v,i:\sigma(v,i) = \infty$ \label{line:init}\;
    \While{$\Placed \ne V$ \label{line:job_loop}}{
        \ForAll{machine groups $k = 1,\ldots,K$}{
            \ForAll{jobs $v \in \grouptojobs{k}$ \label{line:assigned_jobs_loop}}{
                $i \leftarrow \argmin_{j \in \group{k}} \{T_j\}$ \label{line:minmachine}\;
                $A \leftarrow (\predex{v} \cup \{v\}) \setminus \{u: \sigma(u,i) + \proc{u}/\speed{i} \le T_{i} \text{ or } \exists\ j,\ \sigma(u,j) + \proc{u}/\speed{j} \le T_{i}-\delay\}$\label{line:A_def}\;
                \If{\begin{enumerate*}[label={\bf\small(\alph*)}]
                    \item $\setproc{A \setminus \{v\}} \le 8\delay\groupspeed{k}$ and \label{condition:predecessor_size}
                    \item $\setproc{A \setminus \Placed} \ge \setproc{A}/\overlap$ and \label{condition:overlap}
                    \item $A \subseteq \fastergroupstojobs{k}$  \label{condition:predecessor_groups}
                  \end{enumerate*}\label{line:conditions}}{
                    \ForAll{$u \in A$ in topological order \label{line:ancestor_loop}}{
                        $\sigma(u,i) \leftarrow T_{i}$ \label{line:start_time}\;
                        $T_{i} \leftarrow T_{i} + \proc{u}/\speed{i}$ \label{line:update_Ti}\;
                    }
                    $\Placed \leftarrow \Placed \cup A$ \label{line:update_placed}\;
                }
            }
        }
        $T \leftarrow \min \{t: t > T \text{ and either } t = \sigma(v,i) + \proc{v}/\speed{i} \text{ or } t = \delay + \sigma(v,i) + \proc{v}/\speed{i} \text{ for some } v,i \}$ \label{line:updateT}\;
        $\forall j:\ T_j \leftarrow \max\{ T, T_j \}$ \label{line:skipahead}\;
    }
\caption{Group-Based Scheduling with Duplication and Communication Delay}
\label{alg:scheduler}
\end{algorithm}

\junk{
For a given iteration $t$ of line~\ref{line:assigned_jobs_loop}, we use $i\iter{t}$ to denote the value of $i^*$ in iteration $t$ of Algorithm~\ref{alg:scheduler}. Otherwise, we use $\langle \mathsf{term} \rangle \iter{t}$ to denote the value of $\langle \mathsf{term} \rangle$ in iteration $t$ of Algorithm~\ref{alg:scheduler}, where $\langle \mathsf{term} \rangle$ may be $v,A,k,T,\Placed$ or $T_i$ for a given $i$. $\Placed\iter{t}$ is the value of $\Placed$ in iteration $t$ \textit{before} updating it at line~\ref{line:update_placed}. We also use $\sigma$ to denote the schedule output by Algorithm~\ref{alg:scheduler}. We call any non-zero length interval $(g_1,g_2)$ in $\sigma$ during which a machine $i$ is not executing any jobs a \textit{gap} on $i$.
}

We introduce notation to keep track of how the algorithm variables change during the execution.  We refer to each execution of a line of Algorithm~\ref{alg:scheduler} as a \emph{step} of the algorithm. 
For a given step $\step$ of Algorithm~\ref{alg:scheduler}, we use $\langle \mathsf{term} \rangle \superstep{}$ to denote the value of $\langle \mathsf{term} \rangle$ immediately before executing step $\step$, where $\langle \mathsf{term} \rangle$ may be $v,i,A,k,T,\Placed$, or $T_i$ for a given $i$. We also use $\sigma$ to denote the schedule output by Algorithm~\ref{alg:scheduler}. 
Define
    \[\mathcal{T} = \{0\} \cup \{t : \exists\ (v,i), \sigma(v,i) \ne \infty \text{ and } t = \sigma(v,i) + \proc{v}/\speed{i} \text{ or } t = \delay + \sigma(v,i) + \proc{v}/\speed{i} \}\]
    and $V(j,\step)$ to be the set of jobs that have been placed on machine $j$ prior to executing step $\step$.  The following two lemmas characterize how $T_i$ and $T$ vary, respectively, during the execution of the algorithm.  

\begin{lemma}[{\bf $\pmb{T_i}$ tracks maximum completion time on $\pmb{i}$}]
 For any machine $j$ and any step $\step$ executing line 10 or 13, we have 
        $T_j\superstep{+1} = \max \{ T\superstep{+1}, \max_{v \in V(j,\step)} \{ \sigma(v,j) + \proc{v}/\speed{j} \}\}$.
      \label{lem:Ti_in_algo}
\end{lemma}
\begin{proof}
       Let $\step_0 ,\step_1, \step_2, \ldots$ be the sequence of steps defined as follows: $\step_0$ is first step of the algorithm (executing line\ref{line:init}) and $\step_{d+1}$ is the $d$\textsuperscript{th} step, in order, in which either line~\ref{line:update_Ti} or \ref{line:skipahead} is executed. We show that, for any $j$ and $d$, $T_i\superstep{_d+1} = \max \{ T\superstep{_d+1}, \max_{v \in V(j,\step_d)} \{ \sigma(v,j) + \proc{v}/\speed{j} \}\}$.  The proof is by induction on $d$.  The equality holds easily for $d=0$ by the initialization step. So we assume the induction hypothesis equality holds up to $d \ge 0$. We consider two cases. 
    First, suppose that $\step_{d+1}$ executes line~\ref{line:update_Ti}. Then the job $u = u\superstep{_{d+1}}$ was just placed on machine $i\superstep{_{d+1}}$ in line~\ref{line:start_time}.  For $j \neq i\superstep{_{d+1}}$,  neither $T_j$ nor $T$ has changed, and $V(j,\step_{d+1})$ is the same as $V(j,\step_d)$, so the claim follows from the induction hypothesis.  For $j = i\superstep{_{d+1}}$, by the execution of line~\ref{line:start_time} on step $\step_{d+1}-1$ and line~\ref{line:update_Ti} on step $\step_{d+1}$, we have $T_j\superstep{_{d+1}+1} = \sigma(u,j) + \proc{u}/\speed{j}$.  We thus derive
    \begin{align*}
      T_j\superstep{_{d+1}+1} & =  \sigma(u,j) + \proc{u}/\speed{j} & \text{lines~\ref{line:start_time} and~\ref{line:update_Ti}}\\
      & \le  \max_{v \in V(\step_{d+1},j)} \{\sigma(v,j) + \proc{v}/\speed{j}\} &
      \text{$u \in V(\step_{d+1}+1,j)$}\\
      & =  \max\{\sigma(u,j) + \proc{u}/\speed{j}, \max_{v \in V(\step_{d},j)} \{\sigma(v,j) + \proc{v}/\speed{j}\}\} & \text{$V(i,\step_{d+1}) = V(i,\step_d) \cup \{u\}$}\\
      & \le  \max\{\sigma(u,j) + \proc{u}/\speed{j}, T_j\superstep{_d+1}\} & \text{induction hypothesis}\\
      & = \max\{\sigma(u,j) + \proc{u}/\speed{j}, T_j\superstep{_{d+1}}\} & \text{$T_j\superstep{_{d+1}} = T_j\superstep{_d+1}$} \\
      & = \max\{\sigma(u,j) + \proc{u}/\speed{j}, \sigma(u,j)\} & \text{line~\ref{line:start_time}}\\
      & = \sigma(u,j) + \proc{u}/\speed{j} = T_j\superstep{_{d+1}+1}.
      \end{align*}
      It thus follows that every inequality in the above sequence is, in fact, an equality yielding the equation
      \[
      T_j\superstep{_{d+1}+1} = \max_{v \in V(\step_{d+1},j)} \{\sigma(v,j) + \proc{v}/\speed{j}\}.
      \]
    We also have
    \[
    T_j\superstep{_{d+1}+1} \ge T_j\superstep{_{d+1}} = T_j\superstep{_d+1} \ge T\superstep{_d+1} = T\superstep{_{d+1}+1},
    \]
    yielding the desired claim for induction step:
    \[
    T_j\superstep{_{d+1}+1} = \max\{T\superstep{_{d+1}+1}, \max_{v \in V(\step_{d+1},j)} \{\sigma(v,j) + \proc{v}/\speed{j}\}\}.
    \]
    \junk{
    \begin{align*}
        T_i\superstep{_{d+1}+1} &= T_i\superstep{_d+1} + \proc{u}/\speed{i} = \max\{ T\superstep{_d+1}, \max_{v \in V(i,\step_d)} \{ \sigma(v,i) + \proc{v}/\speed{i}\} \} + \proc{u}/\speed{i}  
        \\
        &= \max\{ T\superstep{_{d+1}}, \max_{v \in V(i,\step_{d+1})} \{ \sigma(v,i) + \proc{v}/\speed{i}\} \}.
    \end{align*}
    where the last equality follows from the fact that $T\superstep{_d+1} = T\superstep{_{d+1}}$ when $\step_{d+1}$ executes line~\ref{line:update_Ti} and the fact that $V(i,\step_{d+1}) = V(i,\step_d) \cup \{u\}$ where $\sigma(u,i) = T_i\superstep{_{d+1}}$ by line~\ref{line:start_time}.}
    
    We next consider the second case of the lemma, in which $\step_{d+1}$ executes line~\ref{line:skipahead}.  Fix any machine $j$.  We derive
    \begin{align*}
        T_j\superstep{_{d+1}+1} & = \max\{T\superstep{_{d+1}}, T_j\superstep{_{d+1}}\} & \text{line~\ref{line:skipahead}}\\
        & = \max\{T\superstep{_{d+1}+1}, T_j\superstep{_{d+1}}\}  & \text{$T\superstep{_{d+1}+1} = T\superstep{_{d+1}}$}\\
        & = \max\{T\superstep{_{d+1}+1}, 
        T_j\superstep{_d+1}\} & \text{$T\superstep{_{d+1}} = T_j\superstep{_d+1}$}\\
        & = \max\{T\superstep{_{d+1}+1},
        \max_{v \in V(j,\step_d)}\{\sigma(v,j) + \proc{v}/\speed{j}\}\} & \text{induction hypothesis}\\
        & = \max\{T\superstep{_{d+1}+1},
        \max_{v \in V(j,\step_{d+1})}\{\sigma(v,j) + \proc{v}/\speed{j}\}\},
    \end{align*}
    thus completion the induction step for the second case of the lemma.
\end{proof}

\begin{lemma}[{\bf $\pmb{T}$ iterates through $\pmb{\mathcal{T}}$}]
For any integer $1 \le \ell \le |\mathcal{T}|$, if step $\step$ is the $\ell$th execution of line~\ref{line:updateT}, then $T\superstep{}$ is the $\ell$th minimum element of $\mathcal{T}$.
\label{lem:T_in_algo}
\end{lemma}
\begin{proof}
  Our proof is by induction on $\ell$.  Let $\step_\ell$ denote the $\ell$th execution of line~\ref{line:updateT}.  For the base case $\ell = 1$, we note that $T\superstep{_1} = 0$, which is the smallest element of $\mathcal{T}$.  For the induction hypothesis, suppose the claim holds up to $\ell \ge 1$.  Consider the claim for $(\ell+1)$th execution of line~\ref{line:updateT}.  Since $T$ is only updated in line~\ref{line:updateT}, $T\superstep{_{\ell+1}} = T\superstep{_\ell+1}$.    
        
        By Lemma~\ref{lem:Ti_in_algo} and the setting of $\sigma$ in line~\ref{line:start_time}, for every job $v$ placed after the $(\ell-1)$th execution of line~\ref{line:updateT} and before the $\ell$th execution of line~\ref{line:updateT}, there exists an $i$ such that $\sigma(v,i) + \proc{v}/\speed{i}$ is finite and
        exceeds $T\superstep{_\ell}$, which by induction is the $\ell$th minimum element of $\mathcal{T}$.  So, if there is any job $v$ such that $\sigma(v,i) + \proc{v}/\speed{i} + \rho$ is finite and
        exceeds $T\superstep{_\ell}$, then $T\superstep{_\ell+1}$ equals the $(\ell+1)$th minimum element of $\mathcal{T}$ since every subsequent job placed has its start time at least $T\superstep{_\ell+1}$. 
        
        It remains to show that there is at least one such job.  For the sake of contradiction, suppose not.  Then, it must be the case that for every $v$ already placed, $\sigma(v,i) + \proc{v}/\speed{i} + \rho$ is at most $T\superstep{_\ell}$, and no job was placed after the $(\ell-1)$th execution of line~\ref{line:updateT} and before the $\ell$th execution of line~\ref{line:updateT}.
        If $\Placed$ equals $V$ before the $\ell$th execution of line~\ref{line:job_loop}, then there is nothing to prove since then there is no $\ell$th execution of line~\ref{line:updateT}.  Otherwise, there is a job $v$ not placed, all of whose ancestors already have been placed.  Then, in an execution of line~\ref{line:A_def} between the $(\ell-1)$th and $\ell$th executions of line~\ref{line:updateT}, $A$ is set to $\{v\}$, ensuring that all conditions of line~\ref{line:conditions} are satisfied and $v$ is placed, yielding a contradiction.  
\end{proof}

\begin{lemma}[{\bf Correctness and Runtime}]
    Algorithm \ref{alg:scheduler} outputs a valid schedule in $\poly(n,m)$ time. 
\end{lemma}

\begin{proof}
    We show that precedence and communication delay requirements are satisfied. It is straightforward to establish that remaining requirements (given in Section~\ref{sec:prelim}) are met. 
    We show first that precedence requirements are met. This follows from the fact that when any job is placed on a machine, all of its uncompleted predecessors are placed in topological order on the same machine. Therefore, for any job $v$ with predecessor $u$, the start time of $u$ is less than the start time of $v$. 
    We now show that communication delay requirements are met. This follows from the fact that when Algorithm~\ref{alg:scheduler} places a set of jobs on a machine $i$ starting at time $t$, it duplicates any jobs that have not completed previously on $i$ or have not completed by time $t-\delay$ on any other machine.
    
    We now show that Algorithm~\ref{alg:scheduler} runs in time polynomial in $n$ and $m$.  We first argue that there are at most $2m(n-1)+2$ executions of line~\ref{line:job_loop}.  For the sake of contradiction, suppose there is an $(2m(n-1)+3)$th execution of line~\ref{line:job_loop}.  By Lemma~\ref{lem:T_in_algo}, if $\step$ represents the $(2m(n-1)+1)$th execution of line~\ref{line:updateT}, then $T\superstep{}$ equals the $(2m(n-1)+1)$ smallest element of $\mathcal{T}$.  Since every job $v$ has at most $2m$ finite values of the form $\sigma(v,i) + \proc{v}/\speed{i}$ or $\sigma(v.i) + \proc{v}/\speed{i} + \rho$ in $\mathcal{T}$, it follows from the pigeon hole principle that every job $v$ has at least one finite value of the form  $\sigma(v,i) + \proc{v}/\speed{i}$ that is at most the $(2m(n-1)+1)$ smallest element of $\mathcal{T}$, indicating that every job has been placed before the $(2m(n-1)+1)$th execution of line~\ref{line:updateT}.  This ensures that $\Placed$ equals $V$ no later than the $(2m(n-1)+2)$th execution of line~\ref{line:job_loop}, after which the algorithm terminates, leading to a contradiction.  We thus obtain that line~\ref{line:job_loop} iterates at most $O(nm)$ times. All other loops iterate at most $n$ or $m$ times, and all other instructions take at most time $nm$.
    \junk{
    We show that there is some job placed during every $2mn$ iterations of line~\ref{line:job_loop}. Suppose the algorithm executes $2mn-1$ iterations of line~\ref{line:job_loop} without placing any job on any machine. We observe that each job can be duplicated at most $m$ times, and each job completion yields two values of $T$ by the update at line~\ref{line:updateT}. So, in iteration $2mn$ we have that  $\forall i, T_i = \max_{v,j:\sigma(v,j)\ne\infty} \{\delay + \sigma(v,j) + \proc{v}/\speed{j}\}$. So there is some job $v$ such that all its predecessors have completed at least $\delay$ time before the value of $T_i$ in this iteration, for all machines $i$. Therefore, $v$ is scheduled in this iteration. \zoya{should show that v satisfies all other conditions to be scheduled} Since each job can be duplicated up to $m$ times, we have that line~\ref{line:job_loop} iterates at most $O((nm)^2)$ times. All other loops iterate at most $n$ or $m$ times, and all other instructions take at most time $nm$.}
\end{proof}

\junk{
\begin{lemma}[]
    Let $u \prec v$, let $b$ be the earliest completion time of $u$ in $\sigma$, and let $c$ be the earliest completion time of $v$. If there is some gap $(g_1,g_2)$ on some machine in $\group{\jobtogroup{v}}$ such that, for some time $a$,  $g_1 \le a + \delay$ and $g_2 \ge a$ and $b \le a - \delay$ and $c \ge a + \delay$, then there is some iteration $t$ of line~\ref{line:assigned_jobs_loop} such that $v\iter{t} = v$, $i\iter{t} \in\group{\jobtogroup{v}}$, $T_{i\iter{t}}\iter{t} < a + \delay$, and $u \not\in A\iter{t}$.
    \david{$c$ should be earliest \textit{start} time of $v$.} \david{draw figure to show these values.}
\end{lemma}

We call any non-zero length interval $(g_1,g_2)$ in $\sigma$ during which a machine $i$ is not executing any jobs a \textit{gap} on $i$.

\begin{lemma}[{\bf Predecessor reduction in gaps}]
    Let $u \prec v$, let $c$ be the earliest completion time of $u$ in $\sigma$, let $b$ be the earliest start time of $v$, and suppose that $b > c+\delay$. If there is some gap $(g_1,g_2)$ on some machine in $\group{\jobtogroup{v}}$ such that $g_1 < b$ and $g_2 \ge c + \delay$, then there is some iteration $t$ of line~\ref{line:assigned_jobs_loop} such that $v\iter{t} = v$, $i\iter{t} \in\group{\jobtogroup{v}}$, $T_{i\iter{t}}\iter{t} = \max\{g_1, c+\delay \}$, and $u \not\in A\iter{t}$.
    Figure~\ref{fig:gap_fill} depicts an instance for which these conditions are met.
\end{lemma}
}

The following lemma establishes that, for $u \prec v$ and a gap that occurs before $v$ and sufficiently long after $u$, there are steps in the algorithm that consider placing $v$ and its remaining predecessors $A$ to start inside the gap such that $u \not\in A$. This lemma is used in the proof of Lemma~\ref{lem:heightphase_bound}.

\begin{lemma}[{\bf Predecessor removal in gaps}]
    Let $u \prec v$, let $c$ be the earliest completion time of $u$ in $\sigma$, let $b$ be the earliest start time of $v$ in $\sigma$, and suppose that $b > c+\delay$. If there is some gap $(g_1,g_2)$ on any machine in $\group{\jobtogroup{v}}$ such that $g_1 < b$ and $g_2 \ge c + \delay$, then there is step $\step$ checking the conditions of line~\ref{line:conditions} such that $v\superstep{} = v$, $i\superstep{} \in\group{\jobtogroup{v}}$, $T_{i\superstep{}}\superstep{} = \max\{g_1, c+\delay \}$, and $u \not\in A\superstep{}$.
    \label{lem:predecessor_removal}
\end{lemma}

Figure~\ref{fig:gap_fill} depicts an instance for which the conditions of Lemma~\ref{lem:predecessor_removal} are met. 

\begin{figure}
    \centering
    \includegraphics{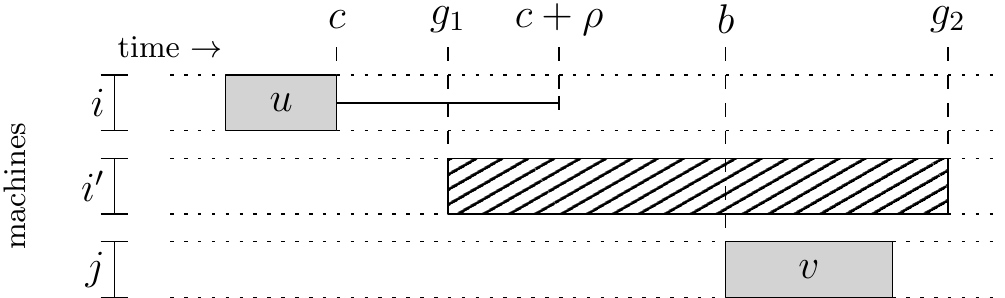}
    \caption{A depiction of jobs $u$ and $v$ for which the conditions of Lemma~\ref{lem:predecessor_removal} are met, given that $i' \in \group{\jobtogroup{v}}$. Machines are listed vertically on the left and time proceeds from left to right. The jobs $u$ and $v$ are shown as gray boxes, and the gap $(g_1,g_2)$ is shown as a diagonal striped box.}
    \label{fig:gap_fill}
\end{figure}

\begin{proof}
    Suppose there are jobs $u$ and $v$ and a gap $(g_1,g_2)$ on machine $j$ satisfying the conditions of the lemma. 
    Let $c_{\max} = \max\{c+\delay,g_1\}$ and let $b_{\min} = \min\{g_2,b\}$. By definition of a gap we have that $(c_{\max}, b_{\min})$ is a gap. Also, since $c$ and $g_1$ is both completion times of some jobs, $c_{\max}$ is in $\mathcal{T}$.  Therefore, by Lemma~\ref{lem:T_in_algo}, there is a step $\step$ executing line~\ref{line:updateT} , immediately before which $T$ equals $c_{\max}$.  Since $T$ does not change between executions of line~\ref{line:updateT}, and line~\ref{line:assigned_jobs_loop} iterates through all jobs between two consecutive executions of line~\ref{line:updateT}, there exists a step $\step^{*} < \step$ executing line~\ref{line:conditions} for which $T\superstep{^*} = c_{\max}$ and $v\superstep{^*} = v$.
     By the assignment of $i$ at line~\ref{line:minmachine}, we have that $i\superstep{^*} \in \group{\jobtogroup{v}}$. The fact that $(c_{\max},b_{\min})$ is a gap entails that the max completion time of any job placed on $j$ prior to step $\step^*$ is at most $c_{\max}$. Therefore, $T_j\superstep{^*} = T\superstep{^*}$ by Lemma~\ref{lem:Ti_in_algo}, so $T\superstep{^*}_{i\superstep{^*}} = T\superstep{^*}$ by the assignment at line~\ref{line:minmachine}. Finally, $u \not\in A\superstep{^*}$ follows from the assignment of $A$ at line~\ref{line:A_def} and the fact that $c +\delay \le c_{\max} = T\superstep{^*}_{i\superstep{^*}}$.
\end{proof}

\begin{lemma}[{\bf Total load within $\pmb{\overlap}$ factor of total job size}]
    Let $V(i)$ be the set of all jobs with some copy scheduled on machine $i$. Then, for any $k$, 
    \[ \sum_{k' \ge k} \sum_{i \in \group{k'}} \setproc{V(i)} \le \overlap \cdot \setproc{\fastergroupstojobs{k}}.\]
    \label{lem:eta_load_increase}
\end{lemma}

\begin{proof}
    Let $\step_1, \step_2, \ldots$ be the series of steps such that $\step_d$ is the first execution of line~\ref{line:start_time} when placing the set $A\superstep{_d}$ on any machine in $\bigcup_{k'\ge k} \group{k}$.
    Let $V(i,\step)$ be the set of jobs placed on machine $i$ prior to executing step $\step$. We show by induction on $d^*$ that
    \begin{equation}
        \lambda(d^*) \equiv \sum_{d=1}^{d^*} \sum_{k' \ge k} \sum_{i \in \group{k'}} \setproc{V(i,\step_d)} \le \overlap \cdot p\bigg( \bigcup_{d=1}^{d^*} \ \bigcup_{k' \ge k} \ \bigcup_{i \in \group{k'}} V(i,\step_d) \bigg).
        \label{eq:load_eta_factor}
    \end{equation}
    Since $\bigcup_{k' \ge k} \ \bigcup_{i \in \group{k'}} V(i) \subseteq \fastergroupstojobs{k}$ by condition \ref{condition:predecessor_groups} in line~\ref{line:conditions}, inequality (\ref{eq:load_eta_factor}) is sufficient to prove the lemma. Inequality (\ref{eq:load_eta_factor}) holds trivially for $d^* = 1$, so we suppose for induction that it holds up to step $d^* \ge 1$. 
    Then 
    \begin{align*}
        \lambda(d^*+1) &= \lambda(d^*) + \setproc{A\superstep{_{d^*+1}}} &\text{by definition of } \lambda
        \\
        &\le \lambda(d^*) + \overlap \cdot \setproc{A\superstep{_{d^*+1}} \setminus \Placed\superstep{_{d^*+1}}} &\text{by condition \ref{condition:overlap}}
        \\
        &\le \overlap \cdot p\bigg( \bigcup_{d=1}^{d^*} \ \bigcup_{k' \ge k} \ \bigcup_{i \in \group{k'}} V\superstep{_d}(i) \bigg) + \overlap \cdot \setproc{A\superstep{_{d^*+1}} \setminus \Placed\superstep{_{d^*+1}}} &\text{by induction}
        \\
        &\le \overlap \cdot p\bigg( \bigcup_{d=1}^{d^*+1} \ \bigcup_{k' \ge k} \ \bigcup_{i \in \group{k'}} V\superstep{_d}(i) \bigg)
    \end{align*}
    where the last inequality follows from the fact that for all jobs $v \in \Placed\superstep{_{d^*+1}}$, $v$ is either an element of $\bigcup_{d=1}^{d^*} \ \bigcup_{k' \ge k} \ \bigcup_{i \in \group{k'}} V\superstep{_d}(i)$ or placed on machines in groups lower indexed than $k$. 
\end{proof}
    
    \subsection{Analysis}
    \label{sec:analysis}
    In this section, we prove that our algorithm yields an approximation ratio with respect to the optimal makespan $\optmakespan$. We assume that we have a preprocessed instance in which all machines of speed less than $1/m$ have been removed, that each job $v$ has been assigned to a group $\jobtogroup{v}$ via the rounding of Section~\ref{sec:group_assign} based on an optimal solution to the linear program LP of Section~\ref{sec:lp} with makespan $C \le \optmakespan$, and that these assignments have been given as input to Algorithm~\ref{alg:scheduler}.  

Our analysis combines elements of the analysis given in \cite{chudak1999approximation} and \cite{LR02}. 
As in \cite{chudak1999approximation}, we define a chain of jobs whose execution time serves as a lower bound on the optimal length of any schedule. 
The chain also has the following property: for any non-chain job $v$ that precedes a job in the chain and is not executed in parallel with the chain, if $v$ is scheduled on any machine $i$ then $v$'s execution time on $i$ is no longer than $8\delay$. This property allows us to leverage the analysis from \cite{LR02} to show that the time spent in $\sigma$ not executing the chain can be bounded by the highest band number and the load on all groups, both of which give lower bounds on the optimal makespan. 

We define the \textit{chain} $\chain = \{ \link{1}, \link{2} , \ldots, \link{|\chain|} \}$ as follows. Let $D$ be the set of pairs $(v,i)$ such that $\sigma$ schedules some copy of $v$ on $i$. Let $L$ be the set of pairs $(v,i) \in D$ such that $\proc{v}/\speed{i} > 8\delay$. Then 
\begin{align*}
    \link{1} &= \argmax_{(v,i) \in L} \{ \sigma(v,i) + \proc{v}/\speed{i} \}
    \\
    \link{q+1} &= \argmax_{(v,i) \in L} \{\sigma(v,i) + \proc{v}/\speed{i} : \link{q} = (u,j) \text{ and } v \in \predex{u}\}.
\end{align*}
We also define sets of jobs that precede jobs in the chain and execute between jobs in the chain. 
Let $V_{0}$ be the set of all jobs that complete after $\link{1}$, i.e.\ for $\link{1} = (i,v)$, $V_{0}$ is the set of $u \in V$ such that, for some machine $j$ we have $(u,j) \in D$ and $\sigma(u,j) + \proc{u}/\speed{j} \ge \sigma(v,i) + \proc{v}/\speed{i}$.
Then, for $0 < q < |\chain|$, let $V_q$ be the set of predecessors of $\link{q}$ that complete after $\link{q+1}$, i.e.\ for $\link{q} = (v,i)$, $V_{q}$ is the set of all jobs $u$ such that $u = v$ or $u \prec v$ and for some machine $j$, $(u,j) \in D$ and $\sigma(u,j) + \proc{u}/\speed{j} \le \sigma(v,i)$ and $\sigma(u,j) + \proc{u}/\speed{j} \ge \sigma(v',i') + \proc{v'}/\speed{i'}$ where $\link{q+1}=(v', i')$. Finally, let $V_{|\chain|}$ be the set of jobs that precede $\link{|\chain|}$.
Figure \ref{fig:chain_construction} depicts the construction of $\chain$ and each $V_q$.

\begin{figure}
    \centering
    \includegraphics[width=\textwidth]{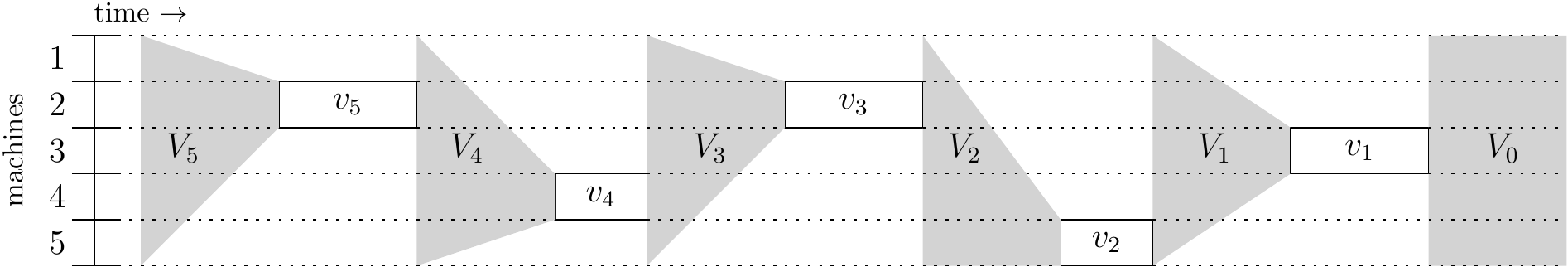}
    \caption{Machines are shown vertically on the left and time increases from left to right. The dotted lines track the jobs being executed on each machine. Jobs $v_1$ through $v_5$ are the only jobs with copies that take longer than $\delay$ time. Other jobs are scheduled but not shown. In this case, $\chain = \{ \link{1} : (v_1, 3), \link{2}:(v_2, 5), \link{3}:(v_3, 2), \link{4}:(v_4,4), \link{5}:(v_5,2)\}$. The sets $V_q$ are shown as shaded regions. 
    All jobs $u \in \predex{v_{q}}$ such that some copy of $v$ is scheduled on $i$ and some completion time of $v$ on $i$ falls inside the shaded trapezoidal region to the left of $v_{q}$ are elements of $V_q$.}
    \label{fig:chain_construction}
\end{figure}

We divide the schedule into phases of length $\delay$ where phase $\tau = [\delay(\tau-1), \delay\tau)$. We say that a machine $i$ is \textit{busy} in phase $\tau$ if the total time spent executing jobs on $i$ in phase $\tau$ is at least $\delay/2$. Otherwise, we say that $i$ is idle in $\tau$.  We classify the phases into three different types. Phase $\tau$ is a \textit{chain phase} if at least $\delay/2$ time is spent working on any job in the chain on any machine. A $\link{q}$-phase is a chain phase in which at least $\delay/2$ time is spent working on $\link{q}$. Phase $\tau$ is a \textit{load phase} if $\tau$ is not a chain phase and every machine in some group is busy in $\tau$. 
The remaining phases are height phases. The different phase types are illustrated in Figure~\ref{fig:phase_types}.

\begin{figure}
    \centering
    \includegraphics[width=\textwidth]{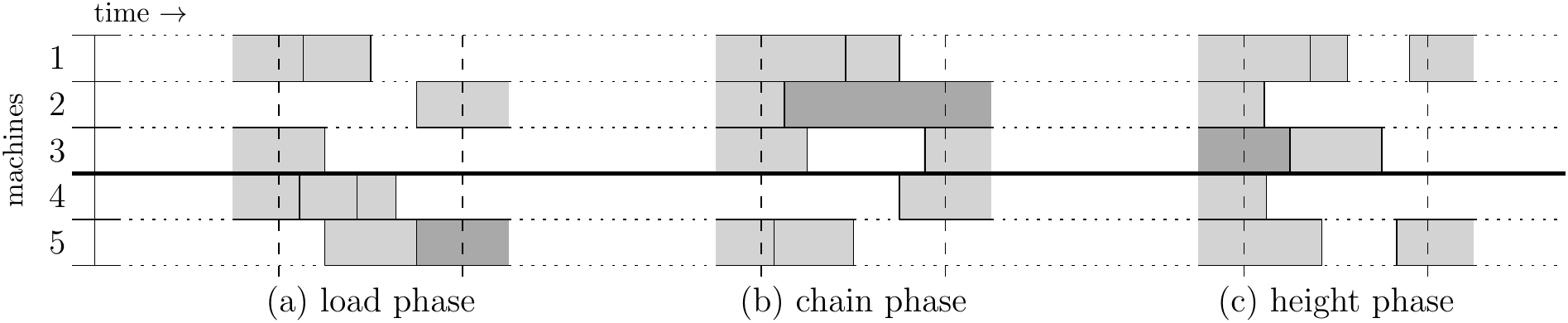}
    \caption{Phase types. Machines are shown vertically on the left and time proceeds from left to right. There are two groups of machines $\group{1} = \{1,2,3\}$ and $\group{2} = \{4,5\}$, separated by a dark line. Chain jobs are shown in dark gray and all other jobs are shown in light gray. For each phase, the time between two dashed lines is $\delay$. (a) A load phase: all machines in group $\{4,5\}$ are busy for at least $\delay/2$ time in the phase. (b) A chain phase: some element in the chain is executing for at least $\delay/2$ time in the phase. (c) A height phase: all groups have some machine that is busy for less than $\delay/2$ time in the phase: machine 2 in $\group{1}$ and machine 4 in $\group{2}$.}
    \label{fig:phase_types}
\end{figure}

\begin{lemma}[{\bf Chain Phase Bound}]
    The number of chain phases is $O(C^*/\delay)$.
    \label{lem:ChainBound}
\end{lemma}

\begin{proof}
  First we observe that for any job $v$, 
  \begin{equation*}
       \sum_i \frac{\jobmachine{v}{i}}{\speed{i}} ~\ge~
       \sum_{i\le \jobtogroup{u}} \frac{\jobmachine{v}{i}}{\speed{i}} ~\ge~
       \sum_{i\le \jobtogroup{u}} \frac{\jobmachine{v}{i}}{2\groupspeed{\jobtogroup{v}}} ~=~
       \frac{1}{2\groupspeed{\jobtogroup{v}}} \sum_{i\le \jobtogroup{u}} \jobmachine{v}{i} ~\ge~
       \frac{1}{4\groupspeed{\jobtogroup{v}}}.
  \end{equation*}
  Recall that $\chain = \{ \link{1}, \ldots, \link{|\chain|} \}$. Let $c_q=(v_q, i_q)$. Let us sum up constraint (\ref{phaseLP_execution}) for $v=v_1, ..., v_{|\chain|-1}$, and, respectively, $u=v_2, ..., v_{|\chain|}$, as well as constraint (\ref{phaseLP_completionlb}) for $v_1$. We get
    \begin{equation*}
       \optmakespan + \sum_{q=1}^{|\chain|-1} \start{v_q} ~\ge~
       \sum_{q=1}^{|\chain|} \start{v_q} + \sum_{q=1}^{|\chain|} \proc{v_q} \sum_i \frac{\jobmachine{v_q}{i}}{\speed{i}} ~\ge~
       \sum_{q=1}^{|\chain|} \start{v_q} + \sum_{q=1}^{|\chain|} \proc{v_q}  \frac{1}{4\groupspeed{\jobtogroup{v}}},
  \end{equation*}
  where the second inequality uses the bound above. This implies that $\optmakespan \ge \start{v_{|\chain|}} + \sum_{v \in \chain} \frac{\proc{v}}{4 \groupspeed{\jobtogroup{v}}}$, and thus $\sum_{v \in \chain} \frac{\proc{v}}{\groupspeed{\jobtogroup{v}}} \le 4\optmakespan$.
  
  \junk{

    Recall that $\chain = \{ \link{1}, \ldots, \link{|\chain|} \}$. We claim that for any $\link{q} \in \chain$ where $\link{q} = (v,i)$ we have \zoya{you are reusing $i$ as an index below}
    \begin{equation}
        \start{v} + \proc{v}\sum_i \jobmachine{v}{i} / \speed{i} \ge \sum_{q' \ge q : \link{q'} = (u,j)} \frac{\proc{u}}{2\groupspeed{\jobtogroup{u}}} \equiv \chi(q). 
        \label{eq:chain_inequality}
    \end{equation}
    By Constraint~(\ref{phaseLP_loadlb}) and Lemma~\ref{lem:Gen:Relaxation}, (\ref{eq:chain_inequality}) entails that $\sum_{v \in \chain} \proc{v}/\groupspeed{\jobtogroup{v}} \le 4\optmakespan$.
    }
    Therefore, in order to prove the lemma,
    all that remains is to show that all chain jobs are scheduled on the groups to which they are assigned.
    \junk{
    We prove (\ref{eq:chain_inequality}) by reverse induction on $q$. The base case follows trivially from Constraint~(\ref{phaseLP_nonnegstart}). We prove the inductive step. Suppose the claim holds for $\{ \link{|\chain|}, \ldots, \link{q} \}$ with $\link{q} = (u,j)$ and $\link{q-1} = (v,i)$ and $q > 1$. By definition of $\chain$, we have that $u \in \predex{v}$. So
    \begin{align*}
        \start{v}  &\ge \start{u} + \proc{u}\sum_i \frac{\jobmachine{u}{i}}{\speed{i}}  &\text{by Constraint (\ref{phaseLP_execution})}
        \\
        &\ge \chi(q)  = \chi(q-1) - \frac{\proc{v}}{2\groupspeed{\jobtogroup{v}}} &\text{by induction}
        \\
        &\ge \chi(q-1) - \proc{v} \sum_{k \le \jobtogroup{v}} \sum_{i \in \group{k}} \frac{\jobmachine{v}{i}}{\groupspeed{\jobtogroup{v}}} \ge \chi(q-1) - \proc{v} \sum_{i} \frac{\jobmachine{v}{i}}{\groupspeed{\jobtogroup{v}}} &\text{by definition of}~ \jobtogroup{v}.
    \end{align*}
    Claim~(\ref{eq:chain_inequality}) follows.
    }
    Specifically, we show that if $\sigma$ schedules job $v$ on machine $i$ and $\proc{v} > 8\delay\speed{i}$, then $i \in \group{\jobtogroup{v}}$.
    Suppose Algorithm~\ref{alg:scheduler} places job $v$ on machine $i \in \group{k}$ and $\proc{v}/\speed{i} > 8\delay$. 
    We fix the step $\step$ to the last execution of line~\ref{line:conditions} prior to placing $v$ on $i$ and show that $v\superstep{} = v$. Suppose otherwise. Then $v \in A\superstep{} \setminus \{v\}$. By condition~\ref{condition:predecessor_size} we have that $8\delay\groupspeed{k} \ge  A\superstep{} \setminus \{v\} \ge \proc{v}$. This, however, contradicts our supposition that $(v,i) \in \chain$. Therefore, $v = v\superstep{}$ which entails, by selection of $v$ at line~\ref{line:assigned_jobs_loop}, that $i \in \group{\jobtogroup{v}}$.
\end{proof}

We define {\em long} jobs be those jobs $v$ such that $\sigma$ schedules some copy of $v$ on any machine $i$ and $\proc{v}/\speed{i} > 8\delay$. All jobs that are not long are {\em short}. 

\begin{lemma}[{\bf Height Phase Bound}]
     The number of height phases is $O(K(\optmakespan + \delay)  \cdot \log_{\overlap} (\delay)/\delay)$.
     \label{lem:heightphase_bound}
\end{lemma}

{\em Proof Idea.}
We provide a high level proof idea. Consider any two consecutive chain jobs, such as $v_4$ and $v_3$ in Figure~\ref{fig:chain_construction}. The goal is to upper bound the number of height phases which occur between these chain jobs, in this case after the completion of $v_4$ on machine 4 and before the start of $v_3$ on machine 2. So we consider such a height phase $\tau$.
Because of the way the algorithm introduces gaps in the schedule, all jobs placed while $T < t$, for some time $t$, will start their execution before the next height phase after time $t$. So we focus on the steps of the algorithm in which $T \le \tau\delay$.
By the fact that all jobs (with the exception of $v_3$) in $V_3$ are short, if all jobs in $V_3$ are placed on any group during these steps, then $v_3$ starts on machine 2 within 8 height phases phases after $\tau$. So we can assume that some job in $V_q$ has not been placed on any group during these steps. 

Consider all such jobs that are in the lowest band according to the LP values (see Definition~\ref{def:band}) and, among those, pick a job $v$ which is assigned to the slowest machine group $\group{k}$.
We know that group $\group{k}$ has a machine with idle time in phase $\tau$, so the only reason $v$ is not placed on some machine in $\group{k}$ during these steps is that it fails to meet one of the three conditions given in line~\ref{line:conditions} of Algorithm~\ref{alg:scheduler}. 
We can infer that $v$ meets condition \ref{condition:predecessor_size} by virtue of being in the lowest band and by Lemma \ref{lem:Gen:HeightBound}. We can also infer that $v$ meets condition \ref{condition:predecessor_groups} by virtue of being assigned to the slowest group. Thus, it must fail condition \ref{condition:overlap}. This means that a $(1-1/\overlap)$ fraction of its incomplete predecessors has already been placed. By definition of $V_q$, all these predecessors must be short. Since these predecessors have all been placed on machines prior to checking $v$, these predecesors will begin before next height phase and have time to complete and communicate $O(1)$ phases after that. By Lemma \ref{lem:Gen:HeightBound} and scaling of speeds and job sizes, $v$ has at most $8\delay$ predecessors in $V_q$. So, in $O(\log_{\overlap}\delay)$ phases, they all complete and $v$ can be scheduled. Repeating this argument for all $K$ groups and all   $O((\optmakespan + \delay)/\delay)$ bands gives the bound in the lemma. We now present the complete proof.

\begin{proof}
    We specify a sequence of phases in $\sigma$ that will structure our argument. We let $\tau_0$ denote the first phase of the schedule. For $1 \le q \le |\chain|$, we let $\tau_q$ denote the first phase in which $v$ is scheduled on $i$, where chain element $\link{q} = (v,i)$. We let $\tau_{|\chain|+1}$ denote the last phase in which some job is started, and $\tau_{|\chain|+2}$ to be the last phase in $\sigma$. We also define the function $\phi$ as follows, where $r_{\max} = \max\{ r : \band{r} \ne \varnothing \}$. 
    \begin{align*}
        \phi(q) = \begin{cases}
            0 &\text{if } q = 0 \\
            r &\text{if } 1 \le q \le |\chain| \text{ and } \link{|\chain| - (q-1)} = (v,i) \text{ and } v \in \band{r} \\
            r_{\max} + d &\text{if } q = |\chain| + 1  + d \text{ with } d \in \{0,1\}.
        \end{cases}
    \end{align*}
    Intuitively, $\phi$ provides us a way to to reason over any pre-chain bands ($\phi(0)$ to $\phi(1)$), those bands of jobs between chain elements ($\phi(q)$ to $\phi(q+1)$), any post-chain bands ($\phi(|\chain|)$ to $\phi(|\chain|+1)$), and the final runtime of any jobs once all jobs have been started ($\phi(|\chain|+1)$ to $\phi(|\chain|+2))$. 
    Finally, we let $H_q$ be the number of height phases up to and including $\tau_q$. We argue that, for any $0 \le q \le |\chain|+2$,
    \begin{equation}
        H_q \le \phi(q) \cdot K \log_{\overlap} \delay. 
        \label{eq:height_phase_bound}
    \end{equation}
    For any job $v \in \band{r_{\max}}$, we have that $\start{v} \ge \delay(r_{\max}-1)/4$ by definition of $\band{r_{\max}}$. So
    $\delay r_{\max} = 4\frac{\delay(r_{\max}-1)}{4} + \delay \le 4\start{v} + \delay \le 4(\optmakespan + \delay)$ by Constraint (\ref{phaseLP_completionlb}). So proving (\ref{eq:height_phase_bound}) is sufficient to prove the lemma.
    
    We prove (\ref{eq:height_phase_bound}) by induction on $q$. The inequality holds trivially for $q=0$, so we suppose it holds for some $0 \le q \le |\chain| + 1$. Since $H_{q+1} = H_q + (H_{q+1} - H_q)$, by induction it suffices  to show that $H_{q+1} - H_q \le (\phi(q+1) - \phi(q)) \cdot K \log_{\overlap} \delay$. 
    
    We first consider the case where $0 \le q \le |\chain|$. Recall that $\tau$ is a height phase if, for every $k$, there is some $i \in \group{k}$ that is idle in $\tau$.
    We consider an arbitrary band $\band{r^*}$ such that $\phi(q-1) < r^* \le \phi(q)$ and an arbitrary group $\group{k^*}$. Let $\tau^*$ be the last phase during which any job in $\band{r^*-1} \cap \fastergroupstojobs{k^*} \cap V_q$ is started. (Recall the sets $V_q$ from Figure \ref{fig:chain_construction}.) Note that, if $q>0$ then $\tau^*$ is after the last chain phase for $\link{q}$.
    For a given phase $\tau$, we define $\step(\tau)$ to be the set of steps $\step$ of Algorithm~\ref{alg:scheduler} for which $T\superstep{} \le \tau\delay$. 
    Let $v$ be some job in $\band{r^*} \cap \fastergroupstojobs{k^*} \cap V_q$ such that no copy of $v$ is placed on any machine in any step $\step \in \step(\tau^*)$.
    Let $\tau_v$ be the phase in which the first copy of $v$ is started. 
    We define the sequence
    $\tau^*_1, \tau^*_{2}, \ldots$ where $\tau^*_0 = \tau^*$ and \begin{align*}
        \tau^*_{d+1} &= \min_{\tau} \{\tau < \tau_v : \exists \tau' > \tau_d^* \text{ such that } \tau \text{ and } \tau' \text{ are height phases and } \tau \ge \tau' + 10\}. 
    \end{align*}
    
    
    We prove that the length of this sequence is $O(\log_{\overlap} \delay)$ by establishing the following claim: for all $d \ge 1$ there is some step $\step$ of line~\ref{line:conditions} such that $v\superstep{} = v$, $i\superstep{} \in \group{k^*}$, $T_{i\superstep{}}\superstep{} \le \delay \tau^*_d$ and $\setproc{A\superstep{}} \le 8\delay\groupspeed{k^*}/\overlap^d$, i.e.\ while scheduling each phase in the sequence, there is a check to schedule $v$ on some machine in $\group{\jobtogroup{v}}$ during which the total size of $v$'s uncompleted predecessor set has been reduced by a factor of $\overlap$.
    The claim entails that for every step $\step$ executing line~\ref{line:conditions} in which $v\superstep{} = v$, $i\superstep{} \in \group{k^*}$, and $T\superstep{} \ge \delay \tau^*_{1 +\log_{\overlap}\delay}$, we have that $\setproc{A\superstep{}} = 0$. 
    Therefore, by Lemma~\ref{lem:predecessor_removal}, $v$ is scheduled before (or during) the height phase immediately after $\tau^*_{1 +\log_{\overlap}\delay}$. Because we have chosen an arbitrary group, this entails that, if $\link{q} = (v^*, i)$, then all predecessors of $v^*$ in $\band{r^*}$ are started within $10 K\log_{\overlap}\delay$ height phases after $\tau^*$. Also, because we have chosen an arbitrary band, all bands $r$ such that $\phi(q) \le r \le \phi(q+1)$, are started in $(\phi(q+1) - \phi(q)) \cdot K \log_{\overlap}\delay$ height phases after the first chain phase of $\link{q+1}$. Therefore, $H_{q+1} - H_q \le 10 \cdot (\phi(q+1) - \phi(q)) \cdot K \log_{\overlap}\delay$. 
    
    We prove the claim by induction on $d$. 
    By definition of $r^*$, all jobs in band $r^*-1$ are placed during steps $\step(\tau^*)$. Note that, since $c_{q+1} = (u,j)$ was the last long predecessor of $v$ to complete, none of $v$'s predecessors placed after $\link{q}$ are long. Since these predecessors have all been placed while $T \le \tau^*\delay$ and since gaps are introduced only when $T$ exceeds some $T_i$ at line~\ref{line:skipahead}, all these predecessors have started before (or during) the next height phase. This entails that all are completed at least $\delay$ time before the start of phase $\tau^*_1$. 
    Since there is some gap on some machine in phase $\tau^*_1$, by Lemmas~\ref{lem:Gen:HeightBound} and \ref{lem:predecessor_removal}, we have that there is some step $\step$ executing line~\ref{line:conditions} in which $v\superstep{} = v$, $i\superstep{} \in \group{k}$, $T_{i\superstep{}}\superstep{} \le \delay\tau^*$, and $\setproc{A\superstep{}} \le 8\delay\groupspeed{\jobtogroup{v}}$. 
    This proves the base case.
    
    We now prove the claim for $d+1$ assuming it holds for $d$.
    Consider phase $\tau^*_d$ and let $\step$ be any step executing line~\ref{line:conditions} in which $v\superstep{} = v$, $i\superstep{} \in \group{k^*}$, $T_{i\superstep{}}\superstep{} \le \delay \tau^*_d$ and $\setproc{A\superstep{}} \le 8\delay\groupspeed{k^*}/\overlap^d$, which exists by induction. By definition of the sequence and our selection of $v$, we know that conditions~\ref{condition:predecessor_size} and \ref{condition:predecessor_groups} are met in step $\step$. Since $v$ was not placed in step $\step$ it must be that condition~\ref{condition:overlap} was not met. This entails that $\setproc{A\superstep{}\setminus\Placed\superstep{}} < \setproc{A\superstep{}}/\overlap$. Since all jobs in $A\superstep{} \cap \Placed\superstep{}$ are short and all are placed while $T \le \delay\tau^*_d$, we have that they all start before the next height phase and complete at least $\delay$ time before the start of phase $\tau^*_{d+1}$. Therefore, by Lemma~\ref{lem:predecessor_removal}, we have that there exists a step $\step'$ in which $v\superstep{'} = v$, $A\superstep{'} \cap A\iter{d} \cap \Placed\iter{d} = \varnothing$, $i\superstep{'} \in \group{k^*}$, and $T_{i\superstep{'}}\superstep{'} \le \delay\tau^*_{d+1}$. This entails that $A\superstep{'} \le \setproc{A\superstep{}} / \overlap$, which proves the inductive step.
    
    Finally, we consider the case where $q = |\chain| + 1$. We show that $H_{q+1} - H_q = O(1)$. 
    By definition, all jobs have been started by the end of phase $\tau_q$. 
    Note that if any of these jobs are chain jobs then the remainder of the phases, with the possible exception of the last, are chain phases, so $H_{q+1} - H_q \le 2$. On the other hand, if no remaining jobs are chain jobs, then all remaining jobs take less than $8\delay$ time on the machines where they've been scheduled, so $H_{q+1} - H_q \le 8$. This completes the proof of inequality (\ref{eq:height_phase_bound}).
\end{proof}

\begin{lemma}[{\bf Load Phase Bound}]
    The number of load phases is $O(K\optmakespan\overlap/\delay)$.
    \label{lem:loadphasebound}
\end{lemma}
\begin{proof}
    For this proof, we assume that $\group{1}, \ldots, \group{K}$ is the set of groups $\group{k}$ for which $\grouptojobs{k} \ne \varnothing$. By definition of $\kappa$, this entails that $\groupsize{k}\groupspeed{k} \ge \groupsize{k+1}\groupspeed{k+1}$. Note that the algorithm does not schedule any jobs on the other groups.
    
    We begin with a useful technical claim. Suppose that we wanted to find values $x_1, ..., x_K \ge 0$ that optimize the following linear program:
    \begin{equation}
        \textrm{maximize~~} \sum_{k} \frac{x_k}{\groupsize{k}\groupspeed{k}} \textrm{~~~~~subject to~~~} \forall k':  \sum_{k\ge k'} x_{k} \le \overlap \cdot \sum_{k\ge k'} \setproc{\grouptojobs{k}}.
        \label{lp:load}
    \end{equation}   
    Note that all $\setproc{\grouptojobs{k}}>0$.
    We claim that setting the variables in such a way that all these constraints are tight, namely $x^*_k = \overlap \cdot \setproc{\grouptojobs{k}}$, is the optimal solution. For contradiction, assume that there is some other solution that gives a higher objective, and let $x'_1, ..., x'_K$ be the lexicographically smallest optimal solution. Note that if $x'_k=0$ for some $k$, then the constraint for $k'=k$ is not tight, since the constraint bounds strictly increase as $k'$ decreases. If $x'_1=0$, then we can increase $x'_1$ and improve the objective. 
    Otherwise, let $k^*$ be such that the constraint for $k^*$ is not tight and $x'_{k^*-1}>0$. Then we can increase $x'_{k^*}$ by a small value $\varepsilon$ and decrease $x'_{k^*-1}$ by the same amount. This is neutral for all constraints except $k^*$, and the objective increases by $\varepsilon / \groupsize{k^*}\groupspeed{k^*} - \varepsilon/\groupsize{k^*-1}\groupspeed{k^*-1} \ge 0$, since $\groupsize{k^*-1}\groupspeed{k^*-1} \ge \groupsize{k^*}\groupspeed{k^*}$. Thus, the solution $x'$ is either not optimal or not lexicographically smallest, which proves the claim.
    
    Let $V(i)$ be the set of jobs scheduled on machine $i$. Now, the values $x_k \equiv \sum_{i\in \group{k}} \setproc{V(i)}$ constitute a feasible solution to (\ref{lp:load}), with Lemma \ref{lem:eta_load_increase} showing that they satisfy the constraints. Thus, their objective value is at most that of $x^*$:
    \begin{equation}
        \sum_k \sum_{i\in \group{k}} \frac{\setproc{V(i)}}{\groupsize{k}\groupspeed{k}} \le
        \sum_k \frac{\overlap \cdot \setproc{\grouptojobs{k}}}{\groupsize{k}\groupspeed{k}}.
        \label{eq:loadineq}
    \end{equation}
    
    Let $L_k$ be the number of load phases in which group $k$ is the slowest group with all machines busy for at least $\delay/2$ time.
    The amount of time that the machines in $\group{k}$ are busy during such phases is at least $\frac{\delay}{2} \cdot \groupsize{k}\cdot L_k$. 
    The total amount of time that the machines in $\group{k}$ are busy is $\sum_{i\in \group{k}} \setproc{V(i)}/ \speed{i}$, so 
    \begin{equation*}
        \frac{\delay}{2} \groupsize{k} L_k ~\le~ 
        \sum_{i\in \group{k}} \frac{\setproc{V(i)}}{ \speed{i}} ~\le~
        \sum_{i\in \group{k}} \frac{\setproc{V(i)}}{\groupspeed{k}}.
    \end{equation*}
    Thus, for all $k$,
    \begin{equation*}
        L_k ~\le~ 
        \frac{2}{\delay} \sum_{i\in \group{k}}\frac{ \setproc{V(i)}}{\groupsize{k}\groupspeed{k}}.
    \end{equation*}
    Summing over $k$, using (\ref{eq:loadineq}), and applying Lemma \ref{lem:LoadBound}, the total number of load phases is
    \begin{equation*}
        \sum_k L_k ~\le~ 
        \frac{2}{\delay} \sum_k \sum_{i\in \group{k}}\frac{ \setproc{V(i)}}{\groupsize{k}\groupspeed{k}} ~\le~
        \frac{2}{\delay} 
        \sum_k \frac{\overlap \cdot \setproc{\grouptojobs{k}}}{\groupsize{k}\groupspeed{k}} ~=~
        O(K\optmakespan\overlap/\delay),
    \end{equation*}
    which proves the lemma.
\end{proof}

\begin{theorem}[{\bf Makespan approximation}]
    There is a polynomial time algorithm that, given an instance of the DAG scheduling with fixed communication delay problem, computes a schedule whose makespan is $O(\log m \log \delay/ \log \log \delay)(\optmakespan + \delay)$, where $\optmakespan$ is the optimal makespan for the given instance.
    \label{thm:makespan}
\end{theorem}

\begin{proof}
    We set $\overlap = \log\delay / \log\log\delay$. By Lemma~\ref{lem:heightphase_bound}, we have the number of height phases is $O(K(\optmakespan+\delay)\log\delay/\log\log\delay)$. By Lemma~\ref{lem:loadphasebound} we have that the number of load phases is $O(\optmakespan K \log\delay/\log\log\delay)$. By Lemma~\ref{lem:ChainBound} we have that the number of chain phases is $O(\optmakespan)$. Since there are no other phases, if $C$ is the makespan of $\sigma$ we have
    \[ C = O(K(\optmakespan+\delay)\log\delay/\log\log\delay) + O(\optmakespan K \log\delay/\log\log\delay) + O(\optmakespan) = O(K\log\delay/\log\log\delay)(\optmakespan+\delay). \]
    Since the number of groups $K \le \log m$, this yields a bound of $O(\log m \log\delay/\log\log\delay)(\optmakespan+\delay)$.
\end{proof}
    

\section{Integrality gap of \texorpdfstring{$\Omega(\sqrt{\log \rho})$}{Lg}}
\label{sec:integrality}
\newcommand{\anote}[1]{}
\newcommand{\eps}{\varepsilon}
\newcommand{\set}[1]{\{#1 \}}
\newcommand{\calP}{\mathcal{P}}

The main result of this section is the development of an instance of the problem for which the integrality gap of LP is $\Omega(\sqrt{\log \rho})$.  

\begin{theorem}[{\bf Integrality gap}]
\label{thm:IG}
There is a family of instances such that for any $\rho$ that is at least some sufficiently large constant, the linear programming relaxation LP has a gap of at least $\Omega(\sqrt{\log \rho})$.
\end{theorem}

\paragraph{Input Instance} The input to the scheduling problem is a layered directed acyclic graph $G$ (see Figure \ref{fig:lb}).  Each vertex in $G$ represents a unit size job.  All machines are of unit speed so that every job can be processed at any machine in unit time.  The instance is parametrized by $d, L\in \N$, where $L$ is the number of levels, and $d$ is the number of immediate predecessors for any job at any of the levels $1$ through $L - 1$.  We set the communication delay $\rho=d^L$.  The set $V_i$ of vertices at level $i$ has size $n = m \rho/L$, where $m$ is the number of machines.    Finally, we specify the precedence constraints.  Between any two levels $V_i, V_{i+1}$ we have a random graph $G_i$ that is bipartite and $d$ left-regular (for each vertex in $V_i$, $d$ random vertices in $V_{i+1}$ are chosen as its neighbors). For our purposes, it will suffice to pick $L=\eps_1 \sqrt{\log n}$, $d=2^{\eps_2\sqrt{\log n}}$ and $m=\rho$, for some appropriate absolute constants $\eps_1, \eps_2 >0$. In particular, for our setting we have $\log \rho = \Theta(\log n)$; so our gap will be $\Omega(\sqrt{\log n})$. 

\begin{figure}
    \centering
	\includegraphics[]{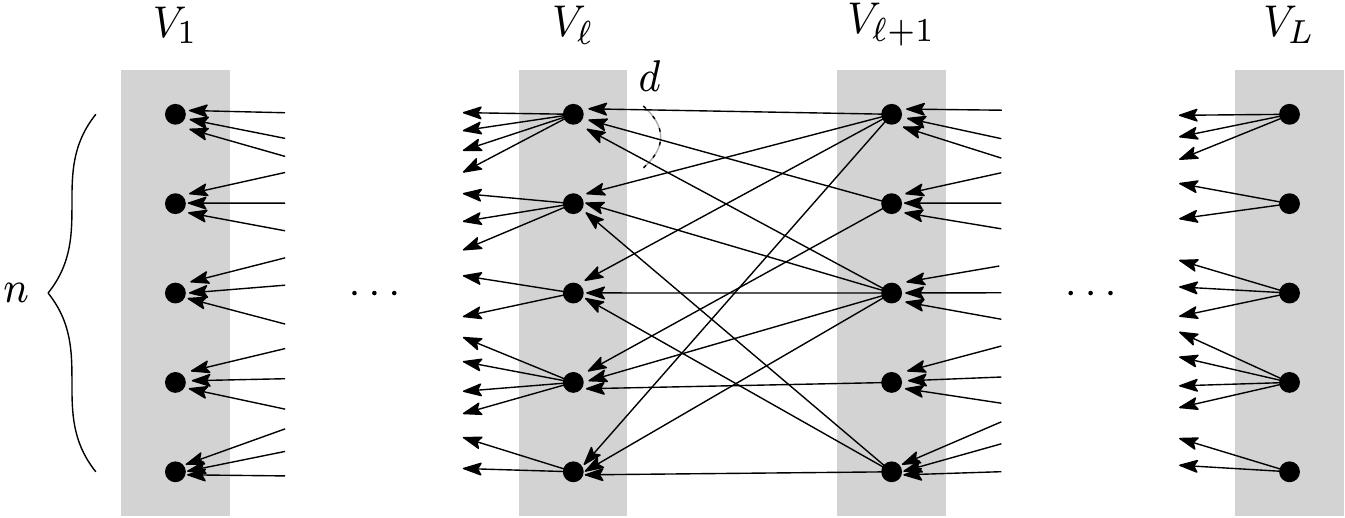}
	\caption{\label{fig:lb} The figure shows the DAG with $L$ layers $V_1, \dots, V_L$ each with $n$ vertices, representing the $nL$ jobs. Each of the $n$ jobs in $V_{\ell}$ has  dependencies on $d$ randomly chosen jobs in $V_{\ell+1}$.} 
\end{figure}

\paragraph{Overview of Challenges and Proof Outline} 
We construct a new integrality gap instance that achieves a $\omega(1)$ integrality gap in the presence of delays. The gap construction consists of a layered DAG (with $L=\Theta(\sqrt{\log n})$ layers), where each layer corresponds to a random graph with (left)-degree $d=2^{\Theta(\sqrt{\log n})}$. The parameters of the construction are set up in such a way that fractionally, all the jobs can be assigned in one phase (hence the LP solution value is $1$). 

The main technical challenge is to argue that $\Omega(L)$ phases are needed in order to schedule all the jobs. From the expansion of the random graph in each layer, it is easy to argue that at most $o(1)$ fraction of the jobs (in layers $\set{1,\dots, L-2}$) can be scheduled in the first phase (since at most $\rho \ll n$ of the jobs can be on one machine). However, in the next phase the jobs that were scheduled previously are now available in all the machines; moreover the choice of these jobs could depend on the randomness in the DAG. Hence the remaining graph in each layer (after removing vertices that have already been scheduled) in the subsequent phases is {\em not random} any longer!

To overcome this technical hurdle, 
we identify and exploit the property of {\em robust expansion}, which may be of independent interest. The vertex expansion property of a random graph says that w.h.p. any subset $S \subset V_\ell$ of size $|S| \le n/d$ has a neighborhood of size $|\Gamma(S)| =\Omega(d |S|)$. However, random graphs have the stronger property that no subset $T$ of size $o(d |S|)$ can have $\Omega(d |S|)$ of the edges from $S$ incident on it. This property of robust expansion along with its generalization to paths of length $\ell < L$ (in Lemma~\ref{lem:generalized}) is crucial in the analysis. 
This property will allow us to prove that we need at least $L/2$ phases before most of the jobs in $V_1$ can be scheduled.

Recall that the relaxation in Section~\ref{sec:makespan} tries to minimize the number of phases, where each phases corresponds to roughly $\rho$ time units (see Lemma~\ref{lem:Gen:Relaxation}). 
In the rest of the section, we will measure the length of the schedule in terms of the number of phases. We first show that the LP has value at most $1$ in Lemma~\ref{lem:integrality_gap_fractional}. Then we prove the necessary robust expansion properties of the instance in Section~\ref{sec:expansion}, and then use this property to lower bound the number of phases in the optimal schedule by $\Omega(L)$ in Proposition~\ref{prop:lb:phases}.  Theorem~\ref{thm:IG} follows directly by just combining Lemma~\ref{lem:integrality_gap_fractional} and Proposition~\ref{prop:lb:phases}. 

\paragraph{Upper bound on the LP solution value.}

We begin by proving that the LP has a fractional solution of value at most $\delay$.  

\begin{lemma}
\label{lem:integrality_gap_fractional}
The linear program LP has a value of at most $\delay$. 
\end{lemma}
\begin{proof}
We distribute each job uniformly across all machines, so $x_{v,i} = 1/m$ for every $v$ and every $i$. Set $z_{u,v,i}= i/m$ for all $u,v,i$. Finally the variables $C_v, S_v$ respect the layered structure of the instance; we set $C_v = (L-\ell+1)\delay/L, S_v=(L-\ell)\delay/L$ for all $v \in V_\ell$, and set $C = \rho$. 
Constraints~\eqref{phaseLP_completionlb} and~\eqref{phaseLP_execution} are immediate, since our assignments of $C_v, S_v$ satisfy the layer structure of the DAG (note all the $p_u, s_i=1$).  
Constraint~\eqref{phaseLP_loadlb} is satisfied with equality since $m= nL/\rho$.  
Constraint~\eqref{phaseLP_delay} is satisfied since $\sum_{j \le i} x_{v,j}=z_{u,v,i}$ for all $u,v,i$; this combined with $|A_v| \le \rho$ also establishes \eqref{phaseLP_phasepredecessors}.  
The other constraints are easily seen to hold, completing the proof of the desired claim about LP. 
\end{proof}

The remainder of this section establishes that there is a setting of the parameters for which an optimal integral schedule has makespan $\Omega(\sqrt{\log \rho})$.  For our purposes, $L=\eps_1 \sqrt{\log n}$, $d=2^{\eps_2\sqrt{\log n}}$ and $m=\rho$, for some appropr absolute constants $\eps_1, \eps_2 >0$.  We begin by establishing certain robust expansion properties of $G$ in Section~\ref{sec:expansion}, and then prove in Section~\ref{sec:integer} the lower bound on the makespan of any integer solution.

    \subsection{Robust expansion properties}
    \label{sec:expansion}
    
In the rest of this section, we will ignore the directionality of the edges, and treat the graphs $\set{G_i}$ between layers as undirected graphs for convenience. 
The following lemma shows that each of the layers $G_i$ has a certain {\em robust expansion} property, that states that for any $S \subset V_i$ and any $T \subseteq \Gamma(S)$ that has a constant fraction of the edges from $S$ incident on it should be of size $\Omega(d |S|)$.  
\begin{lemma}\label{lem:robustexpansion}
For every $i \in [L-1]$, there is an absolute constant $c_1 \le 8$ such that the following holds with probability $1-o(1)$: for every set $S \subset V_i,T \subset V_{i+1}$ with $|S| \le n/d^2$ and $|E(S,T)| \le c_1 |S|+ c_1 |T|$.
\end{lemma}

\begin{proof}
Let $\mathcal{E}$ represent the event that there exists $S,T$ satisfying $|E(S,T)|\ge c_1 |S|$ and $|E(S,T)| \ge c_1 |T|$. 
Fix such a $S,T$. Let $m:=|E(S,T)| \ge c_1 |S|$.
We can also assume without loss of generality that $|T|:=t= m/c_1$; note that $t \le d s \le n/d$. Also fix an index set  $J \subseteq \set{1, \dots, d|S|}$ with $|J|=m$, to represent a subset of edges incident on $S$, and let $E_J$ represent the corresponding edges.
The probability that the edges corresponding to $E_J$ are all incident on $T$ is 
at most $(|T|/n)^m$. 
By a union bound over the choice of the index set $J$, we have
\begin{align*} 
    \Pr[|E(S,T)| \ge m] &= \Pr\Big[ \exists J \subset \set{1, \dots, d|S|} \text{ s.t. } E_J \subset S \times T , |J|=m \Big] \le \binom{d |S|}{m} \Big( \frac{|T|}{n} \Big)^{m}. 
\end{align*}
Now performing a union bound over all $S,T$, 
\begin{align}
    \Pr[\mathcal{E}] & \le \sum_{s \le \frac{n}{d^2}} \sum_{S : |S|=s}\binom{n}{s}  \binom{n}{|T|} \binom{d |S|}{m} \left(\frac{|T|}{n} \right)^{m} \nonumber\\
    &\le \sum_s \exp\Big( s \log(en/s) + t \log(en/t) +m \log(e d s/m) - m \log(n/t)   \Big) \nonumber \\
    &\le \sum_s \exp\left( s \log(en/s) + t \log \Big(en/t\Big)  - m\log\Big( \frac{n m}{etds}\Big)\right). \nonumber\\
    &\le \sum_s \exp\left(-s \Big((\tfrac{m}{2s}-1)\log(\tfrac{en}{s}) - \tfrac{m}{2s}\log(\tfrac{e^2 td}{m}) \Big) -  t \Big((\tfrac{m}{2t}-1)\log(\tfrac{en}{t})-\tfrac{m}{2t}\log(\tfrac{e^2 ds}{m})\Big)\right).\label{eq:exponent}
\end{align}
For the first expression in \eqref{eq:exponent}, recall that $\tfrac{m}{2s} > \tfrac{c_1}{2}  \ge 4$, and $n/s \ge d^2$. Similarly $\tfrac{m}{2t} \ge \tfrac{ c_1}{2} \ge 4$. Hence,
\begin{align*}
    (\tfrac{m}{2s}-1) \log(en/s) &\ge (\tfrac{m}{2s}-1) \log (ed^2) \ge \tfrac{3}{2}\cdot \tfrac{m}{2s} \log(d) \ge \tfrac{3}{2} \cdot \tfrac{m}{2s}\log(e^2 t d/m), \mbox{  and}\\
(\tfrac{m}{2t}-1) \log(en/t) &\ge (\tfrac{m}{2t}-1) \log(\tfrac{en}{\beta m})\ge (\tfrac{m}{2t}-1) \log( \tfrac{e^2 sd^2}{m}) \ge \tfrac{3m}{8t}\log( d\cdot \tfrac{e^2 sd}{m}) \ge \tfrac{3}{2}\cdot  \tfrac{m}{2t}\log(\tfrac{e^2 s d}{m}),
\end{align*}
since $e^2 ds/m \le d$ for our choice of $m \ge c_1 s \ge e^2 s$. Hence, substituting in \eqref{eq:exponent} we have
$$\Pr[\mathcal{E}] \le \sum_s \exp\left(-\tfrac{s}{3} (\tfrac{m}{2s}-1)\log(\tfrac{en}{s})\right) \le \sum_s \exp\left(-s \log(\tfrac{en}{s})\right) = O(\tfrac{1}{n}).$$
This concludes the proof.
\end{proof}

The following lemma generalizes the above lemma to paths of longer length $\ell \ge 1$ as well, and will be crucially used in the analysis. \anote{Note that the parameters below are not optimal. In particular, getting rid of the $exp(\ell-1)$ term in front of $d^{\ell-1} |S|$ term will suffice to go up to $L=\Omega(d)$, I think.}

\begin{lemma}[Robust Expansion for Paths]\label{lem:generalized}
In the above construction $G$, there exists universal constants $c>1$ such that the following holds for any $\ell \le L$ and $i \in \set{1,2,\dots,L-\ell}$ with probability $1-o(1)$. For any $S \subseteq V_i$ s.t. $|S| \le n/(4d^{\ell+1})$ and for any $T \subseteq V_{i+\ell}$, the number of length-$\ell$ paths between $S$ and $T$ is at most $c^{\ell} |T|+ (2cd)^{\ell-1} |S|$.

\end{lemma}
\begin{proof}
We will prove this by induction. The base case when $\ell=1$ follows immediately from Lemma~\ref{lem:robustexpansion}. 

Set $c:=4c_1$. Let us assume that the statement of the lemma holds of all $\ell' < \ell$. We will prove the statement for $\ell$. Let $\calP_\ell$ be the set of length-$\ell$ paths between $S \subseteq V_i$ and $T \subseteq V_{i+\ell}$. Let $T_{\ell-1} \subseteq V_{i+\ell-1}$ be the subset of vertices on which $\calP_\ell$ are incident. 

For $u \in V_{\ell-1}$, let $d_{\ell-1}(u)$ denote the number of length $(\ell-1)$ paths between $u$ and $S$. 
Let $T_{bad}=\set{u \in V_{\ell-1}: d_{\ell-1}(u)> 2c^{\ell-1}}$, and let $T_{good}=T_{\ell-1} \setminus T_{bad}$.
We will count the number of 
length $(\ell-1)$ paths through $T_{bad}$ and $T_{good}$ separately. 

First for $T_{bad}$, we get an upper bound on $|T_{bad}|$ as follows. From the inductive hypothesis, and since the number of length $(\ell-1)$ paths incident on each vertex in $T_{bad}$ is at least $2c^{\ell-1}$, we have
\begin{align*}
    2c^{\ell-1} |T_{bad}| &\le c^{\ell-1} |T_{bad}|  +  (2c d)^{\ell-2} |S|\\
    |T_{bad}| &\le  \frac{(2c)^{\ell-1}}{ c^{\ell-1}} |S| d^{\ell-2} \le  2^{\ell-2} d^{\ell-2} |S|.
\end{align*}

But the total number of length-$\ell$ paths through $T_{bad}$ is at most $d$ times the number of length $(\ell-1)$  paths incident on $T_{bad}$. Hence, the number of length $\ell$ paths through $T_{bad}$ is at most
$$ d \times \Big( c^{\ell-1} |T_{bad}| + (2c)^{\ell-2} |S| d^{\ell-2} \Big) \le \Big( \tfrac{1}{2}(2c)^{\ell-1} + (2c)^{\ell-2} \Big) |S| d^{\ell-1}  \le (2c)^{\ell-1} |S| d^{\ell-1}.$$

On the other hand for $T_{good}$, we first observe that $|T_{good}| \le |S| d^{\ell-1} \le n/(4d^2)$. Hence from Lemma~\ref{lem:robustexpansion}, w.h.p. the number of edges $|E(T_{good}, T)| \le c_1 (|T_{good}|+|T|) \le \tfrac{c}{4} |S| d^{\ell-1}+\tfrac{c}{4} |T|$. The total number of length-$\ell$ paths through $T_{good}$ is at most $\tfrac{1}{2}c^{\ell-1} \times  ( |S| d^{\ell-1}+ |T|)  \le c^{\ell}|T|+c^{\ell} |S| d^{\ell-1}$. Hence, in total the number of $\ell$ paths between $S$ and $T$ is at most $c^{\ell} |T|+ (2c)^\ell |S| d^{\ell-1}$. 
\end{proof}

    \subsection{Bounding the Integer Solution Value}
    \label{sec:integer}
    
For any $i \in [L]$ and $t \le L$, let $n_i(t)$ is the number of jobs in layer $i$ that can be scheduled in the first $t$ phases. The following proposition upper bounds the number of phases required to schedule all the jobs. 

\begin{proposition}\label{prop:lb:phases}
In the above notation and construction, there is an absolute constant $\beta>0$ such that for any layer $i \in [L]$, and any phase $t \le [L/2]$, we have with probability $1-o(1)$, 
\begin{equation}\label{eq:phaselb}
    n_i(t+1) \le n \times \min\Big\{\frac{8L}{(\beta d)^{L-2t-i}}, 1 \Big\} .
\end{equation} 
In particular, with probability $1-o(1)$, we have $n_1(\lfloor \tfrac{L}{2} \rfloor) \le n/2$; hence the number of phases needed in the optimal solution is at least $\lfloor L/2 \rfloor$ with probability $1-o(1)$.  
\end{proposition}

The following lemma is key to upper bound the number of jobs scheduled in the first $t$ phases. 

\begin{lemma}\label{lem:lb:recurrence}
For any $i \in [L]$ and $t \le L$, we have for some absolute constant $c>0$, and any $\ell \in [L]$ with $\ell \ge i$ such that $n_{\ell}(t) \le n/(16 c^{\ell} d)$ that
\begin{equation}\label{eq:lb:recurrence}
    n_i(t+1) \le \frac{4c^{\ell-i}}{d^{\ell-i}}\Big(nL+ n_{\ell}(t)\Big).
\end{equation}
\end{lemma}
We will apply the lemma with $\ell = L-2t$ to obtain our desired bounds.
\anote{The above equation may need to be modified.}

\begin{proof}
The proof of the above lemma will use the robust expansion property for path (Lemma~\ref{lem:generalized}) in two different ways. 
Let $T'_\ell \subset V_{\ell}$ be the subset of jobs in their respective layers that have been completed in the previous phases (hence $|T'_{\ell}|=n_{\ell}(t)$). 

For a job to get scheduled in phase $(t+1)$ on a machine $r \in [m]$, all of its ancestor jobs should have either been completed already, or should be scheduled in the same machine $r$ in phase $t+1$. Let $T^{(r)}$ be the jobs in layer $V_{\ell}$ that will be scheduled on machine $r$ in the current phase $t+1$. Let $S^{(r)} \subset V_i$ represent the subset of jobs in $V_i$ that can be completed in phase $t+1$ on machine $r \in [m]$ in the $(t+1)$th phase.  Hence all of the $|S^{(r)}| d^{\ell}$ paths of length $\ell$ to $S^{(r)}$ have to be incident on the vertices in $T^{(r)} \cup T'_\ell$. We divide such vertices into two cases depending on whether most of its length-$\ell$ paths go from $T'_\ell$ or $T^{(r)}$.  Let $S' \subset V_i$ be the subset of jobs such that at least $\tfrac{1}{2} d^{\ell}$ paths to $S'$ that are incident on $T'_\ell$. Note that $S'$ also includes the jobs from $V_i$ that have been completed in the previous phases. We will first show that $|S'| \le 4 c^\ell n_\ell(t)/d^{\ell}$. Suppose for contradiction that $S'' \subset S'$ with $|S''|=\lfloor \tfrac{4c^\ell n_\ell(t) }{d^{\ell}}\rfloor +1 $. Recall that $|T'_\ell|=n_\ell(t)$, and from assumption $|S''| \le n/(4d^{\ell+1-i})$. Hence we can apply Lemma~\ref{lem:generalized} with $S=S'' \subset V_i$ and $T=T'_\ell \subset V_{\ell}$ (along with the fact that $(2c)^{\ell-1} < d/4$) to conclude that 
\begin{align}
    \frac{1}{2}|S''| d^{\ell} &\le c^{\ell} |T'_{\ell}| + (2c d)^{\ell-1} |S''| \le c^{\ell} |T'_{\ell}| + \tfrac{1}{4}|S''|d^{\ell} ~~~\implies ~~~ 
    |S''| \le  \frac{4 c^{\ell} |T'_\ell|}{d^{\ell}} \nonumber\\
    \text{ Hence, } |S'| &\le \frac{4 c^{\ell} |T'_\ell|}{d^{\ell}} \text{ by contradiction}. \label{eq:recursive:1}
\end{align}


On the other hand, for any job in $S^{(r)} \setminus S'$, at least $\tfrac{1}{2}d^{\ell}$ of the length-$\ell$ paths to it have to be incident on $T^{(r)}$, which has size at most $\rho$. Assume for contradiction that $|S^{(r)}| >  \frac{4 c^{\ell} \rho}{d^{\ell}}$, and let $S \subset S^{(r)}$ of size $|S| =  \lfloor \frac{4 c^{\ell} \rho}{d^{\ell}} \rfloor +1$. Again, from our choice of parameters $|S| \le n/(4d^{\ell})$. Hence, from Lemma~\ref{lem:generalized} applied to set $S$ and $T^{(r)}$
\begin{align}
    \frac{1}{2}|S| d^{\ell} &\le c^{\ell} |T^{(r)}| + (2c d)^{\ell-1} |S| \le c^{\ell} \rho + \tfrac{1}{4}|S| d^{\ell} ~~~\implies ~~~ 
    |S| \le  \frac{4 c^{\ell} \rho}{d^{\ell}} \nonumber \\ 
    \text{ Hence  by contradiction, }    |S^{(r)}| &\le  \frac{4 c^{\ell} \rho}{d^{\ell}} ~\forall r \in [m]. \label{eq:recursive:2}
\end{align}
Combining \eqref{eq:recursive:1} and \eqref{eq:recursive:2} and using $|T'_\ell| = n_{\ell}(t)$, we get
$$n_i(t+1) \le |S'|+\sum_{r \in [m]} |S^{(r)}|  \le  \frac{4c^{\ell}}{d^{\ell}} ( m \rho + n_{i+\ell}(t) ) =\frac{4c^{\ell}}{d^{\ell}} ( n L + n_{i+\ell}(t) ),$$
where the last equality follows from our setting of parameters.
\end{proof}

We now proceed to the proof of Proposition~\ref{prop:lb:phases}. Note that Proposition~\ref{prop:lb:phases} implies that we need at least $L/2$ phases to schedule all the jobs. Hence, the optimum is $\Omega( L)$. 

\begin{proof}[Proof of Proposition~\ref{prop:lb:phases}]
We will prove this by induction using the relation in Lemma~\ref{lem:lb:recurrence}. Set $\beta:=1/c$.

Firstly, $n_\ell(0)=0$ for all $\ell \in [L]$.
For the base case when $t=0$, we have from Lemma~\ref{lem:lb:recurrence} applied with $\ell=L$, that for every $i \in [L]$, $n_i(1) \le   nL/(d/c)^{L-i}$ as required. 

Assume the inductive hypothesis is true for the first $t$ phases. Consider $\ell:=L-2t$. From the inductive hypothesis,
$$n_\ell(t) \le n \cdot \frac{8L}{(\beta d)^{L-2(t-1)-(L-2t)}} \le n \cdot \frac{8L}{(\beta d)^2} \le \frac{n}{16 c^{\ell}d},$$
since for our choice of parameters $c^L < 128 \beta^2 d$. Hence, applying Lemma~\ref{lem:lb:recurrence} with $\ell=L-2t$, we get for every $i \le \ell$
$$n_i(t+1) \le  \Big( \frac{4nL}{(\beta d)^{L-2t-i}} + \frac{n_\ell(t)}{(\beta d)^{L-2t-i}} \Big) \le n \cdot \frac{8L}{(\beta d)^{L-2t-i}}, $$
as required. Hence by induction the proposition follows. 
\end{proof}

\section{Bounding the duplication advantage}
\label{sec:duplication}

In this section, we quantify the advantage that job duplication may offer.  We say that a schedule for a given instance is a {\em no-duplication}\/ schedule if every job in the instance is processed exactly once in the schedule.

\begin{theorem}[{\bf Bounding the duplication advantage}]
There exists an instance with $n/2 = m = 2^\rho$ for which any no-duplication schedule has makespan at least $\rho/\log \rho$ times the optimal makespan.
Given any instance $I$ with $n$ jobs, $m$ machines, communication delay $\rho$, and a schedule $\sigma$ with makespan $M^* \ge \rho$, there exists a polynomial-time computable no-duplication schedule for $I$ with makespan $O(M^* \cdot \log^2 n \log m)$.
\label{thm:duplication_bound}
\end{theorem}
We first present the proof for the case where no job takes more than $\rho$ steps in $\sigma$.  We then show how to extend the theorem to the general case.   

We divide time into phases of length $\rho$ (the communication delay).  Consider the $i$th phase $[i\rho, (i+1)\rho)$ of $\sigma$ for integer $i \ge 0$.  Let $\sigma_i$ denote the schedule $\sigma$ restricted to job executions that begin in the time interval $[i\rho, (i+1)\rho)$.  Let $G_0$ denote the subgraph of $G$ induced by the jobs processed in $\sigma$ during phase 0.  For $i > 0$, let $G_i$ denote the subgraph of $G$ induced by jobs not in $\cup_{\ell < i} G_\ell$ whose first execution in $\sigma$ begins in $[i\rho, (i+1)\rho)$.  For convenience, we use $G_{< i}$ to denote $\cup_{\ell < i} G_\ell$.   

\begin{lemma}
\label{lem:duplication predecessors}
For any job $u$ processed on machine $j$ in $\sigma_i$, every predecessor of $u$ is either in $G_{< i}$ or processed on $j$ in $\sigma_i$.  
\end{lemma}
\begin{proof}
Consider a job $u$ processed on machine $j$ in $\sigma_i$; that is $\sigma[u,j] \in [i\rho, (i+1) \rho)$.  Let $v$ be any predecessor of $u$.  If $v$ is not in $G_{< i}$, then $v$ is first processed at time at least $i \rho$.  Since $u$ is processed at time less than $\rho$ after $i\rho$ and has $v$ as its predecessor, $v$ must be processed in phase $i$ on every machine where $u$ is processed in phase $i$. 
\end{proof}
Our algorithm for transforming an arbitrary schedule $\sigma$ to a no-duplication schedule $\widehat{\sigma}$ consists of transforming each $\sigma_i$, $i \ge 0$, to a  no-duplication schedule $\widehat{\sigma}_i$ that completes $G_i$ in $O(\rho \log^2 n \log m)$ time.  Since each $\sigma_i$ is of length $\rho$, it follows that $\widehat{\sigma}$ is of length $O(M^* \cdot \log^2 n \log m)$.

Fix $i$, and consider schedule $\sigma_i$.  The no-duplication schedule $\widehat{\sigma}_i$ begins with $\rho$ steps allocated for communication delay so that the results of the execution of all jobs in $G_{< i}$ are available to every machine.  So the remainder of the schedule focuses on completing $G_i$.

\begin{lemma}
\label{lem:G_ir.neighbors}
Suppose jobs $u$ and $v$ in $G_i$ share a common predecessor $p$ in $G_i$, and let $m_u$, $m_v$, and $m_p$ denote the number of machines that process $u$, $v$, and $p$, respectively, in $\sigma_i$.
Then, there exist at least $m_u + m_v - m_p$ machines that process both $u$ and $v$ in $\sigma_i$. 
\end{lemma}
\begin{proof}
Let $M_p$ (resp., $M_u$ and $M_v$) denote the set of $m_p$ (resp., $m_u$ and $m_v$) machines processing $p$ (resp., $u$ and $v$) in $\sigma_i$.  Since $M_p \supseteq M_u, M_v$, it follows that at most $m_p - m_u$ (resp., $m_p - m_v$) of the machines in $M_p$ do not process $u$ (resp., $v$).  Thus, at least $m_p -  (m_p - m_u) -  (m_p - m_v) = m_u + m_v - m_p$ machines in $M_p$ process both $u$ and $v$, yielding the desired claim.
\end{proof}

\begin{lemma}
\label{lem:duplication.main}
There exists a polynomial-time computable no-duplication schedule that can complete $G_i$ in $O(\log^2 n \log m)\rho$ steps.
\end{lemma}
\begin{proof}
Recall that $\sigma_i$ is a one-phase schedule for completing $G_i$ (with possible duplication of jobs).  We divide the jobs of $G_i$ into groups based on the number of machines they are replicated on in $\sigma_i$.  For any integer $r \ge 0$, let $G_{ir}$ denote the subset of jobs in $G_i$ with the amount of duplication in $[(1 + \mu)^r, (1 + \mu)^{r+1})$, where $\mu  = 1/(2\log n)$.  Since the maximum amount of duplication is $m$, we obtain that the number of subgroups $r^*$ is at most $\log_{1 + \mu} m = O(\log m \log n)$.

We first argue that for any $r$, any job in $G_{ir}$ has no predecessor in $G_{ir'}$ for $r' < r$.  To see this, note that any job $u$ in $G_{ir}$ is replicated on at least $(1 + \mu)^r$ machines.  Since there is no communication in $\sigma_i$, it follows that any predecessor of $u$ needs to be executed on every machine where $u$ is executed, which implies that that any predecessor of $u$ is in $G_{ir'}$ for some $r' \ge r$.  

Our algorithm consists of computing a no-duplication schedule for $G_{ir}$, in order from $r = r^*$ to $r = 0$.  In the remainder, we show that any $G_{ir}$ can be completed by a no-duplication schedule in $O(\log n)$ phases.  Together with the bound on the number of subgroups, this yields the desired bound.

Fix integer $r \in [0, r^*]$.  We show how to construct a no-duplication schedule that completes at least $1/4$ of the sinks (jobs with no successors) in $G_{ir}$ in one phase.  Repeating this at most $2\log n$ times completes the scheduling of all the sinks of $G_{ir}$, and hence also all of $G_{ir}$ in $2 \log n$ phases (non-sink jobs are scheduled with sink jobs for which they are required).  

We construct an auxiliary undirected graph $H$ over the sinks in $G_{ir}$ as follows: there is an edge between sink job $u$ and sink job $v$ if and only if $u$ and $v$ share a common predecessor in $G_{ir}$.  Using a standard ball-growing technique (or the notion of sparse partitions), we determine a collection $\{S_\ell\}$ of disjoint sets of sinks in $H$ such that (a) in every set $S_\ell$, there exists a sink $s_\ell$ that is within $\log n$ hops of every sink in $S_\ell$, (b) for any distinct $\ell, \ell'$ and any two sinks $s \in S_\ell$ and $s' \in S_{\ell'}$,  $s$ is not adjacent to $s'$; and (c) the total number of sinks in the collection is at least $|H|/2$.  Our  algorithm for obtaining a collection $\{S_\ell\}$ is as follows.  For any undirected graph $K$, vertex $v \in K$, and integer $x \ge 0$, let $B_x(K, v)$ denote the ball of radius $x$ around $v$ in $K$.
\begin{enumerate}
    \item Set $H'$ to $H$ and $\ell$ to 0.
    \item Repeat until $H'$ is empty:
    \begin{enumerate}
        \item Let $s_\ell$ be an arbitrary node in $H'$.
        \item Determine the smallest $x$ such that $|B_{x+1}(H',s_\ell)| \le 2|B_x(H',s_\ell)|$. 
        \item Set  $S_\ell$ to $B_x(H', s_\ell)$ and $H'$ to $H' \setminus B_{x+1}(H', s_\ell)$.
    \end{enumerate}
\end{enumerate}
We now argue the three properties we desire.  For (a), we note that in step~2a, $x \le \log n$ since otherwise $|B_{y+1}(H', s_\ell)| > 2|B_y(H', s_\ell)|$ for $0 \le y  < \log n$, implying that $|B_{\log n}(H', s_\ell)|$ exceeds $n \ge |H'|$, a contradiction.  For (b), we note that once we include a set $S_\ell$, we remove all sinks in $H' \setminus S_\ell$ that are adjacent to a sink in  $S_\ell$, which ensures that any sink in $S_\ell$ is not adjacent to any sink in $S_{\ell'}$ for $\ell' > \ell$, thus establishing (b).  Finally, for (c), we observe that  when $S_\ell$ is included in  the collection, we remove a set of size at most $2|S_\ell|$ from  $H'$, implying that the total number of sinks in the collection $\{S_\ell\}$ is  at least $|H'|/2$, as desired.

Consider any edge $(u, v)$ in $H$.  By Lemma~\ref{lem:G_ir.neighbors}, since $u$ and $v$ share a predecessor in $G_{ir}$, it follows that there exist at least $(1 + \mu)^r(1 - \mu)$ machines that process both $u$ and $v$.  Let $u$  be any job in $S_\ell$.  By a repeated application of the lemma along the shortest path from $s_\ell$ to $u$, we obtain that $s_\ell$ and $u$ are processed on at least $(1 + \mu)^r(1 - \mu)^{\log n} \ge (1 + \mu)^r/2$ machines.  By a standard averaging argument, it follows that there is a machine $j_\ell$ that processes a subset $S'_\ell$ of at least $|S_\ell|/2$ of the jobs in $S_\ell$ in $\sigma_i$.  

The desired no-duplication schedule, which we denote by $\widehat{\sigma}_H$ then consists of processing $S'_\ell$ and all of its predecessors in $G_{ir}$ on machine $j_\ell$, for every $\ell$.  Since no two jobs in $S_\ell$ and $S_{\ell'}$ share  any predecessors, it follows that no job 
is executed on more than one machine, hence ensuring that $\widehat{\sigma}_H$ is indeed a no-duplication schedule.  Furthermore, since the jobs scheduled by $\widehat{\sigma}_H$ on a given machine $j$ is a subset of the jobs scheduled by $\sigma_i$ on $j$, $\widehat{\sigma}_H$ completes in a phase.  Finally, since $|S'_\ell| \ge |S_\ell|/2$ and $|\cup_\ell S_\ell| \ge |H|/2$, it follows that at least $|H|/4$ of the sinks are completed in $\widehat{\sigma}_H$.  We thus have obtained a no-duplication schedule that completes at least $1/4$ of the sinks in $G_{ir}$ in one phase,
thus completing the proof of the lemma. 
\end{proof}

For the special case  where every job completes in $\rho$ steps in $\sigma$, the theorem follows immediately from Lemma~\ref{lem:duplication.main}.  The final no-duplication schedule consists of appending the no-duplication schedules for $G_i$, $i \ge 0$.

For the general case, we extend the above proof by first marking the jobs in $G_i$ that begin in a phase but end at a different phase; there is at most one job on each machine in a given phase.  We then apply the above proof to all the jobs in $G_i$.  In our schedule, any marked job, if executed, would be the last job scheduled on the respective machine.  We  add an additional delay of $\rho$ so that any marked jobs that complete in the following phase in $\sigma$ are completed in the no-duplication schedule.  If a machine works on its marked job for the next $t > 1$ phases in $\sigma$, then we remove the machine from consideration for the next $t$ iterations since it is not executing any jobs in $G_{i+1}$ through $G_{i+t}$. 

\section{Open Problems}
We have presented the first approximation algorithms for scheduling precedence-constrained jobs of non-uniform sizes on related machines with a fixed communication delay, with the objective of minimizing makespan.  Using standard arguments, we can extend our results to the objective of weighted completion times.  Our work leaves several open problems and directions for future research.   Can we improve on the approximation factor achieved for general schedules?  Is there a $\omega(1)$ hardness of approximation for the problem?  We conjecture that the integrality gap of the relaxations is $\Omega(\log \rho/\log\log \rho)$. Improving the current bound and broadening the class of programs is of interest.  Also, there is a gap between the lower and upper bounds for the duplication advantage.  Narrowing this gap, and finding better approximation algorithms for no-duplication schedules would be useful for scenarios where job duplication is not a viable option.

We believe the most significant direction for future research is to study the scheduling problem under more general communication delay environments.  From a practical standpoint, developing algorithms that account for delays in a hierarchical network, which may be modeled for instance by a hierarchically well-separated metric, would be valuable for many datacenter scheduling problems.

\appendix

\section{Scheduling Without Duplication}
\label{sec:no_dup}
\junk{We consider the case where the duplication of jobs in multiple machines is not allowed, in order to arrive at a lower bound on the number of time steps required. Instead of a DAG we consider the special case of a directed binary tree. A lower bound on this naturally extends to the general DAG.}



\begin{lemma}[{\bf Lower bound on duplication advantage}]
\label{lem:duplication}
There is an instance for which the makespan of an optimal no-duplication schedule is at least $ \Omega(\frac{\rho}{\log \rho})$ times that of an optimal schedule.
\end{lemma}
\begin{proof}
Our lower bound instance is a directed complete binary out-tree $G = (V,E)$  whose edges define the precedence order from the root $v_0$ to the leaves. In order to simplify the arguments below, we also include an additional node $v'$ that precedes $v_0$.  So $u \prec v$ if and only if $u$ lies on the path from $v'$ to $v$. $G$ has $2^{\rho-1}$ leaves and $\rho$ levels (not including $v'$). All jobs have unit processing time, i.e. for all $v \in V, \proc{v} = 1$.  We set $m = 2^{\rho-1}$, and assume all machines have unit speed, i.e. for all $i \in M, \speed{i} = 1$.  

If duplication is allowed, then $G$ can be scheduled on $M$ in $\rho+1$ time by executing each path from $v'$ to leaf in topological order on a separate machine. \junk{Additionally, if $\rho$ is the communication delay, the best known approximation algorithm, given in \cite{LR02}, gives a schedule of length $O(\rho \log \rho)$.} We prove that the makespan of any no-duplication schedule of $G$ on $M$ is $\Omega(\rho^2 / \log \rho)$.

Let $\sigma$ be any no-duplication schedule of $G$ on $M$.  We define the $t$th $\rhophase$ to be the time interval $[(t-1)\rho, t\rho)$.  In the following, we will show that the maximum number of jobs that can be completed in $t$th phase $\rhophase$ of $\sigma$ is at most $2\rho^t$. Since the union of all levels $(t-1)\log(\rho) + 1$ to $t \log(\rho)+1$ contains at least $\rho^t$ jobs, this implies that the number of $\rhophase$s in $\sigma$ is at least $\rho / (2\log \rho)$. So we have that $C^* \ge \rho^2/(2\log \rho)$, from which the lemma follows.

In the remainder of this proof, we show by induction on $t$
that the maximum number of jobs that can be completed in $t$th phase $\rhophase$ of $\sigma$ is at most $2\rho^t$.
For the base case, consider $t = 1$. By the structure of the graph, $v'$ must be executed before any other job, and this job can be scheduled on only one machine. Therefore, only one machine can run in the first $\rhophase$, and it can execute at most $\rho$ jobs.
    
    Suppose the claim holds up to some $t \ge 1$. Then the total number of jobs executed in the first $t$ $\rhophase$s is at most $2\rho^t$, by the induction hypothesis. Since there is no duplication, the only jobs in the remaining graph that can be scheduled on different machines are those that share no predacessors. Let $V_t$ be the set of jobs completed in the first $t$ $\rhophase$s. Then the number of machines on which we can schedule jobs in $\rhophase$ $t+1$ is equal to the number of independent trees in the subgraph $G_t$ of $G$ induced on $V \setminus V_t$. We show that this number is no more than $2\rho^{t}$.
    
    We show that there are at most $2\rho^t$ independent trees in $G_t$ using a simple exchange argument. 
    Note that, for any set of jobs $V_t$ completed in the first $t$ $\rhophase$s of $\sigma$, for every $v \in V_t$, $\{u : v_0 \in \predex{u} \text{ and } u \in \predex{v} \} \subseteq V_t$. Let $V[\ell]$ be the union of all sets of jobs in levels $0,\ldots,\ell$ of $G$. Define $V_t = V[t\log\rho]$.
    Then the number of independent subtrees in $G_t$ is at most $2\rho^t$.
    
    We can reach any other configuration by swapping one job at a time: if $v \in V_t$ and $v$ has no descendants in $V_t$ and $u \not\in V_t$ and $\predex{u} \subseteq V_t$, then we replace $v$ with $u$. This swap maintains the property that all jobs in $V_t$ have all their predacessors also in $V_t$. Note, however, that executing a swap can only decrease the number of resulting independent subtrees--if we swap internal node for internal node, the number of subtrees remains the same, but if we swap internal node for leaf, the number of resulting subtrees decreases. Therefore, the maximum number of independent subtrees in $G_t$ is $2\rho^t$.
    
    This entails that the total number of machines that can execute jobs in $\rhophase$ $t+1$ is at most $2\rho^t$. Since each machine can execute at most $\rho$ jobs in a single $\rhophase$, the total number jobs completed after $\rhophase$ $t+1$ is at most $2\rho^{t+1}$.
\end{proof}

We note that the above lower bound is tight for this instance as we can construct an $O(\rho^2 / \log \rho)$ length schedule of $G$ on $M$ as follows. In the first two $\rhophase$s, run one machine for $\rho$ steps to execute all jobs in the first $\log(\rho)$ levels. In the next two $\rhophase$s, we can complete the next $\log (\rho)$ levels by repeating the same process for all $\rho$ subtrees freed up by the previous phases. In general, for each $\rhophase$s $t$ and $t+1$, we can run $\rho^{t-1}$ machines to complete the first $\log(\rho)$ levels of all independent subtrees of $G_t$. This yields a schedule with $O(\rho/\log\rho)$ $\rhophase$s with makespan $O(\rho^2 / \log \rho)$.

\section{Proof of Lemma~\ref{lem:elim_slow}}
\label{sec:elim_slow_machines}

\junk{
We prove Lemma~\ref{lem:elim_slow} in two isolated steps. For a given instance $I = (G,M,\delay)$, we first consider a schedule $\sigma_0$ in which some jobs are scheduled on machines $M_0$ where $M_0 = \{i \in M: \speed{i}<\speed{m}/m\}$. Recall that machine $m$ is the fastest machine in $M$. From $I$, we construct a schedule $\sigma_1$ in which all jobs assigned to some machine in $M_0$ are executed on an added machine $m'$ of speed $\speed{m}$. We show that the makespan of $\sigma_1$ is no more than three times the makespan of $\sigma_0$. We then construct the schedule $\sigma_2$ from $\sigma_1$ in which all jobs execued on machine $m'$ are scheduled on machine $m$. We show that the makespan of $\sigma_2$ is no more than twice the makespan of $\sigma_1$, thereby proving the lemma. 
}

Let $\sigma_0$ be some schedule of $(G,M,\delay)$. We define phase $\tau$ of $\sigma_0$ to be the period of time $[\rho(\tau-1), \rho \tau)$ in $\sigma_0$. We also define $M_0 = \{i \in M: \speed{i}<\speed{m}/m\}$.
We first construct a schedule $\sigma_1$ on $(G,(M \setminus M_0) \cup \{m'\}, \delay)$ where $\speed{m'} = \speed{m}$ (see Figure \ref{fig:elim_slow_machines}). The construction maps all the job start times in a given phase of $\sigma_0$ to a start time in a \textit{step} of the new schedule $\sigma_1$. We define a step of $\sigma_1$ as follows. For any phase $\tau$ of $\sigma_0$, let $U_{\tau}$ be the set of all jobs started on some machine in $M_0$ in phase $\tau$. Then step $\tau$ begins at time 
\[ b(\tau) = \sum_{\tau' < \tau} \Big( \max\{ \rho, \setproc{U_{\tau'}}/\speed{m} \} + \rho \Big)\] 
and ends at time $b(\tau+1) = e(\tau)$.
\begin{figure}
    \centering
    \includegraphics{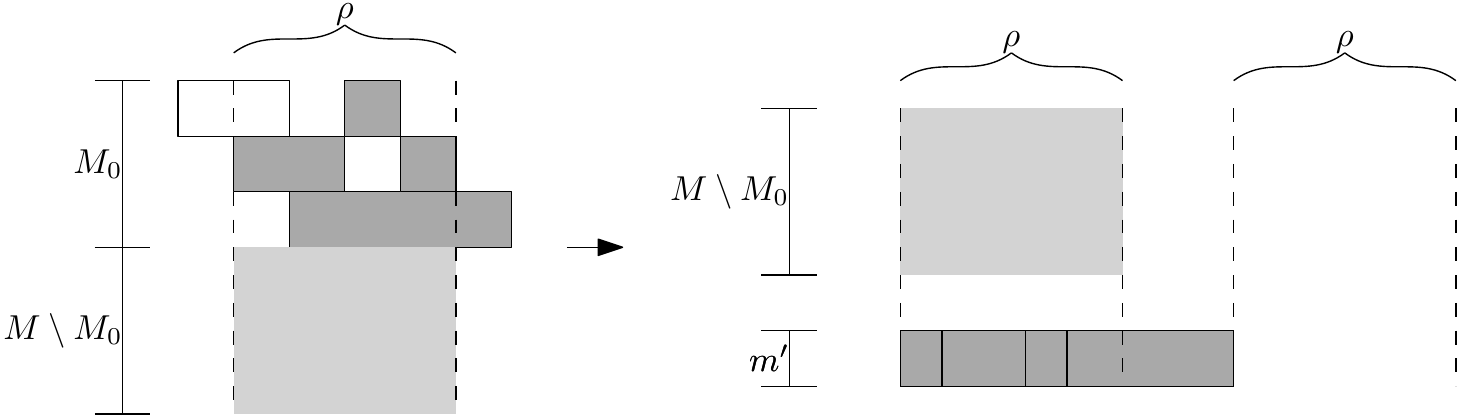}
    \caption{Constructing a step of $\sigma_1$ (right) from a phase of $\sigma_0$ (left). In each, machines are shown vertically on the left and time increases from left to right. Bordered boxes represent jobs. The light gray borderless box represents the jobs being executed on $M \setminus M_0$. All jobs starting in this phase are placed in topological order on the new machine $m'$ and a $\delay$ delay inserted after their completion.  The white job is started in a previous phase and so is ignored.}
    \label{fig:elim_slow_machines}
\end{figure}
We construct an arbitrary step $\tau$ of $\sigma_1$. For each job $v \in U_{\tau}$, schedule $v$ on machine $m'$ at the earliest possible time after $b(\tau)$, maintaining topological order. For each job $v$ that starts on some machine in $M\setminus M_0$ in phase $\tau$, we set 
\[ \sigma_1[v,i] = b(\tau) + \sigma_0[v,i] - (\tau-1)\rho . \]

We now construct a schedule $\sigma_2$ of $(G, M \setminus M_0, \delay)$ (see Figure \ref{fig:elim_extra_machine}). We order all jobs that are executed on machine $m'$ by their start times on $m'$. This gives the ordering $v_1, \ldots, v_L$ where $\sigma_1[v_{\ell},m'] < \sigma_1[v_{\ell+1},m']$. We also define the set $W_{\ell}$ of pairs $(v,i)$ such that $i \ne m'$ and $\sigma_1[v_{\ell},m'] \le \sigma_1[v,i]$ and $\sigma_1[v,i] < \sigma_1[v_{\ell+1},m']$ if $v_{\ell+1}$ exists (and $\sigma_1[v,i] \ne \bot$).
Finally, we define $\mathsf{fin}(v,i,x) = \sigma_x[v,i] + \proc{v}/\speed{i}$ and $u_{\ell} = \arg\max_{v: (v,m) \in W_{\ell}} \{ \sigma_1[v,m]\}$. 
We then construct $\sigma_2$ in $L$ stages, where $L$ is the number of jobs on $m'$. $\sigma_{(2,0)} = \sigma_1$ and
\begin{equation*}
    \sigma_{(2,\ell+1)}[v,i] =
    \begin{cases}
        \sigma_{(2,\ell)}[v,i] &\text{if}~ (v,i) \in \bigcup_{\ell' = 1}^{\ell} W_{\ell'}
        \\
        \mathsf{fin}(u_{\ell},m,(2,\ell))  &\text{if}~ v = v_{\ell+1}
        \\
        \mathsf{fin}(u,m,(2,\ell)) + (\sigma_1[v,m'] - \mathsf{fin}(u_{\ell},m,1)) + &\text{if}~ (v,i) \in \bigcup_{\ell' = \ell+1}^L W_{\ell'}.
        \\
        \quad  \min\{0,\mathsf{fin}(v_{\ell+1},m',1) - \sigma_1[v,m'] \}
    \end{cases}
\end{equation*}
$\sigma_2 = \sigma_{2,L}$.

\begin{figure}
    \centering
    \includegraphics[width=\textwidth]{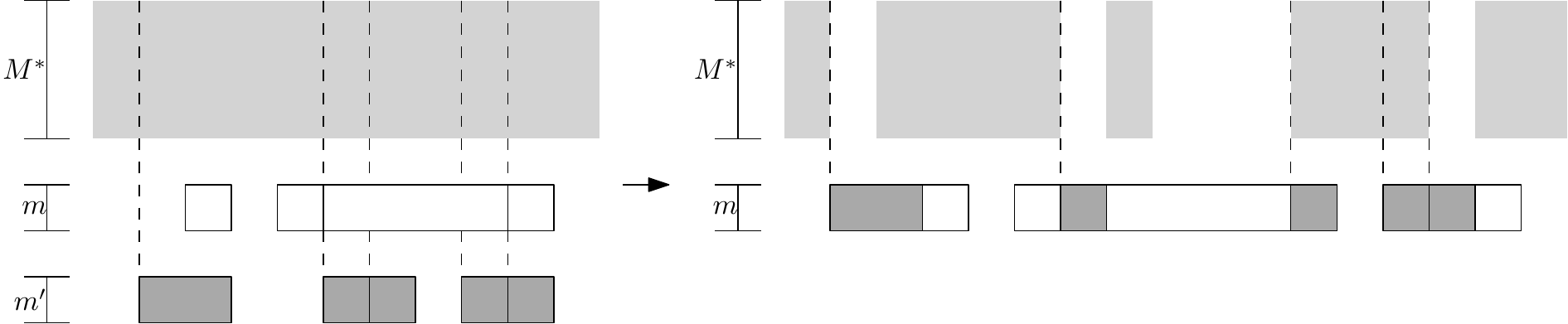}
    \caption{Constructing $\sigma_2$ (right) from $\sigma_1$ (left). In each, machines are shown vertically on the left and time increases from left to right. $M^* = (M \setminus M_0) \setminus \{m\}$. Dashed lines indicate starting times of jobs on $m'$. Bordered boxes represent jobs. Unbordered regions represent starting times of jobs. Note that jobs on $M^*$ may be executing in the white regions of $\sigma_2$. All jobs on $m'$ in $\sigma_1$ are placed immediately prior to later jobs on $m$ in $\sigma_2$, and all starting times are correspondingly delayed.} 
    \label{fig:elim_extra_machine}
\end{figure}

\begin{claim}
    If $\sigma_0$ is a valid schedule of $(G,M,\delay)$, then $\sigma_2$ is a valid schedule of $(G,M\setminus M_0, \delay)$.
\end{claim}

\begin{proof}
    Suppose $\sigma_0$ is a valid schedule. We first show that $\sigma_1$ is valid.
    For any job $u$, let $t_0^c(u)$ and $t_1^c(u)$ be the completion times of $u$ in $\sigma_0$ and $\sigma_1$, respectively, and let $t_0^s(u)$ and $t_1^s(u)$ be the starting times. We show that the communication delay restriction is satisfied with the following claim: if $u$ and $v$ are two jobs such that $u \prec v$, then either $u$ and $v$ are on the same machine in $\sigma$, or $t_1^c(u) - t_1^s(u) \ge \rho$.
    
    Suppose $u$ is completed in phase $\tau_u$ and $v$ is started in phase $\tau_v$. If $\tau_u = \tau_v$ then, since there is no communication within a phase, $u$ and $v$ are executed in order on the same machine in both $\sigma_0$ and $\sigma_1$, by construction. So we assume that $\tau_u <\tau_v$. If both $u$ and $v$ execute on some machine in $M_0$ in $\sigma_0$, then they are on the same machine and in order in $\sigma_1$ by construction. Also, if $u$ is executed on $m'$ in $\sigma_1$, then the result follows immediately since it finishes at least $\rho$ time before the end of its step.
    Also, if $u$ and $v$ both execute on machines in $M\setminus M_0$ then, by construction $t_1^s(v) - t_1^c(u) \ge t_0^s(v) - t_0^c(u)$ so the result follows by the validity of $\sigma_0$.
    So suppose that $u$ executes on some machine in $M \setminus M_0$ and $v$ on machine $m'$. Note that any job's starting phase in $\sigma_0$ is the same as its starting step in $\sigma_1$, and the length of each step is at least twice the length of each phase. So, the ending step of $u$ in $\sigma_1$ is at least as large as its ending phase in $\sigma_0$ and, in its ending step, $u$ completes at least $\rho$ time before the end of the step, by construction. This shows that $\sigma_1$ is valid. 
    
    We now show that $\sigma_2$ is valid if $\sigma_1$ is valid. Consider any two jobs $u$ and $v$ in $G$. By construction of $\sigma_2$, if $v$ starts time $t$ after the completion of $u$ in $\sigma_1$, then $v$ starts at least time $t$ after $u$ in $\sigma_2$. The claim follows by the validity of $\sigma_1$.
\end{proof}

\begin{claim}
    If $\sigma_0$ has makespan $C_0$ and $\sigma_2$ has makespan $C_2$  the $C_2 \le 6C_0$.
    \label{claim:makespan_after_elim}
\end{claim}

\david{simplify proof to get a worse factor approximation}

\begin{proof}
    Let $C_0, C_1$, and $C_2$ be the makespans of $\sigma_0$, $\sigma_1$, and $\sigma_2$ respectively. 
    We first show that $C_1 \le 3C_0$. 
    Let $W = \{v: \sigma_0[v,i]\ne \bot ~\text{for some}~ i \in M_0\}$. The extra communication phase at the end of each step adds $C_0$ to the length of the whole schedule $\sigma_1$. So we have that
    \[ C_1 \le 2C_0 + \setproc{W}/\speed{m} \le 2C_0 +  \sum_{i\in M_0} C_0 \speed{i}/\speed{m} \le 2C_0 + C_0 \sum_{i \in M} 1/m \le 3C_0. \]
    
    We now show that $C_2 \le 2 C_1$.
    Consider difference $\delta_{\ell+1}$ in makespan between $\sigma_{(2,\ell)}$ and $\sigma_{(2,\ell+1)}$. This difference is partially made up of the extra time needed to place $v_{\ell+1}$ on $m$ after the last job on executed on $m$ from a prior $W_{\ell'}$ with $\ell' \le \ell$. The difference is also partially made up of the amount of time needed to delay all start times in $W_{\ell'}$ with $\ell' \ge \ell$. Notice that these two amounts can be charged to the processing time of $v_{\ell+1}$ and to the processing times of those jobs executed in parallel to $v_{\ell+1}$ on machine $m$, and note that all charged regions are disjoint.
    
    \begin{figure}[H]
    \centering
    \includegraphics{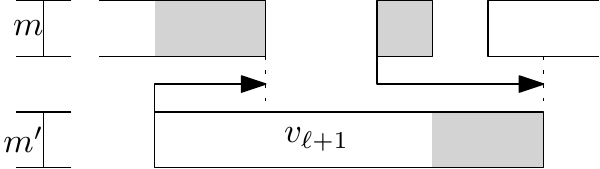}
    \end{figure}
    
    \noindent
    So, $\sum_{\ell=1}^L \delta_{\ell} \le C_1$.  Therefore, by construction of $\sigma_2$, we have that $C_2 = C_1 + \sum_{\ell=1}^L \delta_{\ell} \le 2C_1$.
\end{proof}

\section{Additional Motivation for our Approach}
\label{sec:alternate_LPs}

    \subsection{A Combinatorial Approach}
    \label{sec:combinatorial_approach}
    To motivate our linear program formulation, we show that a natural extension of the combinatorial algorithm in \cite{LR02} performs poorly in the worst case. 
We summarize the original algorithm here, making some simplifications. 

The algorithm given in \cite{LR02} assumes all jobs have unit size and all machines have unit speed. The algorithm constructs the schedule in a series of rounds. In a given round, the algorithm picks the machine $i$ with the least load and finds those jobs that have fewer than $\delay$ predecessors. Let $v$ be one such job. If fewer than $\sfrac{1}{2}$ of $v$'s uncompleted predecessors have already been scheduled on other machines, $v$ and \textit{all} its uncompleted predecessors are placed on $i$. The algorithm then finds the machine with the least load after adding $v$ and its predecessors and continues scheduling the phase. Once no jobs satisfy the above conditions, the algorithm inserts a communication delay and starts the next round. 

We extend this algorithm to the setting with unit job sizes and arbitrary speeds. In this setting, we can load balance using the machine with minimum load as above while looking at the sets of jobs that can be completed in $\delay$ steps for each machine. The extended algorithm is given as Algorithm~\ref{alg:combinatorial_extension}.

\begin{algorithm}
    $T \leftarrow 0;\ \forall i, T_i \leftarrow 0$\;
    \While{there is some unscheduled job}{
        $W \leftarrow \{v :$ some copy of $v$ is scheduled after time $T \}$\;
        $i \leftarrow \argmin_j\{ T_j\}$\;
        $V_i \leftarrow \{v: v$ has fewer than $\delay\speed{i}$ remaining predecessors$\}$\;
        $\forall v \in V_i, U_v \leftarrow (\predex{v} \cup \{v\}) \setminus \{u : \text{some copy of } u \text{ finishes by time } T\}$\;
        \eIf{there is some $v \in V_i$ such that $|U_v \setminus W| \ge |U_v|/2 $}{
            place $U_v$ on $i$ in topological order starting at time $T_i$\;
            $T_i \leftarrow T_i + |U_v|$
        }
        {$T \leftarrow \delay + \max\{$completion time of any currently scheduled job$\}$\;
        $\forall i, T_i \leftarrow \max\{ T_i , T\}$}
    }
    \caption{Extended Combinatorial Algorithm of \cite{LR02}}
    \label{alg:combinatorial_extension}
\end{algorithm}

\begin{lemma}
    For a given instance of precedence constrained scheduling, let $\optmakespan$ be the optimal makespan for that instance. Then there is some instance such that the schedule output by Algorithm~\ref{alg:combinatorial_extension} has makespan upper bound by $O(C^*/\delay) = O(C^* /\sqrt{m}) = O(C^* /n)$.
\end{lemma}

\begin{proof}
    Consider the DAG shown in Figure \ref{fig:gap_constructions}(b) and let $v_{\ell}$ be the job in the chain at level $\ell$ (starting at level 1) and let $u_{\ell,d}$ be the $d$\textsuperscript{th} non-chain successsor of $v_{\ell}$, in any order.  Suppose that $\delay = m^2$ and that we are given $m^2 - 1$ machines of speed 2 and one machine of speed $\delay$. In this case, Algorithm~\ref{alg:combinatorial_extension} may repeatedly place the pair $(v_{\ell}, u_{\ell,1}$ on the speed $\delay$ machine and, for $d = 2,$ place the pairs $(v_{\ell},u_{\ell,d}$ on the $d-1$\textsuperscript{th} speed 1 machine. This would result in a makespan of $\delay^2 = m = \Omega(n)$. However, the optimal makespan is at most $\delay+2$ by placing the entire chain on the speed $\delay$ machine in the first step, then introducing a communication delay, then placing one remaining  job on each machine. 
\end{proof}

We note that the given proof applies to other natural extensions of the combinatorial algorithm, including prioritizing higher speed machines and prioritizing higher capacity groups of machines. 

\begin{figure}
    \centering
    \includegraphics[width=\textwidth]{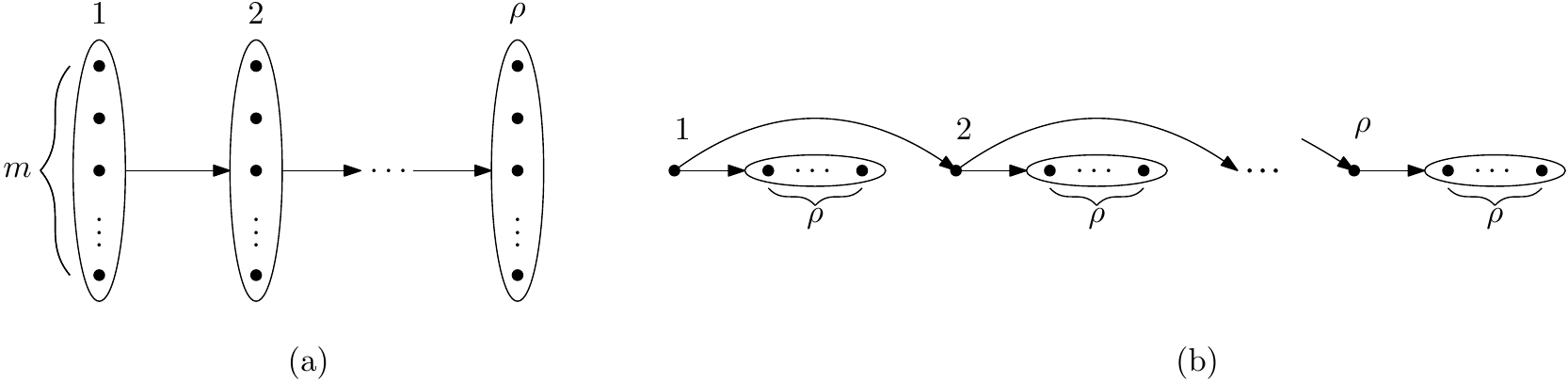}
    \caption{Integrality gap constructions for alternate LP's and the combinatorial heuristic. (a) has $\delay$ layers, each containing $m$ independent jobs. All jobs in layer $\ell$ are predecessors of every job in layer $\ell + 1$. (b) shows a length $\delay$ chain, where each job in the chain precedes a group of $\delay$ independent jobs.}
    \label{fig:gap_constructions}
\end{figure}
    
    \subsection{Same-Machine Variable Relaxation}
    \label{sec:same_machine_variable}
    \newcommand{\reviewerLP}{\ref{revLP_completionlb} - \ref{revLP_samemachine}}

We refer to the following linear program as \LP{\reviewerLP}, which minimizes $C$ subject to the given constraints. 

\noindent
\begin{minipage}{.6\linewidth}
  \begin{flalign}
        \LPmakespan &\ge \start{v}+1 &\forall v
        \label{revLP_completionlb}
        \\
        \LPmakespan\speed{i} &\ge \sum_v  \jobmachine{v}{i} &\forall i
        \label{revLP_loadlb}
        \\
        \start{v}  &\ge \start{u} + \delay \ \Big(1- \sum_{i \in M} \delta_{i,u,v} \Big) & \forall u \prec v
        \label{revLP_delay}
       \\
       \start{v} &\ge \start{u} + 1 &\forall u \prec v
        \label{revLP_execution}
        \\
        \start{v} &\ge 0 &\forall v
        \label{revLP_nonnegstart}
  \end{flalign}
\end{minipage}\hfill
\begin{minipage}{.34\linewidth}
  \begin{flalign}
        \sum_i  \jobmachine{v}{i} &= 1 &\forall v
        \label{revLP_finishjob}
        \\
        \jobmachine{v}{i} &\in [0,1] &\forall v,i
        \label{revLP_jobmachine}
        \\
        \delta_{i,u,v} &\le \jobmachine{v}{i} &\forall v
        \\
        \delta_{i,u,v} &\le \jobmachine{u}{i} &\forall u
        \\
        \delta_{i,u,v} &\in [0,1] &\forall u,v,i
        \label{revLP_samemachine}
  \end{flalign}
\end{minipage}
\vspace{4mm}

Consider the instance given in Figure~\ref{fig:gap_constructions}(a).  The instance is divided into $\delay$ \textit{levels}, each of which contains $m$ independent jobs. All jobs at level $\ell$ precede those jobs at level $\ell+1$. There are $n = m^2$ jobs, $m$ machines, and communication delay of $\delay = m$. All jobs have unit sizes and all machines have unit speed. We show that the optimal duplication schedule has makespan $\Omega(\delay^2)$. Note that it is trivial to construct a $\delay^2 + \delay$ schedule by scheduling all jobs in level $\ell$ to start at time $\delay\ell + (\ell-1)$, and so complete by time $\delay\ell + \ell$. We prove the lower bound by induction on $\ell$ for the following claim, which straightforwardly entails the lower bound.
    
\begin{lemma}
    In any schedule of the instance given in Figure~\ref{fig:gap_constructions}(a) with $m$ unit speed machines and communication delay $\delay$, no job in level $\ell$ can be scheduled before time $\delay \ell$. 
    \label{lem:gapinstance_optlb}
\end{lemma}

\begin{proof}
    For $\ell = 0$ the claim is trivially true, so we suppose the claim is true for all levels up to $\ell \ge 0$. In this case, all jobs in level $\ell$ are scheduled after time $\delay\ell$. Consider any job $v$ in level $\ell+1$. Since $m > \delay$, we can schedule $v$ earliest by scheduling all jobs in level $\ell$ on different machines at time $\delay\ell$ and scheduling $v$ on some machine at time $\delay\ell + 1 + \delay$. Thus, the claim is proved. 
\end{proof}

\begin{lemma}
    For a given instance of precedence constrained scheduling with fixed communication delay, let $C^*$ be the optimal makespan for that instance. Then there is some instance for which the optmal value of \LP{\reviewerLP} is upper bounded by $O(C^*/\delay) = O(C^*/m) = O(C^*/\sqrt{n})$.
    \label{lem:same_machine_gap}
\end{lemma}

\begin{proof}
    For each $v,i$, we set $x_{v,i}$ to $1/m$. For each $v$ in level $\ell$, we set set $S_v = \ell$. For all $u,v,i$ we set $\delta_{u,v,i} = 1/m$. Finally, we set $C = \delay$. It is easy to see that constraints (\ref{revLP_completionlb}), (\ref{revLP_loadlb}), (\ref{revLP_execution}) - (\ref{revLP_samemachine}) are satisfied, so we focus on (\ref{revLP_delay}). Consider any two jobs $u,v$ such that $u\prec v$. In this case, we have $\sum_i \delta_{i,u,v} = 1$, so (\ref{revLP_delay}) is satisfied.
\end{proof}

    \subsection{A Time-Indexed Relaxation}
    \label{sec:time_indexed}
    \newcommand{\KulkarniLP}{\ref{KulkarniLP_start}-\ref{KulkarniLP_end}}
\newcommand{\timeindexedvar}[1]{x_{#1}}

We refer to the following time-indexed program as \LP{\KulkarniLP}. Here, we use notation $[a]$ to refer to the set of all positive integers less than or equal to $a$.




\begin{minipage}{.55\linewidth}
  \begin{flalign}
        &\sum_{i,t} \timeindexedvar{v,i,t} = 1 & \forall v 
        \label{KulkarniLP_start}
        \\
        &\timeindexedvar{v,i,t+1} + \sum_{i' \in [m] \setminus \{i\}} \sum_{t'=t-\delay}^t \timeindexedvar{u,i',t} \le 1 &\forall u \prec v, i,t
        \label{KulkarniLP_delay}
        \\
        &\sum_{i, t' \le t+1} \timeindexedvar{v,i,t'} \le \sum_{i,t' \le t} \timeindexedvar{u,i,t'} &\forall u \prec v, t
        \label{KulkarniLP_precedence}
  \end{flalign}
\end{minipage}\hfill
\begin{minipage}{.3\linewidth}
  \begin{flalign}
        &\sum_{v} \timeindexedvar{v,i,t} \le 1 & \forall i,t
        \label{KulkarniLP_capacity}
        \\
        &\timeindexedvar{v,i,t} \in [0,1] & \forall v, i,t
        \label{KulkarniLP_end}
  \end{flalign}
\end{minipage}
\vspace{4mm}



\begin{lemma}
    For a given instance of precedence constrained scheduling with fixed communciation delay, let $C^*$ be the optimal makespan for that instance. Then there is an instance for which the optimal value of \LP{\KulkarniLP} is upper bound by $O(C^*/\delay) = O(C^*/m) = O(C^*/\sqrt{n})$.
\end{lemma}

\begin{proof}
    We consider the same instance as Lemma~\ref{lem:same_machine_gap} for which the optimal schedule has makespan $\Omega(\delay^2)$ by Lemma~\ref{lem:gapinstance_optlb}. 
    We assign the variables of \LP{\KulkarniLP} as follows.
    For each $\ell$ and $v \in \ell$ and $i$, we assign $\timeindexedvar{v,i,\ell} = 1/m$ and $\timeindexedvar{v,i,t} = 0$ for $t \ne \ell$. 
    Note that, since every level is completed in a single step, the total number of steps with any nonzero variable is $\delay$.
    
    We show that this assignment satisfies all constraints of \LP{\KulkarniLP}. Constraint \ref{KulkarniLP_end} is trivially satisfied.
    Since every job has a $1/m$ fraction assigned to $m$ machines, constraint \ref{KulkarniLP_start} is satisfied. Since there are $m$ jobs with a $1/m$ fraction assigned to each machine at each step $t$, constraint \ref{KulkarniLP_capacity} is satisfied. There is no constraint \ref{KulkarniLP_precedence} for those jobs in the first level, so we show that \ref{KulkarniLP_precedence} is satisfied for those jobs in levels $\ell \ge 1$. Let $v$ be a job in level $\ell \ge 1$ and let $u$ be any predecessor of $v$. By construction, $u$ is in some level $\ell' < \ell$. Let $v$ be fractionally scheduled to start in step $t$. By construction, all predecessors are completed by time $t$. So we have that \ref{KulkarniLP_precedence} is satisfied for all $t' < t$. Consider $t' \ge t$. In this case, both the right and left sums equal 1, so the constraint is satisfied. 
    
    All that remains is to show that constraint \ref{KulkarniLP_delay} is satisfied. Again, we consider a job $v$ in level $\ell \ge 1$ and some predecessor $u$ of $v$ in level $\ell' < \ell$. We fix the machine $i$ and suppose that $v$ is scheduled in step $t^*$. By construction, we have that $\sum_{i' \in [m] \setminus \{i\}} \sum_{t'=t-\delay}^t \timeindexedvar{u,i',t}$ equals 0 or $(1-1/m)$, depending on $t$. Similarly, we have that $\timeindexedvar{v,i,t}$ equals 0 or $1/m$ depending on $t$.
    For a given time $t < t^* -1$, the term $\timeindexedvar{v,i,t+1} = 0$, so the constraint is trivially satisfied. So we assume $t \ge t^*-1$. In this case, $\timeindexedvar{v,i,t+1} = 1/m$. 
    Let $\hat{t}$ be the step when $u$ is scheduled.
    By construction, $\hat{t} \ge t^* - \delay$. So, $\sum_{i' \in [m] \setminus \{i\}} \sum_{t'=t-\delay}^t \timeindexedvar{u,i',t} = 1 - 1/m$, so the constraint is satisfied (and tight).
    
    Thus, the makespan of the fractional schedule resulting from the variable assignments is $O(\delay) = O(m)$ by definition of $m$, and the best duplication schedule has makespan $\Omega(\delay^2) = \Omega(m^2)$. The lemma follows.
\end{proof}


Note that, since duplicated schedules include non-duplicated schedules, the above lemma applies to non-duplicated schedules as well. Finally, note that, since our algorithm finds an approximation within $\log \rho$ of the optimal solution to our LP, this implies that there is a gap of $\Omega(\delay/\log \delay)$ between \LP{\KulkarniLP} and our LP.
    
    \subsection{A Same-Phase Variable Relaxation}
    \label{sec:simple_same_phase_extension}
    \newcommand{\simpleLP}{\ref{simple_completionlb} - \ref{simple_samephase}}

The linear program \LP{\simpleLP} minimizes $C$ subject to the following constraints.

\noindent
\begin{minipage}{.6\linewidth}
  \begin{flalign}
        \LPmakespan &\ge \start{v} &\forall v
        \label{simple_completionlb}
        \\
        \LPmakespan\speed{i} &\ge \sum_v \proc{v} \jobmachine{v}{i} &\forall i
        \label{simple_loadlb}
        \\
        \start{v} &\ge \start{u} + \delay(1-\phasevariable{u}{v}) &\forall u \prec v
        \label{simple_delay}
       \\
       \delay \sum_{i}  \speed{i} \jobmachine{v}{i}  &\ge \sum_{u \prec v} \proc{u} \phasevariable{u}{v}   &\forall v
       \label{simple_phase}
  \end{flalign}
\end{minipage}\hfill
\begin{minipage}{.34\linewidth}
  \begin{flalign}
        \start{v} &\ge \start{u} &\forall u \prec v
        \label{simple_precedence}
        \\
        \start{v} &\ge 0 &\forall v
        \label{simple_nonnegstart}
        \\
        \sum_i  \jobmachine{v}{i} &= 1 &\forall v
        \label{simple_finishjob}
        \\
        \jobmachine{v}{i} &\in [0,1] &\forall v,i
        \label{simple_jobmachine}
        \\
        \phasevariable{u}{v} &\in [0,1] &\forall u,v
        \label{simple_samephase}
  \end{flalign}
\end{minipage}
\vspace{4mm}

\begin{lemma}
    For a given instance of precedence constrained scheduling with fixed communication delay, let $C^*$ be the optimal makespan for that instance. Then there is an instance for which the optimal solution to \LP{\simpleLP} is upper bound by $O(C^*/m^{1/12}) = O(C^*/\delay^{1/6}) = O(C^*/n^{1/18})$.
\end{lemma}

\begin{proof}
    We consider graph with $\sqrt{m}$ copies of a graph similar to the construction in Figure~\ref{fig:gap_constructions}(a), except these instances we have $\sqrt{m}$ levels with $\sqrt{m}$ jobs per level. For machines, we have $m$ machines of speed 1 and one machine of speed $m^{1/3}$. The communication delay $\delay = \sqrt{m}$ and the number of jobs $n = m^{3/2}$.
    
    A similar argument to the one used in the proof of Lemma~\ref{lem:gapinstance_optlb} entails that any schedule of this graph has makespan $\Omega(\delay\sqrt{m})$. We now prove an upper bound on the optimal value of \LP{\simpleLP}.  We assign values to each variable as follows. Let $i^*$ be the speed $m^{1/3}$ machine. For each $v,i$, if $i = i^*$ then we set $\jobmachine{v}{i} = 1/m^{1/4}$ and we distribute the remaining $\jobmachine{v}{i} = 1-1/m^{1/4}$ evenly over all remaining machines. We partition the levels into $\sqrt{m}/m^{1/12}$ classes where class $c_{r} = \{v \in$ level $\ell$ for $(r-1)m^{1/12} \le \ell < rm^{1/12}\}$. For each $v \in c_r$, we set $S_v = \delay r$ and $C_v = S_v + 1$. We set $\phasevariable{u}{v} = 1$ if, for some $r$, $u \in c_r$ and $v \in c_r$ and set $\phasevariable{u}{v} = 0$ otherwise. Finally, we set $C = \delay \sqrt{m}/m^{1/12}$.  
    
    We show that all constraints of \LP{\simpleLP} are satisfied. It is straightforward to check that constraints  (\ref{simple_precedence}) - (\ref{simple_samephase}) are satisfied, so we focus on constraints (\ref{simple_completionlb}), - (\ref{simple_phase}). We note that $\max_v \{\start{v}\} = \delay(\sqrt{m}/m^{1/12})$, so (\ref{simple_completionlb}) is satisfied. 
    For $i^*$, we have that $C\speed{i^*} = m^{17/12} > m^{15/12} = \sum_v \proc{v}\jobmachine{v}{i^*}$. For $i \ne i^*$, (\ref{simple_loadlb}) follows easily, so the constraint is satisfied.
    (\ref{simple_delay}) follows from the fact that, if $\phasevariable{u}{v} = 1$ then $\start{v} = \start{u}$ and, if $\phasevariable{u}{v} = 0$ and $u \prec v$ then $u$ is in a lower class than $v$ so $\start{v} \ge \start{u} + \delay$. 
    We now show that (\ref{simple_phase}) is satisfied. For any $v$,
    \begin{align*}
        \delay\sum_i \speed{i} \phasevariable{v}{i} &= \delay\sum_{i \ne i^*} \phasevariable{v}{i} + \delay \phasevariable{v}{i^*}m^{1/3} &\text{by construction}
        \\
        &= \delay(1-1/m^{1/4}) + \delay m^{1/3}/m^{1/4} &\text{by assignment}
        \\
        &= m^{1/2} - m^{1/4} + m^{7/12} &\text{by definition of }\delay.
    \end{align*}
    Since any $v$ has at most $\sqrt{m} \cdot m^{1/12} = m^{7/12}$ predecessors $u$ such that $\phasevariable{u}{v} = 1$, this shows the constraints are satisfied. 
    
    By assignment of $C$, we have that $C = \delay\sqrt{m}/m^{1/12} = m^{11/12}$. Since the optimal makespan is at least $m$ this proves the lemma. 
\end{proof}

\bibliographystyle{plain}
\bibliography{bibliography.bib}
\end{document}